\documentclass[twocolumn]{article}
\usepackage[fontsize=9.5pt]{fontsize}

\usepackage{cite} 		  
\usepackage{xcolor}       
\usepackage{hyperref}	  
\hypersetup{colorlinks=true,linkcolor=blue,citecolor=blue,breaklinks={true}} 
\usepackage[normalem]{ulem}

\usepackage{graphicx}
\usepackage{sidecap}
\usepackage{commath}
\usepackage{float}
\usepackage[font=footnotesize,labelfont=bf,justification=justified]{caption}
\graphicspath{{./FIG/}}
\DeclareGraphicsExtensions{.eps,.png,.jpg,.jpeg,.bmp,.gif,.pdf}

\usepackage{tabularx,booktabs}	
\usepackage{multirow}	        
\usepackage{array}
\usepackage{makecell}

\usepackage{amsmath}      		
\usepackage{amssymb}	  		
\usepackage{bm}		      		
\usepackage{physics}	  		
\usepackage{mathtools}			
\usepackage{dsfont}		  		
\usepackage{soul}				
\usepackage{relsize}			

\newcommand{\R}{\mathbb{R}}		
\newcommand{\C}{\mathbb{C}}		
\newcommand{\transp}{\mathsf{T}}					

\DeclareMathOperator*{\argmin}{arg\,min}
\newcolumntype{C}[1]{>{\centering\arraybackslash}m{#1}}

\usepackage{amsthm}
\newtheorem{defin}{Definition}
\newtheorem{theor}{Theorem}
\newtheorem{corol}{Corollary}
\newtheorem{prop}{Proposition}
\newtheorem{lemma}{Lemma}
\newtheorem{conjec}{Conjecture}
\theoremstyle{definition}
\newtheorem{assump}{Assumption}
\newtheorem{remark}{Remark}
\newtheorem{example}{Example}

\newcommand*{\QEDA}{\hfill\ensuremath{\triangle}}   %

\usepackage[left=44pt,%
                right=44pt,%
                top=50pt,
                bottom=50pt,%
                headheight=0pt,
                headsep=10pt,%
                footskip=25pt,
                marginparwidth=38pt]{geometry}

\usepackage[switch]{lineno}       

\usepackage[resetlabels]{multibib}
\usepackage{authblk}              
 
\usepackage{titlesec}             
\titleformat{\section}{\large\bfseries}{\thesection.}{1em}{}
\titleformat{\subsection}{\normalsize\bfseries}{\thesubsection.}{1em}{}
\titlespacing*{\section}
{0pt}{1em}{0.2em} 
\titlespacing*{\subsection}
{0pt}{1em}{0em} 
\usepackage{changepage}

\usepackage{tocloft}

\title{{\huge\textbf{Disorder-promoted stability}}}

\author[1,2$*\dagger$]{Arthur~N.~Montanari}
\author[1,2$\dagger$]{Pietro Zanin}
\author[1,2,3,4]{Adilson E. Motter}

\affil[1]{Center for Network Dynamics, Northwestern University, Evanston, IL 60208, USA}
\affil[2]{Department of Physics and Astronomy, Northwestern University, Evanston, IL 60208, USA}
\affil[3]{Northwestern Institute on Complex Systems, Northwestern University, Evanston, IL 60208, USA}
\affil[4]{Department of Engineering Sciences and Applied Mathematics, Northwestern University, Evanston, IL, 60208, USA}

\date{}

\begin{document}

\twocolumn[
    \maketitle
    \vspace{-4.5em} 
    \begin{center}
        \begin{minipage}{1\textwidth}
            \begin{abstract}
            \vspace{2.5em}
\noindent
Previous studies of network dynamics have suggested that heterogeneity among nodes inhibits stability, at odds with the ubiquity of inherently heterogeneous natural and engineered systems. Here, we show that this conclusion arises from model reductions introduced for mathematical tractability and breaks down when nodal dynamics are higher-dimensional, yielding non-Hermitian Jacobians. In such systems, including neural, power-grid, and material networks, nodal heterogeneity can instead enhance stability, even when parameters are randomly disordered. Non-Hermiticity also underlies the stabilizing effects of network heterogeneity, which can arise even in one-dimensional nodal dynamics through nonreciprocal interactions, as shown for ecological networks. Our framework reveals disorder not as a liability but as a general resource for stabilizing complex systems.
\end{abstract}
    \end{minipage}
\end{center}
\vspace{1em} 
]

\let\thefootnote\relax\footnotetext{$^*$Corresponding author: arthur.montanari@northwestern.edu.}

\let\thefootnote\relax\footnotetext{$^\dagger$These authors contributed equally to this work.}

\noindent
Stable dynamical states are essential for the function of complex systems \cite{may1972will,deco2011emerging,meena2023emergent}. Examples include synchronization in power grids \cite{Motter2013,Dorfler2013} and laser arrays \cite{soriano2013complex}, coordinated motion in animal flocks and drone swarms \cite{katz2011inferring,olfati2006flocking,ren2008distributed}, distributed computation in neural networks \cite{fair2009functional,slotine2003modular}, and species coexistence in ecological communities \cite{allesina2012stability,chen2024stability}. In these systems, network-mediated interactions give rise to robust states that persist despite the pervasiveness of perturbations. 
The large scale of such complex systems makes their collective dynamics notoriously difficult to analyze mathematically or computationally. Previous work has often turned to simplified models that reduce the dimensionality of the system while keeping the essence of its emergent behavior \cite{ott2008low,nakao2016phase,kuehn2021universal,nijholt2022emergent,thibeault2024low}.
Yet, as shown here, reduced models may misrepresent the impact of parameter choices and network structures on the stability of dynamical states.

A prominent example of dimensionality reduction at the node level is the Kuramoto model, which captures the emergence of synchronization in large networks of coupled oscillators \cite{Kuramoto1975,Rodrigues2016}. This model has led to the fundamental understanding that synchronization is facilitated when oscillators are near-identical \cite{Kuramoto1975,Dorfler2013}. Consequently, many system-level approaches, such as the master stability function (MSF) used to decouple network modes \cite{fujisaka1983stability,Pecora1998,nishikawa2006synchronization,pecora2014cluster,carletti2023global}, have been derived on the premise of oscillator homogeneity. 
In apparent contradiction, a different body of literature has recently shown theoretically and experimentally that oscillator heterogeneity can promote, rather than inhibit, synchronization stability across a wide variety of network systems. These systems include power networks \cite{molnar2020network,molnar2021asymmetry},  laser arrays \cite{nair2021using,barioni2025,ye2025disorder}, optical systems \cite{perego2022synchronization,cao2022harnessing}, electronic circuits \cite{mallada2015distributed,sugitani2021synchronizing}, superconducting devices \cite{braiman1995disorder}, coupled chemical oscillators \cite{zhang2021random}, quantum systems \cite{lorch2017quantum}, multi-agent systems \cite{yang2022emergent,montanari2025optimal}, and neuronal circuits \cite{gast2024neural}, among others. Certain systems have even been shown to synchronize exclusively when the oscillators are nonidentical \cite{nishikawa2016symmetric,perego2022synchronization,bolotov2025heterogeneity}.

A principled approach for understanding the interplay between heterogeneity and synchronization stability is based on the characterization of symmetries of the system \cite{nishikawa2016symmetric,garbin2020asymmetric,medeiros2021asymmetry,perego2022synchronization,bolotov2025heterogeneity}. A key insight from this perspective is that stabilizing a symmetric state often requires breaking the given symmetry of the governing dynamical equations. This phenomenon can be seen as a converse of symmetry breaking \cite{molnar2020network} and has been observed across many network systems in which symmetries are broken through heterogeneities in nodal parameters. Most studies have investigated this effect using approaches tailored to specific systems. However, determining when it occurs and how prevalent it is across a broad range of systems and dynamical states, including and beyond synchronization, remains a fundamental open question in physics.

\begin{figure}[!t]
    \centering
    \includegraphics[width=\linewidth]{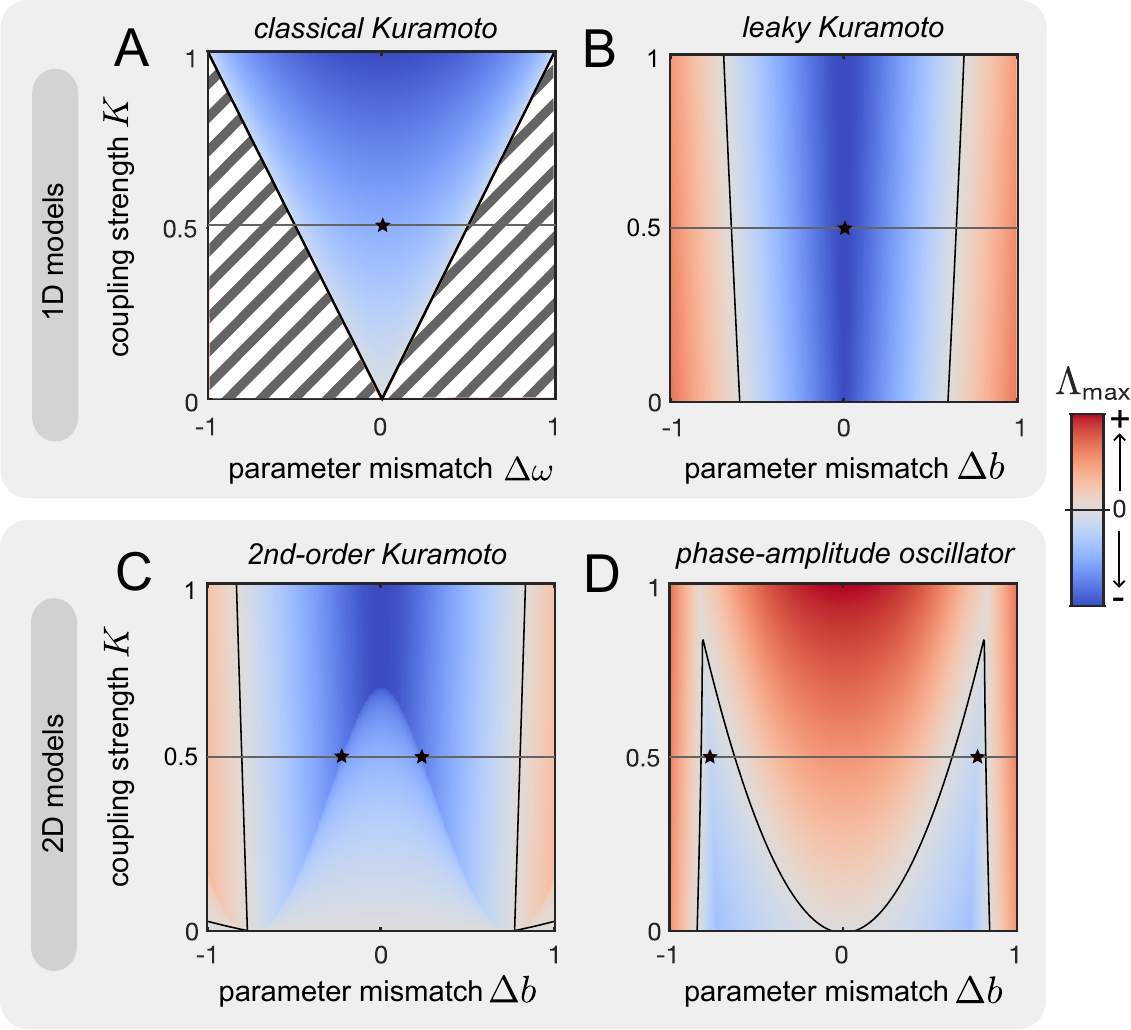}
    \caption{\textbf{Emergence of synchronization in heterogeneous oscillator networks.} 
    The stability regions (Arnold tongues) are shown for pairs of coupled classical Kuramoto oscillators (\textbf{A}), leaky Kuramoto oscillators (\textbf{B}), second-order Kuramoto oscillators (\textbf{C}), and phase-amplitude oscillators (\textbf{D}). The parameter regions where the synchronous state is stable (unstable) are marked blue (red), with the color code indicating the largest transverse Lyapunov exponent $\Lambda_{\mathrm{max}}$. The stars indicate the optimal parameter choices for a normalized coupling strength of $K=0.5$.
    For the classical Kuramoto model, synchronization occurs strictly when the frequency mismatch is sufficiently smaller than the coupling strength, with no synchronous equilibria outside the bound $|\Delta\omega| \leq 2K$ (hatched; where $
    \Delta\omega$ denotes the frequency mismatch). Both first-order models exhibit higher stability when the oscillators are nearly identical, with stronger coupling further expanding the stability region. In contrast, the second-order Kuramoto model attains {optimal} stability in weak coupling regimes at heterogeneous damping configurations. The phase-amplitude model exhibits an even more distinct behavior, where synchronization is only achieved with heterogeneity. The models and stability conditions are presented in the Supplementary Materials (SM), Section~\ref{sec.sm.arnold}.
    }
    \label{fig.arnold}
\end{figure}

Here, we {developed} a theory that reveals the mechanism through which heterogeneity stabilizes a dynamical state in general complex systems, showing that the beneficial role of heterogeneity can be traced to the dimensionality of the dynamics at the single node level. Specifically, we {established} that when the problem of maximizing stability through the tuning of nodal parameters is nonconvex, the optimal parameter configuration can be, and often is, heterogeneous. We thus {identified} that the necessary condition for the occurrence of this phenomenon is the non-Hermiticity of the Jacobian matrix, which commonly arises in systems whose nodal dynamics are described by two or more state variables. In the studies listed above, these additional state variables follow from second-order equations of motion, phase-amplitude dynamics, or other internal variables, setting them apart from one-dimensional (1D) models such as first-order Kuramoto oscillators. Fig.~\ref{fig.arnold} illustrates this effect, showing that the impact of heterogeneity on synchronization stability is fundamentally different in oscillators with one versus two degrees of freedom. {Moreover, for general network structures, we {found} that even when nodal parameters are randomly sampled, systems with moderate levels of ``disorder'' often exhibit states that are more stable than those of systems with identical parameters.}

Our analysis focuses on continuous-time systems in which all nodes obey the same underlying dynamical equations (e.g., Kuramoto, Lotka-Volterra, or FitzHugh-Nagumo dynamics), but admit heterogeneous nodal parameters or heterogeneous network structures. {Although heterogeneity has been widely studied for its role in instigating instability, particularly in synchronization dynamics, pattern formation, and multi-agent systems, here we {identified} conditions under which it can nevertheless promote stability in certain regions of parameter space.}
We {began} by analyzing nodal heterogeneity, establishing the conditions under which optimized and, more importantly, disordered parameter configurations can enhance stability relative to the best homogeneous configuration (Fig.~\ref{fig.disorder}). These conditions are governed by the nonconvexity and differentiability of the stability function, which form the basis of our framework for disorder-promoted stability (summarized in Fig.~\ref{fig.modemixing}A). Within this framework, we {showed} that, even though heterogeneity generically hinders dimensionality reduction by preventing the decoupling of network modes, in non-Hermitian systems this very mode mixing is the mechanism that enhances stability {(Fig.~\ref{fig.modemixing}B,C). We then {extended} our framework to the analysis of network heterogeneities, where non-Hermitian Jacobians can arise from directed network interactions and nonconvexity can also follow from optimization constraints.}
Thus, in contrast with the problem of optimizing nodal parameters, stability can be maximized through symmetry-broken network configurations regardless of the dimensionality of the nodal dynamics  {(Fig.~\ref{fig.symmetries})}. In the context of ecological systems, these results show that disordered networks provide a natural route toward resolving May’s longstanding complexity-stability paradox \cite{may1972will}: systems with more complexity are more likely to be stabilized by disorder, with the extent of this effect depending on the nature of the interactions {(Fig.~\ref{fig.ecologicalnet})}.

These findings demonstrate that some level of disorder is, in fact, expected to enhance the stability of a broad class of desired collective states, extending well beyond synchronization.

\section*{Results}

\subsection*{Stability and symmetry in network systems}
Many natural and engineered complex systems can be modeled as networks of $N$ interacting nodes (e.g., agents, oscillators, or species). We {considered} the broad class of models in which the state of each node, denoted $\bm x_i \in \mathbb{R}^q$, evolves according to 
\begin{equation}
    \dot{\bm x}_i = \bm f(\bm x_i;b_i) + \sum_{j=1}^N A_{ij}\bm g(\bm x_i,\bm x_j), \quad \text{for} \,\, i=1,\ldots,N.
\label{eq.generalnetwork}
\end{equation}
Here, the (generally nonlinear) functions $\bm f$ and $\bm g$ capture the self-dynamics and coupling terms, respectively, and the adjacency matrix $A\in\R^{N\times N}$ encodes the interaction network. The system state is denoted by $\bm x = (\bm x_1^\transp,\ldots,\bm x_N^\transp)^\transp$, with dimension $n=Nq$. We {examined} the case in which the heterogeneity among nodes is represented by the parameter vector $\bm b = (b_1,\ldots,b_N)$, which captures, for instance, differences in damping, natural frequency, or bias. A particularly important class of models of this form arises when the dynamics follow from Newton’s second law and are thus governed by coupled second-order equations of motion. In this case, the equations can be expressed as
\begin{equation}
    \ddot{\bm y}_i + b_i \dot{\bm y}_i + \tilde{\bm f}(\bm y_i, \dot{\bm y}_i) = \sum_{j=1}^N A_{ij} \tilde{\bm g}(\bm y_i, \dot{\bm y}_i,\bm y_j, \dot{\bm y}_j),
\label{eq.secondordersys}
\end{equation}
for $i=1,\ldots,N$, where $b_i$ is now the damping parameter. These equations map to a special case of Eq.~\eqref{eq.generalnetwork} by representing the system state as $\bm{x} = (\bm{y}_1^\transp, \ldots, \bm{y}_N^\transp, \dot{\bm y}_1^\transp, \ldots, \dot{\bm y}_N^\transp)^\transp$. 
The dynamics of drone swarms, mechanical metamaterials, and power grids all fall within this class. There are also cases such as the van der Pol oscillator, which arises both in second-order mechanical and electronic modeling and in phase-amplitude oscillators reduced from Stuart-Landau equations describing the universal dynamics near a Hopf bifurcation.
Therefore, the nodal dynamics of many network systems can be described by higher-dimensional models and/or second-order differential equations\textemdash Eqs.~\eqref{eq.generalnetwork} and \eqref{eq.secondordersys}, respectively\textemdash as illustrated in SM, Table~\ref{tab.systems}.

A central question in network dynamics is how the system symmetries influence the stability of the dynamical states.
To approach this question, we {examined} the system dynamics near an equilibrium point $\bm x^{\mathrm{eq}}$, as in the co-rotating frame of periodic orbits. Because small deviations from this state, $\bm x(t) = \bm x^{\mathrm{eq}} + \delta\bm x(t)$, evolve according to the first-order approximation $\delta\dot{\bm x} = J (\bm b;A) \delta\bm x$, {their asymptotic decay or growth rates are characterized by the associated Lyapunov exponents. At an equilibrium point, the Lyapunov exponents are given by the real part of the eigenvalues $\lambda_i$ of the Jacobian matrix $J\in\R^{n\times n}$. The equilibrium is therefore stable (unstable) if the largest transverse Lyapunov exponent $\Lambda_{\mathrm{max}} = \max_{i\notin \mathcal Z} \mathrm{Re}(\lambda_i)$ is negative (positive)}, where the set $\mathcal Z = \{i : \lambda_{i}\big(J(\bm b;A)\big)=0, \forall {\bm b\in \R^{N}}\}$ represents identically null eigenvalues (which are excluded when quantifying the stability against transverse perturbations).
In network systems, the relevant symmetries are defined by node permutations that leave the dynamical equations invariant. For example, in a ring network of identically coupled identical oscillators, these symmetries include rotations and reflections. Formally, system \eqref{eq.generalnetwork} is said to be symmetric under some node permutation represented by matrix $P$ if both the adjacency matrix $A$ and the parameter vector $\bm b$ satisfy the symmetry relations $P^\transp A P = A$ and $P \bm b = \bm b$, respectively. Likewise, a state $\bm x$ is symmetric if it satisfies the relation $(P\otimes I_q) \bm x = \bm x$, where $I_q$ is the $q\times q$ identity matrix and $\otimes$ denotes the Kronecker product.

Fig.~\ref{fig.arnold}A,B illustrates that, in a pair of symmetrically coupled 1D oscillators, the stabilization of the synchronous state $\bm x_1^{\rm eq}=\bm x_2^{\rm eq}$ requires a homogeneous choice of parameters. This result corresponds to the intuitive scenario in which symmetric systems are the optimal design choice to attain stable symmetric states\textemdash a behavior that has been tacitly assumed to be general.
Fig.~\ref{fig.arnold}C,D shows that this assumption is false. For 2D oscillators, the stabilization of such symmetric states may require breaking the system symmetry\textemdash in this case, through a heterogeneous (asymmetric) parameter choice $b_1\neq b_2$.
In what follows, we first {present} a general condition for the occurrence of this phenomenon in arbitrary network systems. We then {show} that the phenomenon is prevalent even when the heterogeneous parameters are disordered.

\begin{figure*}[!t]
    \centering
    \includegraphics[width=0.97\linewidth]{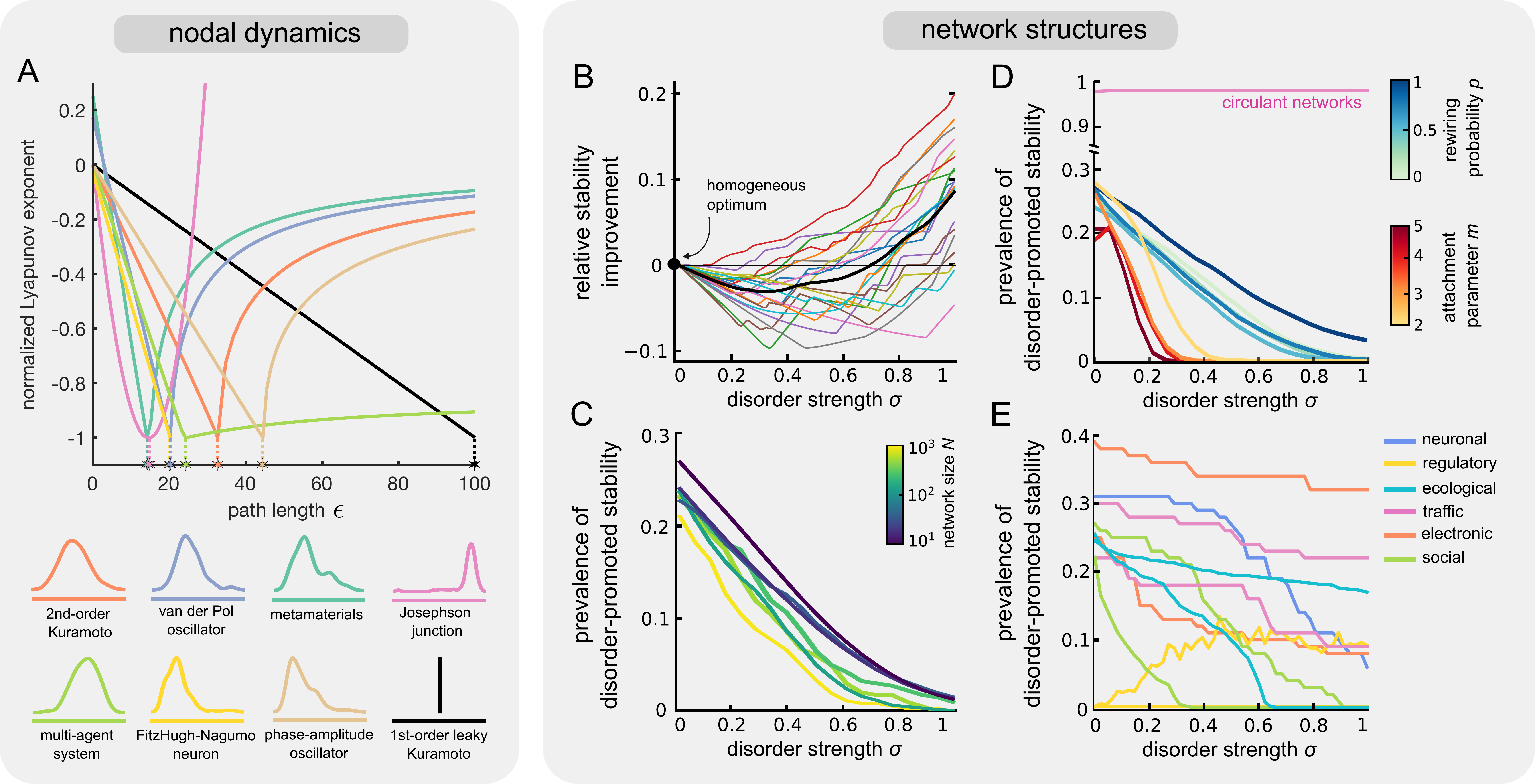}
    \caption{\textbf{Disorder-promoted stability across dynamical models and complex network structures.} 
    (\textbf{A})~Normalized largest transverse Lyapunov exponent, $\Lambda_{\mathrm{max}}(\bm b(\epsilon))/|\Lambda_{\mathrm{max}}(\bm b^*)|$, as a function of path length $\epsilon$ along the optimal direction. The curves are color coded by the dynamical models listed in SM, Table~\ref{tab.systems}, and the asterisks mark the optimum point within the interval. 
    \textit{Bottom insets}: Distribution of optimal parameter values $\bm b^*$, showing that all models achieve maximal stability at heterogeneous configurations, except for the first-order model. The simulations were performed on undirected all-to-all networks with 100 nodes. 
    (\textbf{B})~Relative stability improvement, $(\Lambda_{\mathrm{max}}(\sigma) - \Lambda_{\mathrm{max}}(0))/|\Lambda_{\mathrm{max}}(0)|$, as a function of disorder strength $\sigma$. Each colored curve corresponds to an independent realization of a directed SW network with $(N,p) = (100,0.2)$, using the best perturbation $\delta\bm b$ out of 1{,}000 trials. The average stability improvement is shown in black.
    (\textbf{C--E})~Fraction of systems with disorder-promoted stability (satisfying $\Lambda_{\mathrm{max}}(\sigma) \leq \Lambda_{\mathrm{max}}(0)$) as a function of $\sigma$. The results are shown for SW networks of varying sizes with $p=0.2$ (\textbf{C}), SW and SF networks of varying model parameters with $N=100$ (\textbf{D}), and empirical networks from diverse domains (\textbf{E}).
    The curves represent averages over 1{,}000 perturbation trials (all panels) and 100 network realizations (panels C and D).
    In panel D, results for the class of directed first-neighbor circulant networks are included for comparison.
    See Materials and Methods for details on empirical networks and models.
    }
    \label{fig.disorder}
    \vspace{-0.2cm}
\end{figure*}

\subsection*{Emergence of stability out of asymmetry}
To establish the relationship between state stability and system symmetry, {we {formulated} the following optimization problem: 
\begin{equation}
\label{eq.optimization}
    \min_{\bm b\in\R^N_{\geq 0}} \,\,\, \Lambda_{\mathrm{max}}\left(J(\bm b; A)\big|_{\bm x^{\rm eq}}\right),
\end{equation}
which seeks the (global) optimal parameter configuration $\bm b^*$ that minimizes $\Lambda_{\rm max}$, thereby maximizing the stability of an equilibrium point $\bm x^{\rm eq}$. (In Eq. 3, we implicitly assume a very large bound $|b_i|\leq b_{\rm max}$ to facilitate physical interpretation in cases in which $\bm b^*$ would be otherwise unbounded.) A key result is that the nonconvexity of this optimization problem is a necessary condition for the existence of asymmetric parameter configurations that promote stability beyond what is possible with symmetric parameters. To show this, we first {considered} the convex case and} {introduced} the following main theorem: 

\begin{adjustwidth}{0.5cm}{0cm}
``{{Let the Jacobian matrix $J(\bm b;A)$ be an affine function of $\bm b$ and $P$ be any permutation matrix representing a symmetry of the network (i.e., $P^\transp A P = A$).} If the optimization problem \eqref{eq.optimization} is convex, then at least one global optimal solution $\bm b^*$ {must satisfy all symmetries of the network (i.e., $\bm b^* = P\bm b^*$)}.}''
\end{adjustwidth}

\noindent
The problem is convex if $\Lambda_{\rm max}(\bm b)$ is a convex function of $\bm b$ and the feasible set of $\bm b$ forms a convex set, conditions under which every local minimum is also a global minimum.
A rigorous proof, based on Jensen's inequality, is presented in the SM, Section \ref{sec.sm.jensen}, {along with a more general formulation of this result for objective functions other than the largest Lyapunov exponent $\Lambda_{\rm max}$.} 

{An important consequence of this theorem is that it determines whether the stability of a state $\bm x^{\rm eq}$ is maximized by a homogeneous or a heterogeneous parameter configuration $\bm b^*$, where a parameter vector $\bm b$ is said to be homogeneous if all its components $b_i$ are equal. When $A$ corresponds to a network in which all nodes are structurally indistinguishable, any parameter that preserves all system symmetries is homogeneous (given that, for each pair of nodes $(i,j)$, there exists a permutation $P$ such that $P^\transp AP=A$ for which $P\bm e_i = \bm e_j$, where $\bm e_i$ denotes the canonical basis vector representing node $i$). For such networks, it follows from our theorem that, whenever {the optimization problem} is convex, at least one global optimal parameter $\bm b^*$ is homogeneous. It follows from the contrapositive that nonconvexity of {the optimization problem} is a necessary, but not sufficient, condition under which $\Lambda_{\rm max}$ is strictly minimized at heterogeneous parameter configurations $\bm b^*$. As we show next, such heterogeneous
optimal parameter configurations are not only possible but in fact common across network systems.}
These results are directly relevant to globally symmetric adjacency matrices $A$, including ring networks, all-to-all networks, and, more generally, all classes of circulant networks. More broadly, in arbitrary networks, the same statements can be made within each symmetry cluster.
Furthermore, because we {imposed} no requirements on the symmetries of $\bm x^{\mathrm{eq}}$, this result is not limited to synchronization states and is broadly relevant to the stabilization of arbitrary states (including cases in which $\bm x_i^{\mathrm{eq}} \neq \bm x_j^{\mathrm{eq}}$), with the understanding that the {convexity of Eq.~\eqref{eq.optimization}} itself depends on $\bm x^{\rm eq}$.
In our formulation, it is therefore important to distinguish between the symmetries of the network ($A$), of the optimal parameters ($\bm b^*$), and of the state under consideration ($\bm x^{\rm eq}$).

Fig.~\ref{fig.disorder}A shows, for a variety of network systems, how the stability of the equilibrium state $\bm x^{\mathrm{eq}}$ changes as the parameter $\bm b$ is varied along a straight line from $\bm b=0$ to the optimum $\bm b = \bm b^*$ according to the parameterization $\bm b(\epsilon) = \epsilon\bm b^*/\norm{\bm b^*}$. 
Here, $\bm b$ represents damping coefficients, leak/decay rates, or regulation feedback depending on the specific system, and we note that $\bm b = 0$ typically corresponds to the least stable configuration for $b_i\geq 0$, $\forall i$. 
For second-order systems of the form \eqref{eq.secondordersys}, it is well known that optimal stability is achieved at finite damping values, which is often referred to as critical damping. {More generally,} the optimal parameter $\bm b^*$ often lies at a nondifferentiable point of $\Lambda_{\rm max}$, as illustrated in Fig.~\ref{fig.disorder}A for a wide variety of systems governed by second-order, phase-amplitude, or excitation-inhibition dynamics.

Crucially, for the networks with 2D nodal dynamics, the optimal stability $\Lambda_{\mathrm{max}}(\bm b^*)$ always {occured} at heterogeneous parameter configurations $\bm b^*$, characterized by broad, dominantly unimodal distributions (Fig.~\ref{fig.disorder}A, bottom); {the 1D Kuramoto model (black line) {was} the only case for which optimal stability, within a finite interval, {was} attained at homogeneous values.}
These observations are directly aligned with our theorem, given that the optimization problem~\eqref{eq.optimization} is generally nonconvex for networks with 2D nodes, {but convex for 1D nodes. To demonstrate this, we compare networks of second- and first-order Kuramoto oscillators (the same systems considered in Fig.~\ref{fig.arnold}B,C).}

\medskip\noindent
{\textit{When heterogeneity is optimal.} For the second-order Kuramoto model,} the stability of the synchronous state $\bm x_1^{\rm eq} = \ldots = \bm x_N^{\rm eq}$ is given by the Jacobian matrix
\begin{equation}
    J(\bm b;A) = \begin{bmatrix}
        0_N & I_N \\ - L & -B
    \end{bmatrix},
\label{eq.jacobian}
\end{equation}
where $0_N$ is an $N \times N$ zero matrix, $B = \operatorname{diag}(b_1,\ldots,b_N)$ is a diagonal parameter matrix, $L=D-A$ is the Laplacian matrix, $D = \operatorname{diag}(d_1,\ldots,d_N)$ is the diagonal degree matrix, and $d_i = \sum_{j} A_{ij}$ is the in-degree of node $i$ {(SM, Table~\ref{tab.systems})}. {Consistent with the assumption of our theorem (SM, Sec.~\ref{sec.sm.jensen}), the Jacobian matrix $J$ depends affinely on $\bm b$; that is, each of its entries varies linearly with $\bm b$ up to a constant.} Importantly, $J$ is non-Hermitian (i.e., $J\neq J^\dagger$) for any connected network, even when the network itself is undirected (i.e., when $A=A^\dagger$). (Note that we use the term ``Hermitian'' rather than ``symmetric'' for real matrices to avoid conflating matrix invariance under transposition with symmetry properties in the parameter space, state space, and dynamical equations.) {In this setting,} the {optimization problem \eqref{eq.optimization}} is always nonconvex, implying that the optimal stability of the symmetric state $\bm x^{\rm eq}$ may be realized at heterogeneous (asymmetric) parameter configurations. 

More generally, we {proved} that the optimization problem \eqref{eq.optimization} is nonconvex {only if $J$ is non-Hermitian and that, for a broad class of non-Hermitian systems, nonconvexity holds for almost all networks (SM, Propositions~\ref{cor.sm.nonhermitian} and \ref{cor.sm.localcurvature}).} That is, convexity holds only for a measure-zero set in the space of possible adjacency matrices.
{This result indicates that regimes in which parameter heterogeneity enhances stability are common rather than exceptional in complex systems, as most such systems feature nodal dynamics with two or more state variables and therefore non-Hermitian Jacobians.}
Within this class of systems, these results apply to both directed and undirected networks.
{A relation can also be established between our convexity-symmetry results and the nonnormality of $J$ (a stronger condition than non-Hermiticity) as well as other stability measures for transient dynamics (SM, Sec.~\ref{sec.sm.nonconvexity}). {In particular, for the class of systems considered here, non-Hermiticity {is} equivalent to nonnormality, and hence it follows that heterogeneity can promote stability only if $J$ is nonnormal. This connection is noteworthy given that previous research has shown a common occurrence of nonnormality in empirical networks \cite{asllani2018structure,duan2022network}, with direct consequences for transient response analysis and the optimization of synchronization stability \cite{nishikawa2017sensitive,nazerian2024efficiency}.}}

\medskip\noindent
{\textit{When homogeneity is optimal}. To facilitate comparison to the second-order systems shown in Fig.~\ref{fig.disorder}A, we {considered} the ``leaky'' Kuramoto model \cite{Dorfler2014} (SM, Table~\ref{tab.systems}), which is a representative system with first-order nodal dynamics in which $\bm b$ also acts as a linear dissipative term. The corresponding Jacobian matrix around the synchronous state is} 
\begin{equation}
    J(\bm b;A)=-(B+L).
\label{eq.jac1storder}
\end{equation}
{In this case, $J$ is Hermitian whenever $A$ is Hermitian. Under this condition, {the optimization problem \eqref{eq.optimization}} is convex, implying that optimal stability is attained at a homogeneous parameter $\bm b^*$.} This analysis extends directly to other first-order models that share the same Jacobian structure, including externally driven phase oscillators \cite{Rodrigues2016,wang2025explosive}, consensus models in distributed coordination \cite{ren2008distributed}, and opinion dynamics \cite{castellano2009statistical} (SM, Table~\ref{tab.systems}).

\medskip
\subsection*{Disorder-promoted stability in complex networks}
\label{sec.modemixing}

Our results established that optimally stable states are generally achieved with heterogeneous parameters because most network systems have nodal dynamics with two or more state variables, leading to non-Hermitian Jacobians and nonconvex stability landscapes. Thus, the optimization problem is $N$-dimensional rather than 1D, and we now turn to the question of {convergence}: under what conditions are such heterogeneous optimal configurations easily reached through optimization or controlled parameter perturbations?
Another fundamental question concerns {generality}: to what extent do these conditions hold across complex empirical and model networks?

To address these questions, we {focused} our analysis on network models of coupled second-order equations described by the Jacobian matrix \eqref{eq.jacobian}, with the understanding that many of the analytical results also extend to Jacobians with different structures (SM, Sec.~\ref{sec.sm.dpsconditions}). 
Let the homogeneous optimal parameter, which is the solution to the optimization problem \eqref{eq.optimization} subject to the constraint $\bm b = b\bm 1_N$, be denoted $\bm b_{\mathrm{hom}}^* = b_{\rm hom}^* \bm 1_N$, where $\bm 1_N$ is the $N$-dimensional ones vector.  
For {undirected networks} (i.e., Hermitian $A$), we {showed} that all systems described by the Jacobian matrix \eqref{eq.jacobian} have their homogeneous optima $\bm b_{\rm hom}^*$ located at nondifferentiable points of $\Lambda_{\rm max}(\bm b)$ and that this property also holds for all global optima $\bm b^*$ (SM, Sec.~\ref{sec.sm.differentiability}). 
{The nondifferentiable nature of the problem exposes a critical barrier to optimization:} gradient- and Hessian-based approaches are generally inapplicable or incapable of converging to global optima, including when starting from $\bm b^*_{\rm hom}$.

However, complex systems are generally described by {directed networks}, whose optimization landscapes can exhibit quite distinct differentiability properties. (As a shorthand, we say that a network is (non)differentiable when $\Lambda_{\rm max}(\bm b)$ is (non)differentiable at the given point.) 
A broad class of directed circulant networks is differentiable at homogeneous optimal points $\bm b^*_{\mathrm{hom}}$ (SM, Fig.~\ref{fig.circulant}). This result can be rigorously proven, for example, for all directed first-neighbor ring networks. Whenever a circulant network is differentiable at $\bm b^*_{\mathrm{hom}}$,  we can further show~that $\bm b^*_{\mathrm{hom}}$ is also a saddle point of $\Lambda_{\rm max}$ (SM, Sec.~\ref{sec.sm.circulant}).
In~such cases, there exist arbitrarily small perturbations of the form $\bm b(\sigma) = \bm b_{\mathrm{hom}}^* + \sigma \delta\bm b$ that locally improve stability, so that $\Lambda_{\mathrm{max}}\big(\bm b(\sigma)\big) < \Lambda_{\mathrm{max}}(\bm b_{\mathrm{hom}}^*)$ for small $\sigma$. Crucially, the perturbation direction $\delta \bm b$ need not be finely tuned. As demonstrated next, even random perturbations $\delta b_i$ drawn from a normal distribution $\mathcal N(0,1)$ {were} likely to enhance stability for controlled values of $\sigma$. We refer to this stabilizing effect as ``{disorder-promoted stability}'' as it arises from random parameter perturbations. 

We {investigated} the prevalence of this phenomenon in complex network topologies, extending our analysis beyond circulant networks. To this end, we first {considered} an ensemble of small-world (SW) networks of size $N$, rewiring probability $p$, and random weights $A_{ij}\sim\mathcal U[0,1]$. This procedure yields directed SW networks, ensuring that the homogeneous optimum $\bm b_{\mathrm{hom}}^*$ can be differentiable. Fig.~\ref{fig.disorder}B shows, for different network realizations, the stability improvement promoted by certain perturbations $\delta\bm b$ as a function of $\sigma$. While the magnitude of the improvement is dependent on the network realization, disorder-promoted stability is observed systematically across the entire ensemble. 
A more practical measure of this effect is the likelihood that a random perturbation improves stability (Fig.~\ref{fig.disorder}C--E). For small $\sigma$, roughly 25\% of all perturbations yield stability gains in SW networks, a percentage that is largely independent of the network size (Fig.~\ref{fig.disorder}C).
Most importantly, disorder-promoted stability is prevalent across different directed network models (Fig.~\ref{fig.disorder}D){, including both SW and scale-free (SF) networks, as well as} empirical networks (Fig.~\ref{fig.disorder}E). Compared to these networks, directed circulant networks {exhibited} the highest prevalence of the effect, with probabilities close to 100\%, as anticipated by our theory.
Moreover, within the networks considered, the effect {was} more prevalent in {degree-homogeneous} networks, such as circulant networks and SW networks at extreme values of $p$, than in {degree-heterogeneous} networks such as SF ones. These networks differ in their spectral properties: the hub-dominated structure of SF networks leads to eigenmode localization \cite{goh2001spectra,metz2021localization}, whereas the eigenmodes of degree-homogeneous networks tend to be delocalized \cite{erdHos2013spectral,davis1979circulant} (SM, Fig.~\ref{fig.sm.isotropy}). This observation hints that the structure of the eigenmodes plays a role in disorder-promoted stability, which we explore next.

\begin{figure*}[t]
    \centering
    \includegraphics[width=1\linewidth]{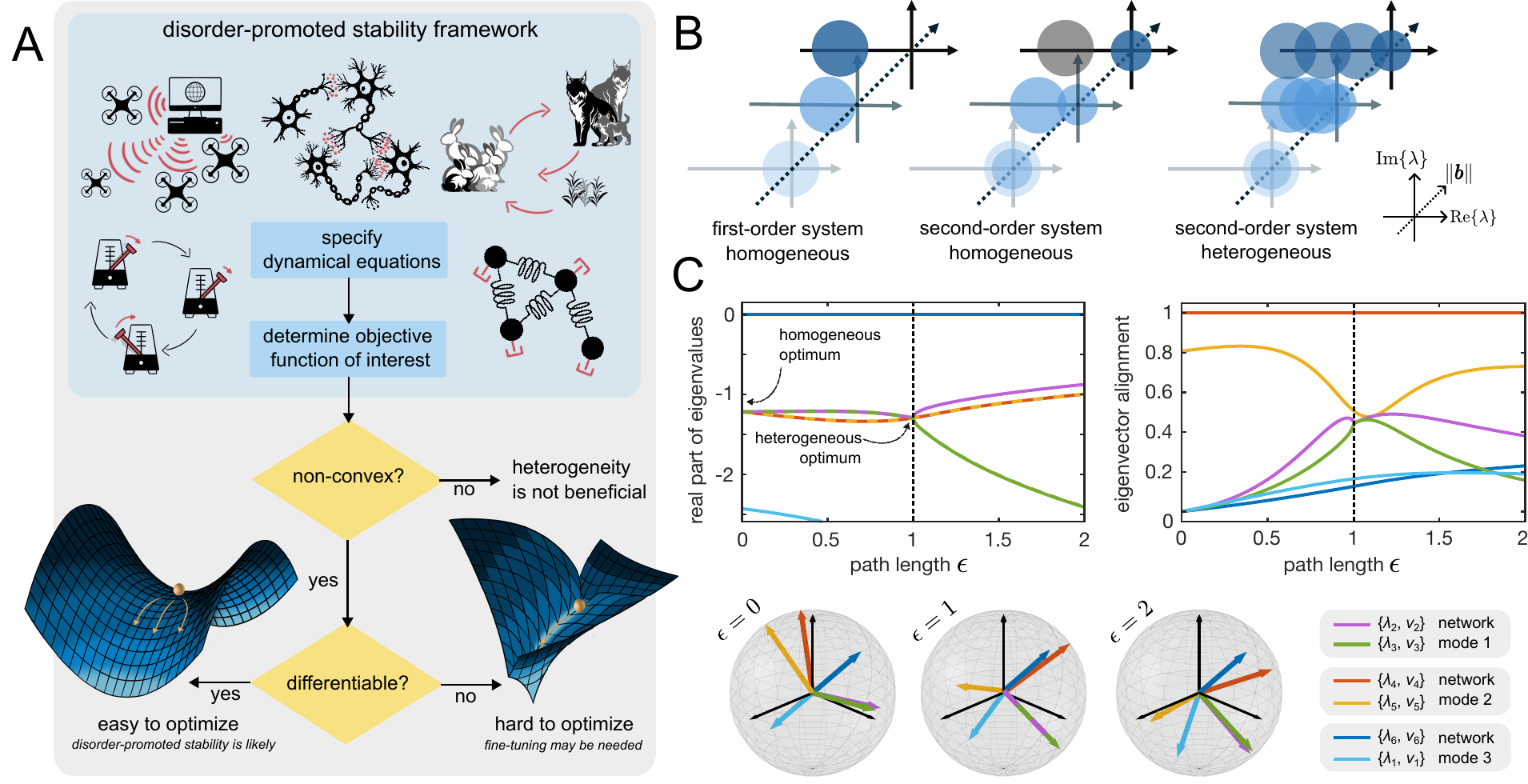}
    \caption{
    \textbf{Mechanism underlying disorder-promoted stability.} 
    (\textbf{A})~Framework to identify disorder-promoted stability. First, identify the system’s symmetries and the nodal parameters of interest. 
    Formulate the governing equations and determine the objective function for stability analysis (which need not be strictly restricted to Lyapunov exponents; see SM, Lemma~\ref{lem.sm.generaljensen}, for a general formulation).
    The existence of disorder-promoted stability depends on the objective function's nonconvexity. In circulant networks, convex landscapes yield homogeneous global optima, whereas nonconvex landscapes allow the occurrence of heterogeneous global optima.
    The hardness of improving stability relative to homogeneous configurations depends on whether the function is differentiable or not at these parameters. Orange arrows in the landscapes show possible parameter perturbations that yield improved stability relative to the homogeneous optimum (orange dot). 
    (\textbf{B})~Placement of Gershgorin discs in the complex plane for different classes of network models. As damping $\|\bm b\|$ grows, the disc placement differs qualitatively between first- and second-order models: in first-order systems, all discs shift leftward for larger damping, whereas second-order systems always exhibit a family of discs trapped at the origin of the complex plane. The blue regions delimit the possible range of $\Lambda_{\mathrm{max}}$ based on Gershgorin's theorem, which can be substantially more negative for heterogeneous second-order systems compared to their homogeneous counterparts.
    (\textbf{C}) Real part of the eigenvalues and alignment of eigenvectors as functions of the normalized path length $\epsilon$ connecting $\bm b(0) = \bm b_{\mathrm{hom}}^*$ to $\bm b(1) = \bm b^*$. The results are shown for a 3-node undirected network with random weights, where eigenpairs are color coded and alignment is measured relative to the red eigenvector (i.e., $\bm v_4^\dagger \bm v_i$, for $i=1,\ldots,6$). A 3D projection of the eigenvectors is illustrated for varying $\epsilon$, highlighting the eigenvector alignment at the heterogeneous optimal parameter ($\epsilon = 1$).
    \label{fig.modemixing}}
    \vspace{-0.2cm}
\end{figure*}

\subsection*{{Framework to identify disorder-promoted stability}}
We summarize our framework to identify disorder-promoted stability in  Fig.~\ref{fig.modemixing}A, highlighting two key ingredients discussed thus far: the nonconvexity of the optimization landscape and the differentiability with respect to system parameters. Nonconvexity provides a necessary condition for the existence of heterogeneous parameter configurations (fine-tuned or not) whose stability exceeds that of the best homogeneous configuration. Differentiability, on the other hand, allows the homogeneous optimum to be a saddle point, which is a sufficient condition for the existence of heterogeneous configurations with higher stability\textemdash often reachable through random parameter perturbations.
Both ingredients are encoded in the structure of the Jacobian matrix, capturing the complex interplay between nodal dynamics and network topology.

{
Our results can be contextualized within this framework as follows. In networks governed by coupled second-order equations of motion, the Jacobian is non-Hermitian and thus the objective function of interest, $\Lambda_{\rm max}(\bm b)$, is nonconvex, implying that the global optimum $\bm b^*$ may be heterogeneous. For undirected networks ($A=A^\dagger$), the objective function $\Lambda_{\rm max}$ is nondifferentiable at any minimum and hence finding the global minimum $\bm b^*$ is a challenging optimization task. In contrast, for directed networks ($A\neq A^\dagger$), the function $\Lambda_{\rm max}$ can be differentiable at $\bm b_{\rm hom}^*$. In this case, $\bm b_{\rm hom}^*$ may be a saddle point, so that small random perturbations $\bm b_{\rm hom}^* + \delta\bm b$ can improve stability, thus leading to disorder-promoted stability.
}

\subsection*{{A mechanistic interpretation}}
We start with a simple geometric intuition based on Gershgorin’s disc theorem, which bounds all eigenvalues of $J$ within the union of closed discs $D_i(J_{ii},R_i)$ centered at $J_{ii}$ with radius $R_i = \sum_{j\neq i} |J_{ij}|$, for $i=1,\ldots,n$.
For coupled first-order systems described by the Jacobian \eqref{eq.jac1storder}, the increase of homogeneous damping shifts all discs uniformly leftward in the complex plane, reducing $\Lambda_{\mathrm{max}}$ and thus increasing stability (Fig.~\ref{fig.modemixing}B, left). In contrast, second-order systems described by the Jacobian \eqref{eq.jacobian} have two distinct families of discs: one associated with ``inertial modes'' $\mathcal D_1 = \{D_1(0,1),\ldots,D_N(0,1)\}$ and another with ``damping modes'' $\mathcal D_2 = \{D_{N+1}(-b_1,2 d_1),\ldots,D_{2N}(-b_N,2 d_N)\}$.  We recall that a disjoint cluster of $k$ connected discs must contain exactly $k$ eigenvalues \cite{varga2011gervsgorin}. Under homogeneous damping, the families $\mathcal D_1$ and $\mathcal D_2$ are disjoint for large $\|\bm b\|$, trapping part of the spectrum\textemdash including $\Lambda_{\mathrm{max}}$\textemdash within the less stable region $\mathcal D_1$ (Fig.~\ref{fig.modemixing}B, middle). In contrast, under heterogeneous damping, the damping discs are shifted leftward unevenly for large $\|\bm b\|$; this creates overlaps that allow eigenvalues to ``escape'' from $\mathcal D_1$ into the more stable regions in $\mathcal D_2$ (Fig.~\ref{fig.modemixing}B, right).
{Even though Gershgorin’s theorem highlights how heterogeneity can alter the eigenspectrum to potentially enhance stability, it only provides conservative bounds. A leftward shift of the Gershgorin discs does not necessarily imply that $\Lambda_{\rm max}$ becomes more negative.}

To precisely characterize how the eigenspectrum is modulated by heterogeneity, we must examine the structure of the eigenvectors. Assume that $L$ is diagonalizable, with $U^{-1} L U = \operatorname{diag}(\lambda_{L,1},\ldots,\lambda_{L,N}) = D_{L}$, where $\lambda_{L,i}$ are eigenvalues of $L$. The Jacobian matrix \eqref{eq.jacobian}  is thus similar to
\begin{equation}
\label{eq.similarJ}
    J'(\bm b;A) = (I_2\otimes U)^{-1} J (I_2\otimes U) = \begin{bmatrix}
        0_N & I_{N} \\
        -D_{L} & -U^{-1}BU\end{bmatrix}.
\end{equation}
This representation shows that, when $B = b_{\rm hom}I_N$ and $U$ can be chosen to be unitary (as in the case of undirected networks), each eigenvector of $L$ maps uniquely to a pair of eigenvectors of $J$. We refer to each such pair as a ``{network mode},'' noting that under homogeneous damping all network modes are mutually orthogonal (SM, Proposition~\ref{prop.JJ'}). A central result is that, for undirected networks, {the optimum parameter $\bm b^*$ is reached only if at least two eigenvectors of $J$ become parallel} (SM, Theorem~\ref{thm.mode_mixing}). We now {show} that heterogeneity can induce such parallelism among eigenvectors of $J$.

\begin{figure*}[t]
    \centering
    \includegraphics[width=0.89\textwidth]{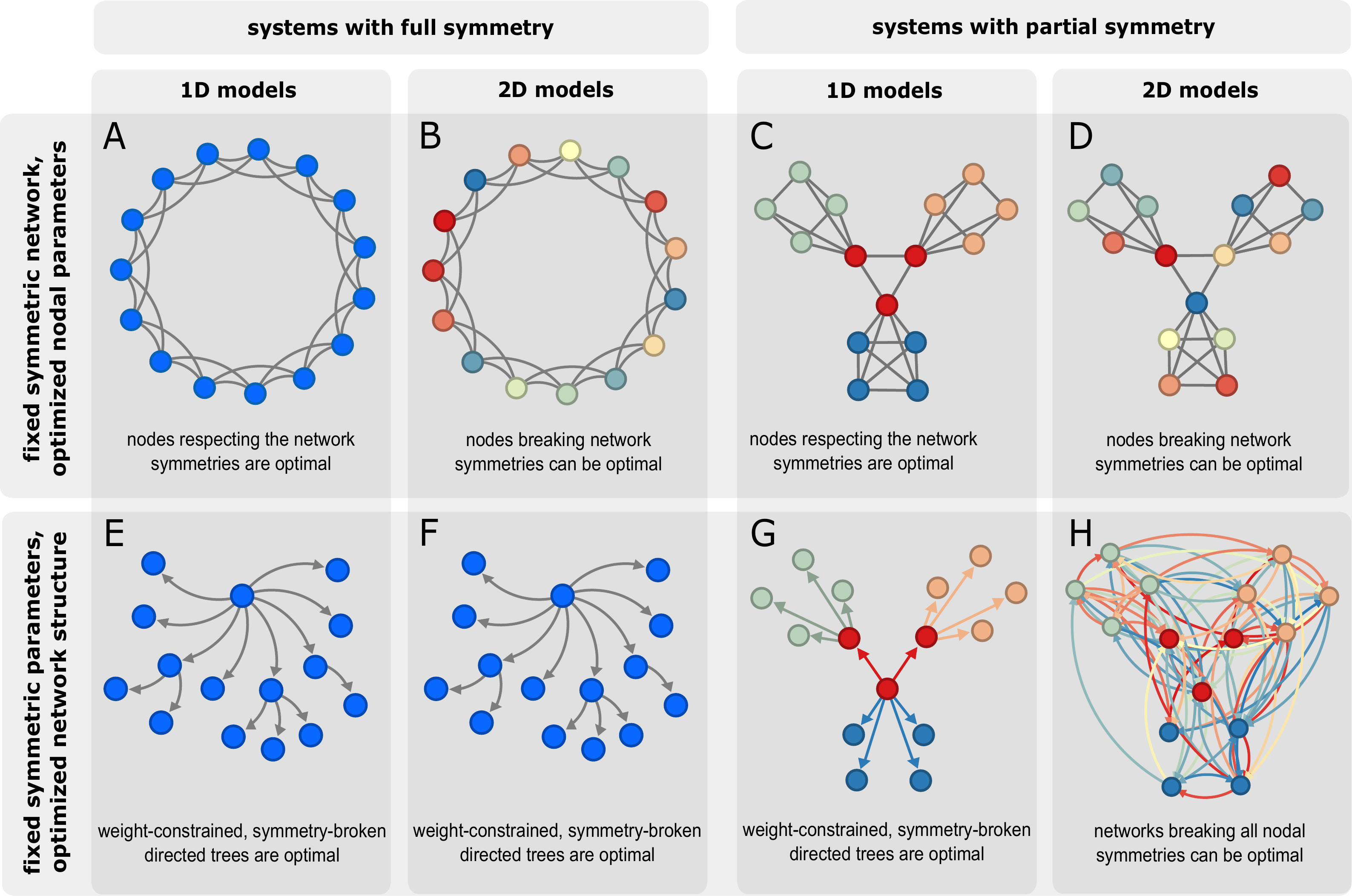}
    \caption{\textbf{Relationship between optimal state stability and system symmetries.} 
    (\textbf{A--D})~Optimal nodal parameter configurations that maximize stability for a {fixed} (symmetric or partially symmetric) network structure. The node colors indicate distinct nodal parameters. Panels A and B correspond to circulant networks, in which all nodes belong to a single symmetry cluster determined by cyclic permutations, whereas panels C and D show networks with four distinct permutation symmetry clusters.
    (\textbf{E--H})~~Optimal network structures that maximize stability for a {fixed} (symmetric or partially symmetric) assignment of nodal parameters. Different colors indicate distinct interaction weights and nodal parameters. Panels E and F show systems with identical (fully symmetric) nodal parameters, whereas panels G and H show systems with four distinct parameter values (corresponding to four symmetry clusters). In all cases, the illustrated configurations that optimize stability are only achieved for specific choices of nodal parameters (panels A to D) or edge weights (panels E to H).
    }
    \label{fig.symmetries}
    \vspace{-0.2cm}
\end{figure*}

Under the assumption that $U$ in Eq.~\eqref{eq.similarJ} is unitary, we apply the transformation $\delta \bm \eta = U^{-1}\delta \bm y$ to the linearized dynamics $\delta\dot{\bm x} = J (\bm b;A) \delta\bm x$, where $\delta\bm{x} = (\delta{\bm y}^\transp, \delta{\dot{\bm y}}^\transp)^\transp$, to obtain the variational equation
\begin{equation}
    \delta \ddot \eta_i + \sum_{k=1}^N \sum_{j=1}^N U_{ki}U_{kj} b_k  \delta\dot\eta_j = - \lambda_{L,i} \delta\eta_i, \quad {\rm for} \,\, i=1,\ldots,N.
\label{eq.variational}
\end{equation}
In the homogeneous case, we have that $U^{-1}BU = b_{\rm hom}I_N$ and hence the modes $\delta \eta_i$ are completely decoupled, as in the classical MSF analysis \cite{fujisaka1983stability,Pecora1998}. Introducing heterogeneity in $\bm b$ creates an off-diagonal structure in $J'$, which in turn induces {mode mixing} among $\delta\eta_i$ through the cross-terms $U_{ki}U_{kj}b_k\delta\dot{\eta}_j$. 
This mode mixing has been previously accounted for in the MSF analysis of near-identical oscillators \cite{sun2009master,sugitani2021synchronizing}.
Eq.~\eqref{eq.variational} allows us to extend the analysis of mode mixing even when the mismatches $(b_i-b_j)$ are very large, providing an exact description of how heterogeneities affect stability for equilibrium points that are independent of $\bm b$. 
Importantly, in the context of this study, mode mixing is the mechanism that forces two (or more) eigenvectors of $J$ to align, satisfying the necessary condition for optimal stability.
The analysis applies not only to the second-order Kuramoto model but also to a broader class of second-order systems, including all of those in SM, Table~\ref{tab.systems}.

Fig.~\ref{fig.modemixing}C illustrates the mode-mixing mechanism in a 3-node network as the damping $\bm b(\epsilon)$ varies along a straight line from $\bm b = \bm b_{\mathrm{hom}}^*$ to the (heterogeneous) global optimum $\bm b = \bm b^*$, parameterized as $\bm b(\epsilon) = \bm b_\mathrm{hom}^* + \epsilon(\bm b^* - \bm b_\mathrm{hom}^*)$. As predicted by our theory (SM, Proposition~\ref{prop.JJ'} and Theorem~\ref{thm.mode_mixing}), the initially orthogonal network modes progressively align, and a pair of eigenvectors (purple-green) becomes parallel exactly when $\Lambda_{\mathrm{max}}$ reaches its minimum value at $\bm b^*$. Thus, heterogeneity acts as the driver of mode mixing, which in turn causes eigenvector alignment, leading to maximal stability.

\vspace{-2pt}
\subsection*{Nodal versus network heterogeneities}
Thus far, we have focused on heterogeneities in {nodal parameters}. The same framework also applies to the analysis of how heterogeneities in {network structure} affect state stability.
In particular, we {addressed} the problem of identifying an optimal network that maximizes state stability under budget constraints on the interaction weights, formulated as
\begin{equation}
\begin{aligned}
    \min_{A} \,\, & \Lambda_{\rm max}\left(J(\bm b; A)\big|_{\bm x^{\rm eq}}\right) \\
    \text{s.t.} \,\, & |d_i|\leq d_{\textrm{max},i}, \,\, \|A\|_0 \leq S_{\rm max}.
\label{eq.optimizationet}
\end{aligned}
\end{equation}
The constraints on the node in-degrees are introduced to prevent trivial solutions involving unbounded edge weights, and the $\ell_0$-norm enforces a sparsity constraint on the total number of edges in the network. Extending our result for nodal heterogeneities, we {proved} that, whenever Eq.~\eqref{eq.optimizationet} is convex, at least one optimal adjacency matrix $A^*$ must respect the system symmetries imposed by the nodal parameter vector $\bm b$ (i.e., the permutations $P$ satisfying $\bm b = P\bm b$). See SM, Theorem~\ref{thm.sm.optadj}, for a proof of this result.
This constrained optimization problem is always nonconvex (SM, Corollary~\ref{cor.sm.optadj_constrain}), even if $A$ were restricted to be Hermitian, and hence we can conclude that network structures that break node-permutation symmetries imposed by $\bm b$ can yield higher stability.
More generally, even in the unconstrained case (with $S_{\rm max} = N^2$ and $d_{{\rm max},i}=\infty$), the adjacency matrix (and hence the Jacobian matrix) is generally non-Hermitian in the optimization space considered, independently of the dimensionality of the nodal dynamics; thus, the unconstrained optimization problem is also generally nonconvex.

These theoretical results are illustrated in Fig.~\ref{fig.symmetries}. The second row demonstrates that symmetry breaking induced by network heterogeneity can enhance stability in both 1D and 2D models, which is in contrast with the corresponding case of nodal heterogeneity summarized in the first row. 
Specifically, for systems with homogeneous nodal parameters, we {showed} that directed trees are optimal (SM, Example~\ref{examp.sm.optadj}). This result is consistent with previous MSF analyses identifying directed networks as optimal \cite{nishikawa2006synchronization,nishikawa2006maximum}, and it aligns well with our nonconvexity condition, given that directed trees explicitly break system symmetries. For example, in Fig.~\ref{fig.symmetries}E,F, the nodal parameter vector is fully symmetric (identical nodes), but the optimal network is not invariant upon permutations of nodes in different layers of the tree. Only in the special case $S_{\rm max} \geq N^2 - N$ is one of the optimal solutions fully symmetric (in fact, an all-to-all network).
For systems with partially symmetric nodal parameters, directed trees remain optimal in 1D models only if the assignment of edge weights respects the symmetry clusters induced by $\bm b$ (Fig.~\ref{fig.symmetries}G). In contrast, in 2D models described by the Jacobian~\eqref{eq.jacobian}, the optimal networks exhibit a more complex structure with broad degree distributions   and heterogeneous weights (Fig.~\ref{fig.symmetries}H and SM, Example~\ref{examp.sm.optadj2ndorder}).
These results lead naturally to the question of whether {disordered} network structures, characterized by randomly wired edges and/or interaction weights, can also promote stability, as considered next.

\begin{figure*}[ht!]
    \centering
    \includegraphics[width=0.8\textwidth]{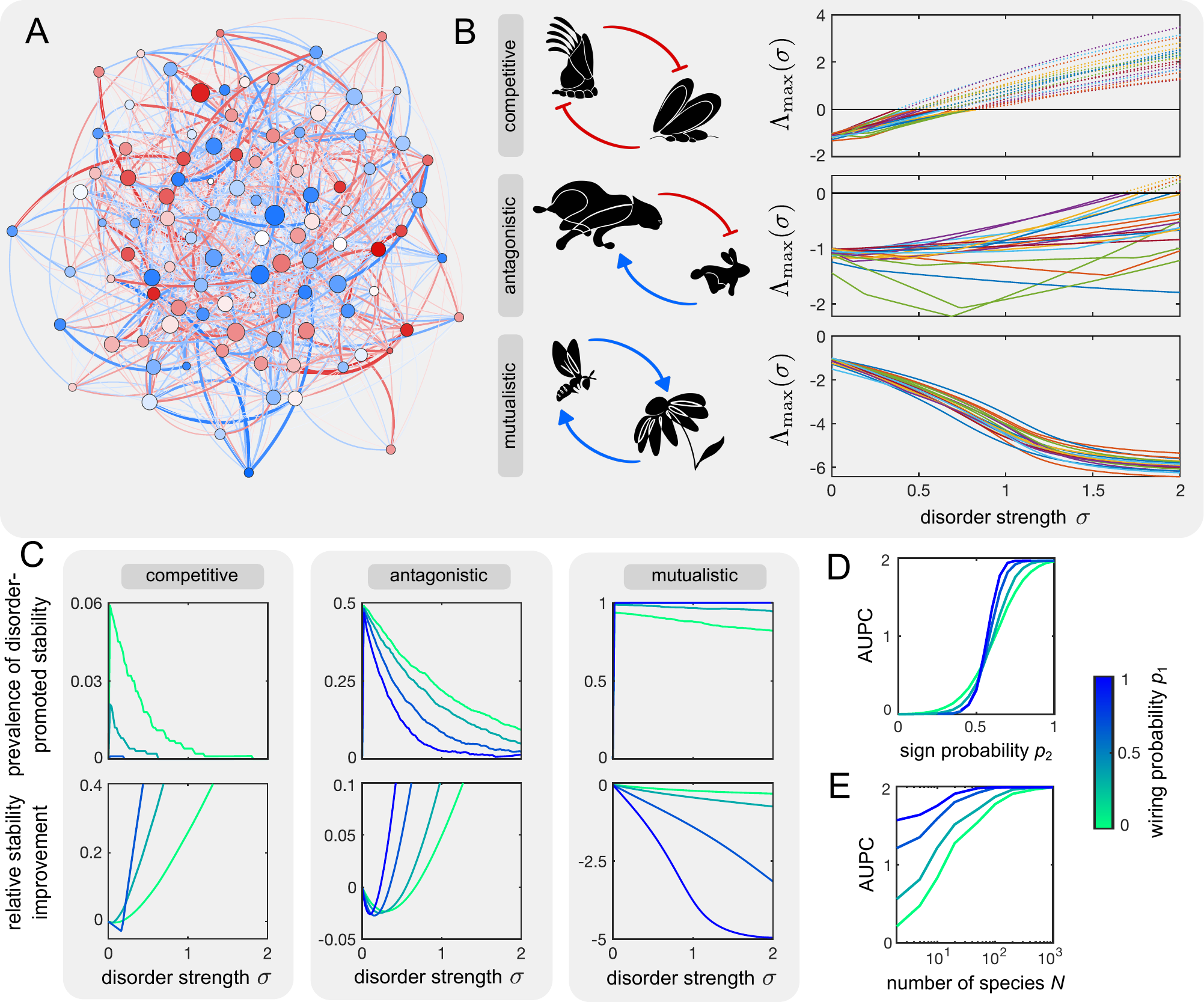}
    \caption{\textbf{Stability of large disordered ecological networks.} 
    (\textbf{A})~Example of a random ecological network with balanced excitatory (blue) and inhibitory (red) interactions, generated with parameters $(N,p_1,p_2)=(100,0.1,0.5)$. The nodes and edges are color coded by the in-degree and interaction weight, respectively.
    (\textbf{B}) Largest Lyapunov exponent $\Lambda_{\mathrm{max}}$ as a function of the disorder strength $\sigma$ for LV models with predominantly competitive (\textit{top}, $p_2 = 0.2$), antagonistic (\textit{middle}, $p_2 = 0.5$), and mutualistic (\textit{bottom}, $p_2 = 0.8$) interactions. Each curve corresponds to an independent network realization, where the dotted lines mark regions with nonfeasible equilibria (i.e., negative species abundances). Adding disorder to the system can improve stability in antagonistic and mutualistic networks.
    (\textbf{C})~Prevalence of systems exhibiting disorder-promoted stability (\textit{top}) and relative stability improvement (\textit{bottom}) as functions of $\sigma$, shown for different interaction types and color coded by the wiring probability $p_1$ of ER graphs. 
    (\textbf{D}) Area under the prevalence curve (AUPC) of disorder-promoted stability, shown as a function of the sign probability $p_2$ and color coded by the wiring probability $p_1$.
    (\textbf{E}) AUPC of disorder-promoted stability for mutualistic networks ($p_2 = 0.8$), now presented as a function of the number of species $N$.
    The results in panels C--E are averaged over 1{,}000 network realizations, and the AUPC is computed over the interval $\sigma\in[0,2]$ (as considered in panel C). These results indicate that mutualistic ecological networks with increased complexity (i.e., edge density and network size) are more likely to have their stability improved by disorder.  
    }
    \label{fig.ecologicalnet}
\end{figure*}

\subsection*{Stable biodiversity in disordered ecological networks}
Ecological systems have long been central to the study of disorder in complex systems, as species interactions are inherently heterogeneous due to variations in biological traits, population sizes, and environmental conditions. The coexistence of species, and thus biodiversity, critically depends on ecosystem stability against perturbations. A longstanding theoretical question concerns the conditions under which communities persist in the presence of disorder \cite{may1972will,chen2024stability}. Studies based on multi-species population dynamics, particularly the generalized Lotka-Volterra (LV) model, have shown that structural heterogeneity in interaction networks can either inhibit or promote coexistence, depending on whether the interactions are excitatory or inhibitory \cite{jansen2003complexity,sahasrabudhe2011rescuing,allesina2012stability,jacquet2016no,yeakel2020diverse,mazzarisi2024complexity,lechon2024robust}. 
For example, mutualistic ecological networks often exhibit modular, highly connected structures that promote coexistence, whereas antagonistic interactions typically stabilize weakly coupled communities \cite{thebault2010stability}.
Here, we {applied} our framework to this ecological context and {showed} that disorder\textemdash now arising from irregular connectivity patterns, including (weighted and signed) degree heterogeneity and interaction asymmetry\textemdash often increases the robustness of ecosystems against perturbations.

Consider the generalized LV model describing the population dynamics of $N$ interacting species:
\begin{equation}
    \dot{x}_i = x_i \Bigg(b_i + \sum_{j=1}^N A_{ij} x_j\Bigg), \quad {\mathrm{for}} \,\,\, i = 1,\ldots,N,
\label{eq.lvmodel}
\end{equation}

\noindent
where $x_i$ denotes the abundance and $b_i$ the intrinsic growth rate of species $i$. When the adjacency matrix $A$ is invertible, the system admits the equilibrium $\bm x^{\mathrm{eq}} = - A^{-1}\bm b$. We refer to this equilibrium as {feasible} if all abundances are nonnegative (i.e., $x_i\geq 0$, $\forall i$); a coexistence equilibrium corresponds to all abundances being strictly positive. Linearizing Eq.~\eqref{eq.lvmodel} around $\bm x^{\mathrm{eq}}$ yields the Jacobian matrix $J = XA$, where $X = \operatorname{diag}(x_1^{\mathrm{eq}}, \ldots, x_N^{\mathrm{eq}})$. The state stability depends on the balance between excitatory and inhibitory interactions \cite{martinez2024stabilization,cure2023antagonistic}. 
Specifically, an interaction between species $i$ and $j$ is defined as {mutualistic} if $A_{ij}, A_{ji} > 0$ (e.g., pollinators and flowering plants), {competitive} if $A_{ij}, A_{ji} < 0$ (e.g., barnacles and mussels competing for substrate), and {antagonistic} if $A_{ij} > 0$ and $A_{ji} < 0$ (e.g., predator-prey pairs).
To further examine the impact of disordered network structures, we {considered} (random) perturbations $\delta A$ to the interaction matrix given by $A(\sigma) = A_0 + \sigma \delta A$, where $\sigma$ controls the disorder strength, $A_0 = -\alpha I_N$ represents the uncoupled baseline, and $\alpha\geq 0$ is the decay rate. 
{Although the nodal dynamics of the LV model are 1D, we note that disorder-promoted stability is allowed in this setting, given that disorder is introduced through the network structure (rather than the nodal parameters). For the matrix perturbations $\delta A$ considered here, the Jacobian matrix $J$ is non-Hermitian and hence $\Lambda_{\rm max}$ is generally a nonconvex function of $A$. Therefore, network stability is not expected to favor configurations that exhibit permutation symmetries.}
In fact, using Gershgorin disc analysis, we can prove that, in strictly mutualistic systems, disorder-promoted stability is guaranteed for any perturbation $\delta A$ and sufficiently small $\sigma$ (SM, Theorem~\ref{thm.ecologicalnet}).

\vspace{3pt}

Fig.~\ref{fig.ecologicalnet} characterizes the impact of network disorder on ecological stability.
For each realization, the matrix perturbation $\delta A$ follows the (possibly sparse) structure of an Erdős–Rényi (ER) graph: off-diagonal entries $\delta A_{ij}$ are independently set as nonzero with probability $p_1$, and the nonzero weights are drawn according to 
$\delta A_{ij}\sim \vartheta|\mathcal N(0,1)|$, where $\vartheta=1$ with probability $p_2$ and $\vartheta=-1$ with probability $1-p_2$ (Fig.~\ref{fig.ecologicalnet}A). In this construction, small $p_1$ leads to sparse networks, whereas $p_2 \rightarrow 1$ ($p_2\rightarrow 0$) yields only mutualistic (competitive) interactions. 
The outcome of the analysis strongly depends on the interaction type (Fig.~\ref{fig.ecologicalnet}B). Disorder tends to destabilize the equilibria in competitive networks but strongly stabilizes mutualistic ones, leading to a fivefold improvement in $\Lambda_{\mathrm{max}}$ for the networks considered. Antagonistic networks ($p_2=0.5$) exhibit intermediate behavior: disorder initially promotes stability, but beyond a network-dependent critical $\sigma$, the equilibrium becomes unstable and infeasible. {Indeed, for both competitive and antagonistic networks, we {observed} in the numerical simulations that loss of feasibility always {preceded} loss of stability\textemdash where the former is in direct agreement with the theoretical prediction that feasibility implies stability in competitive LV networks \cite{lechon2024robust}.}  The coexistence equilibrium $\bm x^{\mathrm{eq}}$ is stable over a finite disorder range $\sigma\in[0,\sigma_{\mathrm{max}}]$ across all interaction types, and, moreover, disorder can promote coexistence in networks dominated by antagonistic and mutualistic interactions. The trends are robust against changes in edge density (Fig.~\ref{fig.ecologicalnet}C) and the balance between excitatory and inhibitory interactions  (Fig.~\ref{fig.ecologicalnet}D).  

\vspace{3pt}

Nonlinear growth effects have been proposed as a resolution to May's longstanding complexity-stability paradox \cite{allesina2012stability,meena2023emergent,chen2024stability}, which arises from the theoretical expectation that increasing the complexity of an ecological system (i.e., network size and edge density) would lead to less stability despite empirical evidence on the contrary \cite{may1972will}. Our results offer a new perspective on this problem, showing that the likelihood of disorder-promoted stability increases both with network size and edge density when interactions are predominantly mutualistic (Fig.~\ref{fig.ecologicalnet}E). These results {suggest} that, in the ensemble of disordered networks, systems with increased complexity and stability are available to be selected for by evolution.

\section*{Discussion}
\label{sec.discussion}

In 1972, Philip Anderson popularized the view that the behavior of complex systems cannot be understood through reductionist approaches, as new properties and organizing principles emerge through their interactions \cite{anderson1972more}. Following this insight, high-dimensional models with complex network topologies but simplified nodal dynamics\textemdash including Boolean networks, spin-based agent models, and Kuramoto oscillator networks\textemdash have been widely used to study emergent phenomena and phase transitions. Despite their success, such models can lead to fundamentally different conclusions about the relationship between the stability of emergent states and system parameters. Near a stable limit cycle, models of weakly coupled oscillators can generally be reduced to the dynamics of first-order oscillators \cite{nakao2016phase,omel2025phase}. Crucially, this phase reduction omits the (potentially stabilizing) mode-mixing terms that arise from nodal heterogeneities in the full high-dimensional oscillators.
In this study, we {showed} that dimensionally reduced models that neglect the non-Hermitian structure of the Jacobian matrix fail to capture the stabilizing role of heterogeneity. 
This point is well illustrated through the example of Janus oscillators, where each oscillator consists of two Kuramoto oscillators \cite{nicolaou2019multifaceted}; in this case, because each node is effectively 2D, heterogeneity {can} stabilize synchronization.

Even though synchronization dynamics have been a major focus of research concerning the impact of disorder in complex systems, here we {addressed} this problem in a more general setting and showed that disorder can stabilize a wider range of collective states. These states include fully symmetric states (such as synchronization and consensus), partially symmetric states (such as opinion polarization, associative-memory patterns, and chimera states), and fully asymmetric states (such as species coexistence). Within our framework, the symmetries of the state are implicitly encoded in the structure of the Jacobian matrix, whose (non-) Hermiticity determines the (non-)convexity of the optimization problem for the  objective function $\Lambda_{\rm max}(J)$. These relationships allow us to assess whether stability can be promoted by breaking system symmetries either through {nodal heterogeneity} or {network heterogeneity} (see SM, Table~\ref{tab.summary}, for a summary of our theoretical contributions). In undirected networks, the existence of optimal symmetry-broken nodal parameter configurations requires the nodal dynamics to be 2D or higher.
These conclusions are based on the stability analysis of a single, prespecified dynamical state and are quantified using the largest transverse Lyapunov exponent. 
The same framework can be extended to other objective functions of interest, including those designed to modify the stability region and/or the state of interest\textemdash for instance, to broaden the parameter regime over which critical behavior persists \cite{sanchez2023heterogeneity}, to increase the stability region in MSF analysis \cite{barahona2002synchronization}, or to induce frequency synchronization between clusters of oscillators \cite{ocampo2025frequency}.
Certain objective functions can be nonconvex even in networks of symmetrically coupled 1D oscillators, allowing for asymmetric solutions irrespective of the oscillator dimensionality \cite{ocampo2025frequency}.

Our results should not, however, be construed as implying that arbitrary heterogeneity promotes stability. Indeed, nodal and network heterogeneity have been shown to be destabilizing in many contexts, including pattern formation, phase transitions, and synchronization \cite{nakao2010turing,fruchart2021non,pando2024synchronization}. These observations are not in contradiction with this study, as our framework identifies conditions under which certain heterogeneous configurations improve stability, whereas others do not. For instance, Fig.~\ref{fig.arnold}C,D shows that the regimes in which parameter mismatches promote stability are smaller than regimes in which they promote instability, with sufficiently large heterogeneities always leading to instability. Moreover, when stability is promoted by disorder, the effect typically occurs over finite ranges of disorder strength (e.g., see Fig.~\ref{fig.disorder}B). 
This perspective has a direct connection to parameter optimization and related research on the optimal allocation of damping, inertia, and network weights in engineered systems, as extensively studied in the literature of power grids \cite{sun2021stability,fritzsch2024stabilizing}, multi-agent systems \cite{sarlette2009consensus,asadi2017distributed}, and other second-order systems \cite{nair2008stable}. 
For these systems, our framework reveals that optimized configurations of nodal heterogeneity need not arise as compensation for heterogeneity in network structure, given that heterogeneous optimal nodal parameters emerge even in fully symmetric (thus homogeneous) networks. 
Additional examples of disorder-promoted stability may be found in systems beyond those described by Eq.~\eqref{eq.generalnetwork}, including those with external periodic driving \cite{lucas2019nonautonomous,huang2026optimizing}, nonsmooth dynamics \cite{liu2024heterogeneity}, and higher-order interactions \cite{wang2026frequency} (see SM, Example~\ref{examp.sm.higherorder}, for a brief discussion on the latter). Even though we focused on equilibrium dynamics throughout the text, we emphasize that disorder-promoted stability can also occur in systems with periodic, chaotic, and multistable dynamics (SM, Sec.~\ref{sec.sm.spatiotemp}); in multistable systems, enhanced linear stability can be accompanied by a concurrent enlargement of the basin of attraction, thereby increasing resilience to perturbations.

Finally, we suggest that our results are relevant for interpreting
recent empirical and theoretical studies on the constructive role of
heterogeneity in living systems.
For example, it has been shown that interindividual differences drive schooling and consensus in animal groups \cite{jolles2020role,doering2022noise}, social heterogeneity promotes depolarization  \cite{ojer2025social}, and cell heterogeneity shapes self-organized patterns in developmental processes \cite{mishra2026geometry} and brain computation  \cite{padmanabhan2010intrinsic,dahmen2025heterogeneity}.
In this broad context of nonequilibrium collective dynamics, our findings can thus inspire new strategies for artificial systems to optimize network stability in drone flocking \cite{montanari2025optimal}, chaos control \cite{braiman1995taming}, weakly-driven resonance \cite{tessone2006diversity,patriarca2025dynamical},  pattern formation \cite{teng2022heterogeneity}, and deep learning \cite{ziyin2025parameter}, among other applications.

\begin{footnotesize} 


\end{footnotesize}


\section*{Acknowledgments}
\begin{small}
The authors thank Jorin Graham for insightful discussions on ecological networks.
\textbf{Funding:}
This work is supported by the Army Research Office (MURI Grant No.\ W911NF-22-2-0109) and the National Science Foundation (Grant No.\ DMS-2308341).
The authors also acknowledge the stimulating research environment provided by the NSF-Simons National Institute for Theory and Mathematics in Biology (NSF Grant No.\ DMS-2235451 and Simons Foundation Grant No.\ MP-TMPS-00005320).
\textbf{Author contributions:}
A.N.M., P.Z., and A.E.M. designed the research; A.N.M. and P.Z. developed the theory, performed the numerical simulations, and analyzed the data; A.N.M. led the writing of the manuscript; A.N.M., P.Z., and A.E.M. contributed to the interpretation of the results and editing of the final version of the paper.
\textbf{Competing interests:}
None declared.
\textbf{Data and materials availability:}
Code and data are available through our Zenodo repository \cite{papercode}.
All (other) data needed to evaluate the conclusions in the paper are present in the paper or the Supplementary Materials.
\end{small}

\begin{small}
\medskip\noindent
\section*{Supplementary Materials}
Materials and Methods; Supplementary Text (Sections S1 to S5); Figs. S1 to S8; Tables S1 to S3; References [98]--[126].
\end{small}

\nolinenumbers
\onecolumn

\setcounter{equation}{0}
\setcounter{figure}{0}
\setcounter{table}{0}
\setcounter{page}{1}
\setcounter{section}{0}
\makeatletter
\renewcommand{\theequation}{S\arabic{equation}}
\renewcommand{\thefigure}{S\arabic{figure}}
\renewcommand{\thesection}{S\arabic{section}}
\renewcommand{\thetable}{S\arabic{table}}
\renewcommand{\thepage}{S\arabic{page}}

\noindent
\textbf{\LARGE Materials and Methods} 

\bigskip\noindent
\textbf{Numerical optimization and model parameters}

\noindent
For each network system listed in Table~\ref{tab.systems}, we numerically solve both the unconstrained optimization problem \eqref{eq.optimization} and its homogeneously constrained counterpart (subject to $\bm b = b\bm 1_N$) in order to determine the optimal parameter configurations $\bm b^*$ and $\bm b_{\rm hom}^*$, respectively.
%
%
To improve convergence towards (near) global optima, we employ a Quasi-Newton method, implemented by the MATLAB function \texttt{fminunc}, over 100 random initial conditions $b_i\sim\mathcal N(\bm b_{\mathrm{hom}}^*,1)$, and then select the best solution across all trials.
 The resulting distribution of parameter values $\bm b^* = (b_1^*,\ldots,b_N^*)$ is shown in Fig.~\ref{fig.disorder}A for each network model. 

The Jacobian $J$ of each system is also reported in Table~\ref{tab.systems}. In Fig.~\ref{fig.disorder}A, we consider undirected all-to-all networks with i.i.d.~random weights drawn from a uniform distribution (i.e., $A_{ij}\sim\mathcal U[0,1]$). Under this choice, Proposition~\ref{cor.sm.localcurvature} ensures that nonconvexity holds with probability $1$ for all the systems with 2D nodal dynamics, confirming that our numerical experiments probe the generic case.
All other model parameters, except for the optimization parameter $\bm b$, are fixed and identical across nodes. Specifically, we set $\omega = 0$ in the first-order leaky Kuramoto model, $P = 0$ in the second-order Kuramoto model, $k = 0.5$ in the mass-spring-damper network, and $(\omega,\varepsilon) = (0, 0.1)$ in the phase-amplitude oscillator. In general, we examine stability at the equilibrium point $\bm x^{\mathrm{eq}} = 0$. For the Josephson junction model, we set $(I_b,I_c) = (1,1.5)$ to ensure the existence of this equilibrium. In the FitzHugh-Nagumo model, we set $(a,\tau)=(1.5,1)$ as it ensures the existence of a nontrivial equilibrium $\bm x^{\mathrm{eq}}\neq 0$; this equilibrium is computed numerically and, for $\epsilon>20$, it undergoes a bifurcation and ceases to exist (Fig.~\ref{fig.disorder}A, yellow curve).

\begin{table*}[h]
    \footnotesize
    \centering
    \caption{\label{tab.systems} \textbf{Network models of coupled dynamical systems across different domains.}}
    \vspace{-0.3cm}
    \addtolength{\tabcolsep}{-0.2em}
    \begin{tabular}{l l l l l}
        \toprule[1.5pt]
        & \textbf{System} & \textbf{Equations of motion} & \textbf{Jacobian matrix} & \textbf{Quantities} \\
        \toprule[1.5pt]

        \multirow{7}{*}{\rotatebox[origin=c]{90}{\parbox{2.5cm}{\centering first-order models}}}
        & \multirow{2}{*}{%
            \begin{tabular}[t]{@{}l@{}}
                Leaky Kuramoto \\ \scriptsize{(synchronization)}
            \end{tabular}
        }
        & \multirow{2}{*}{$\dot \phi_i + b_i\phi_{i} = \omega + \sum_j A_{ij}\sin(\phi_j - \phi_i)$}
        & \multirow{2}{*}{$-(B+L)$}
        & \multirow{2}{*}{%
            \begin{tabular}[t]{ll}
                $\phi_i$ & oscillator phase \\
                $\omega$ & natural frequency
            \end{tabular}
        }\\
        & & & & \\
        \cmidrule(l){2-5}
        & \multirow{3}{*}{\begin{tabular}[t]{@{}l@{}}
                Driven Kuramoto \\ \scriptsize{(pacemaker)}
            \end{tabular}}
        & \multirow{3}{*}{$\dot{\phi}_i + b_i \sin(\phi_i - \psi) =  \omega + \sum_j A_{ij}\sin(\phi_j - \phi_i)$}
        & \multirow{3}{*}{$-(B+L)$}
        & \multirow{3}{*}{%
            \begin{tabular}[t]{ll}
                $\phi_i$ & oscillator phase \\
                $\omega$ & natural frequency \\
                $\psi$ & reference phase
            \end{tabular}
        }\\
        & & & & \\
        & & & & \\
        \cmidrule(l){2-5}
        & \multirow{2}{*}{\begin{tabular}[t]{@{}l@{}}
                Consensus \\ \scriptsize{(opinion, coordination)}
            \end{tabular}}
        & \multirow{2}{*}{$\dot x_i + b_i x_i = \sum_j A_{ij}(x_j - x_i)$}
        & \multirow{2}{*}{$-(B+L)$}
        & \multirow{2}{*}{%
            \begin{tabular}[t]{ll}
                $x_i$ & opinion/state
            \end{tabular}
        }\\
        & & & & \\
        
        \addlinespace[3pt]
        \toprule[1.5pt]
        \addlinespace[3pt]

        \multirow{11}{*}{\rotatebox[origin=c]{90}{\parbox{4cm}{\centering second-order models}}}
        & \multirow{2}{*}{%
            \begin{tabular}[t]{@{}l@{}}
                2nd-order Kuramoto \\ \scriptsize{(power grids)}
            \end{tabular}
        }
        & \multirow{2}{*}{$\ddot{\phi}_i + b_i \dot{\phi}_i = P + \sum_j A_{ij} \sin(\phi_j - \phi_i)$}
        & \multirow{2}{*}{$\begin{bmatrix} 0_N &\!\!\!\! I_N \\ - L &\!\!\!\! - B \end{bmatrix}$}
        & \multirow{2}{*}{%
            \begin{tabular}[t]{ll}
                $\phi_i$ & oscillator phase \\
                $P$ & power generation
            \end{tabular}
        }\\
        & & & & \\
        \cmidrule(l){2-5}
        & \multirow{2}{*}{\begin{tabular}[t]{@{}l@{}}
                van der Pol \\ \scriptsize{(circuits, vibration)}
            \end{tabular}}
        & \multirow{2}{*}{$\ddot{x}_i - b_i(1 - x_i^2)\dot{x}_i + x_i = \sum_j A_{ij}(x_j - x_i)$}
        & \multirow{2}{*}{$\begin{bmatrix} 0_N &\!\!\!\! I_N \\ -(L+I_N) &\!\!\!\! -B \end{bmatrix}$}
        & \multirow{2}{*}{%
            \begin{tabular}[t]{ll}
                $x_i$ & displacement
            \end{tabular}
        }\\
        & & & & \\
        \cmidrule(l){2-5}
        & \multirow{2}{*}{%
            \begin{tabular}[t]{@{}l@{}}
                Spring-mass system \\ \scriptsize{(materials)}
            \end{tabular}
        }
        & \multirow{2}{*}{$\ddot{x}_i + b_i \dot{x}_i + k x_i = \sum_j A_{ij}(x_j - x_i)$}
        & \multirow{2}{*}{$\begin{bmatrix} 0_N &\!\!\!\! I_N \\ - (L + kI_N) &\!\!\!\! - B \end{bmatrix}$}
        & \multirow{2}{*}{%
            \begin{tabular}[t]{ll}
                $x_i$ & displacement \\
                $k$ & stiffness
            \end{tabular}
        }\\
        & & & & \\
        \cmidrule(l){2-5}
        & \multirow{3}{*}{\begin{tabular}[t]{@{}l@{}}
                Josephson junction \\ \scriptsize{(superconductors)}
            \end{tabular}}
        & \multirow{3}{*}{%
            \begin{tabular}[t]{l}
                $\ddot{\phi}_i + b_i \dot{\phi}_i + I_{\mathrm{c}} \sin(\phi_i) =  I_{\mathrm{b}} + \sum_j A_{ij}(\dot \phi_j - \dot \phi_i)$
            \end{tabular}
        }
        & \multirow{3}{*}{$\begin{bmatrix} 0_N &\!\!\!\! I_N \\ -\sqrt{I_{\mathrm{c}}^2 - I_{\mathrm{b}}^2} I_N &\!\!\!\! - (L+B) \end{bmatrix}$}
        & \multirow{3}{*}{%
            \begin{tabular}[t]{ll}
                $\phi_i$ & junction phase \\
                $I_{\mathrm{c}}$ & critical current \\
                $I_{\mathrm{b}}$ & bias current
            \end{tabular}
        }\\
        & & & & \\
        & & & & \\
        \cmidrule(l){2-5}
        & \multirow{2}{*}{%
            \begin{tabular}[t]{@{}l@{}}
                Multi-agent system \\ \scriptsize{(flocking, navigation)}
            \end{tabular}
        }
        & \multirow{2}{*}{%
            \begin{tabular}[t]{l}
                $\ddot{x}_i + b_i (\dot{x}_i - \dot x_{\mathrm{t}}) + b_i(x_i - x_{\mathrm{t}}) =$\\
                $\,\,\,\sum_j A_{ij}[(x_j - x_i) + (\dot x_j - \dot x_i)]$
            \end{tabular}
        }
        & \multirow{2}{*}{$\begin{bmatrix} 0_N &\!\!\!\! I_N \\ -(L+B) &\!\!\!\! - (L+B) \end{bmatrix}$}
        & \multirow{2}{*}{%
            \begin{tabular}[t]{ll}
                $x_i$ & agent position \\
                $x_{\mathrm{t}}$ & target position
            \end{tabular}
        }\\
        & & & & \\

        \addlinespace[3pt]
        \toprule[1.5pt]
        \addlinespace[3pt]

        \multirow{9}{*}{\rotatebox[origin=c]{90}{\parbox{3.4cm}{\centering two-dimensional models}}}
        & \multirow{6}{*}{%
            \begin{tabular}[t]{@{}l@{}}
                FitzHugh-Nagumo \\ \scriptsize{(neurons)}
            \end{tabular}
        }
        & \multirow{6}{*}{%
            \begin{tabular}[t]{l}
                $\,\,\,\,\,\dot v_i = v_i - \frac{v_i^3}{3} - w_i + \sum_j A_{ij}(v_j-v_i)$ \\
                $\tau\dot w_i = v_i + a - b_i w_i$
            \end{tabular}
        }
        & \multirow{6}{*}{$\begin{bmatrix} - L + I_N - V^2 &\!\!\!\!\!\! - I_N \\ \tau^{-1} I_N &\!\!\!\!\!\! - \tau^{-1} B \end{bmatrix}$}
        & \multirow{6}{*}{%
            \begin{tabular}[t]{ll}
                $w_i$ & neuron activator \\
                $v_i$ & neuron inhibitor \\
                $\tau$ & timescale separation \\
                $a$ & excitability threshold \\
                $v_i^{\mathrm{eq}}$ & equilibrium \\
                $V$ & $\operatorname{diag}({v_1^{\mathrm{eq}}},\ldots,v_N^{\mathrm{eq}})$
            \end{tabular}
        }\\
        & & & & \\
        & & & & \\
        & & & & \\
        & & & & \\
        & & & & \\
        \cmidrule(l){2-5}
        & \multirow{4}{*}{%
            \begin{tabular}[t]{@{}l@{}}
                Phase-amplitude \\ oscillator \\ \scriptsize{(lasers, circuits)}
            \end{tabular}
        }
        & \multirow{4}{*}{%
            \begin{tabular}[t]{l}
                $\dot r_i = b_i r_i (1-r_i) + \varepsilon r_i\sum_j A_{ij} \cos(\phi_j - \phi_i)$ \\
                $\dot\phi_i = \omega + r_i - 1 + r_i\sum_j A_{ij} \sin(\phi_j - \phi_i)$
            \end{tabular}
        }
        & \multirow{4}{*}{$\begin{bmatrix} - B &  -\varepsilon A \\ I_N & - L\end{bmatrix}$}
        & \multirow{4}{*}{%
            \begin{tabular}[t]{ll}
                $r_i$ & oscillator amplitude \\
                $\phi_i$ & oscillator phase \\
                $\omega$ & natural frequency \\
                $\varepsilon$ & coupling strength \\                
            \end{tabular}
        }\\
        & & & & \\
        & & & & \\
        & & & & \\
        \addlinespace[3pt]
        \bottomrule[1.5pt]
    \end{tabular}
    \begin{minipage}{\textwidth}
    \raggedright
    \vspace{2pt}
    \footnotesize
    $\,\,\,\,$ Across all systems, $B=\operatorname{diag}(b_1,\ldots,b_N)$ and $L$ is the Laplacian matrix associated with the adjacency matrix $A$.
\end{minipage}
\end{table*}

\bigskip\noindent
\textbf{Complex network models and datasets}

\noindent
For the generation of complex network models used in Fig.~\ref{fig.disorder}B--D, we considered Watts-Strogatz SW networks with parameters $p \in [0, 1]$ and $k=2$ as well as Barabási-Albert SF networks with parameter $m\in\{2,3,4,5\}$. Here,  $p$ is the rewiring probability, $k$ is the number of nearest neighbors in a ring graph, and $m$ is the number of edges of each node iteratively attached to the network. 
The circulant matrices in Fig.~\ref{fig.disorder}D represent directed first-neighbor ring networks, constructed using the generating vector $[1,-\delta,0,...,0,-1+\delta]$, with $\delta\sim\mathcal U[0,1]$ (see Definition~\ref{def.sm.circulant} in Section~\ref{sec.sm.circulant}).
For the empirical networks in Fig.~\ref{fig.disorder}E, we considered several adjacency matrices available in public
datasets (reported in Table~\ref{tab.realworldnets}).
For each model and empirical network, we generated a directed weighted network by independently assigning weights $A_{ij}\sim\mathcal U[0,1]$ for every directed edge $(i,j)$. Fig.~\ref{fig.sm.isotropy} provides a characterization of the eigenvector localization in directed SW and SF networks.

\begin{table}[t!]
\footnotesize
\centering 
\caption{\label{tab.realworldnets} \textbf{Datasets of empirical networks.}}
\vspace{-0.3cm}
\begin{tabular}{llrrl}
    \toprule[2pt]
    \textbf{Type} & \textbf{Name} & \textbf{Vertices}  & \textbf{Edges} & \textbf{Description}\\
    \toprule[2pt]
    Neuronal 
        & \textit{C. elegans} \cite{Kunegis2013} & $297$ & $2{,}345$ & Neuronal network of \textit{C. elegans}. \\
    \midrule[0.5pt]
    Ecological 
        & Florida wetlands \cite{clauset2016icon} & $128$ & $2{,}106$ & Food web in the wetlands of South Florida. \\
        & Little Rock Lake \cite{Kunegis2013} & $183$ & $2{,}494$ & Food web in Little Rock Lake. \\
    \midrule[0.5pt]
    Regulatory 
        & TRN-EC-2 \cite{Milo2002} & $418$ & $519$ & Transcriptional network of \textit{E. coli}. \\
        & \textit{C. elegans} \cite{Kunegis2013} & $453$ & $2{,}040$ & Metabolic network of \textit{C. elegans}. \\
    \midrule[0.5pt]
    Traffic  
        & FAA routes \cite{clauset2016icon} & $1{,}226$ & $2{,}615$ & Air traffic routes. \\
        & US agencies \cite{clauset2016icon} & $1{,}127$ & $5{,}480$ & Website traffic among government agencies. \\
    \midrule[0.5pt]
    Electronics & s208 \cite{Milo2002}  & $122$ & $189$ & Sequential logic circuit. \\ 
			& s838 \cite{Milo2002}  & $512$ & $819$ & Sequential logic circuit. \\ 
    \midrule
    Social 
        & Macaques \cite{clauset2016icon} & $62$ & $1{,}188$ & Dominance interactions among Japanese macaques. \\
        & Email company \cite{clauset2016icon} & $167$ & $82{,}928$ & Emails exchanged among employees of a manufacturing company. \\
\toprule[2pt]
\end{tabular}
\end{table}

\bigskip\noindent
\textbf{Ecological networks generation}

\noindent
In the ensemble of LV networks analyzed in Fig.~\ref{fig.ecologicalnet}, the growth rates are independently drawn from a uniform distribution, $b_i\sim\mathcal U[1,11]$. The adjacency matrix is parameterized as $A = A_0 + \sigma\delta A$, where $A_0 = \alpha I_N$, $\alpha = -1$ is the self-regulation strength, and $\sigma$ controls the disorder strength. For each realization, the perturbation $\delta A$ encodes structural disorder through three factors: edge probability $p_1$, sign probability $p_2$, and weight heterogeneity (drawn from a half-normal distribution). Thus, $\delta A_{ij} \sim +|\mathcal N(0,1)|$ with probability  $p_1 p_2$, $\delta A_{ij} \sim -|\mathcal N(0,1)|$ with probability $p_1(1-p_2)$, and $\delta A_{ij} = 0$ otherwise.


\newpage
\noindent
\textbf{\LARGE Supplementary Text}

\setcounter{tocdepth}{2}
\renewcommand{\contentsname}{}
\tableofcontents


\bigskip\noindent
\textbf{Outline.} The Supplementary Text is organized as follows. 
Section~\ref{sec.sm.arnold} provides analytical conditions for the synchronizability of different oscillator models, characterizing the stability regions (i.e., the Arnold tongues) in Fig.~\ref{fig.arnold}.
Section~\ref{sec.sm.jensen} proposes a theorem relating the symmetry of the optimal parameter vector to the convexity of the stability optimization problem. We then extend this theorem to establish an analogous relationship between the symmetry of the optimal network structure and the problem convexity. The applications of these results to several dynamical systems of interest are also discussed.
Section~\ref{sec.sm.dpsconditions} provides conditions for the existence and prevalence of disorder-promoted stability. Here, our focus is restricted to dynamical systems described by the Jacobian matrix \eqref{eq.jacobian} (e.g., power grids and spring-mass-damper systems) due to their analytical tractability and importance in physics. We introduce theorems that characterize the behavior of the eigenvalue spectrum and determine whether the globally optimum parameter configurations are located at homogeneous or heterogeneous points. Particular attention is given to circulant networks, which comprise a broad class of networks with an automorphism with full symmetry.
Section~\ref{sec.sm.modemix} builds on the theorems proposed in Secs.~\ref {sec.sm.jensen} and \ref{sec.sm.dpsconditions} to show that mode mixing is the fundamental mechanism underlying disorder-promoted stability. 
{Finally, Sec.~\ref{sec.sm.spatiotemp} presents numerical examples showing that disorder-promoted stability can also arise in broader classes of dynamical states other than equilibrium states, including periodic, chaotic, and multistable regimes. Table~\ref{tab.summary} summarizes the main theoretical contributions.
}

\newcommand{\vlabel}[2]{%
    \rotatebox[origin=c]{90}{%
        \begin{tabular}{@{}c@{}}
        \textbf{#1}\\[-1pt]
        \textbf{#2}
        \end{tabular}%
    }%
}

\begin{table*}[h!]
\footnotesize
\centering
\caption{\textbf{{Summary of main theoretical statements.}} \label{tab.summary}}
\vspace{-0.25cm}

\renewcommand{\arraystretch}{1.35}
\setlength{\tabcolsep}{4pt}

\begin{tabularx}{\textwidth}{
    >{\centering\arraybackslash}m{1cm}
    >{\raggedright\arraybackslash}X
    >{\raggedright\arraybackslash}m{2.25cm}
    >{\raggedright\arraybackslash}m{3cm}
}
\toprule[2pt]
\textbf{Type}
&
\textbf{Theoretical result}
&
\textbf{Statement}
&
\textbf{Section (main text)}
\\
\midrule[1.5pt]

\multirow{13}{*}{\vlabel{nodal}{heterogeneity}}
&
Convex stability landscape $\implies$ at least one global optimum $\bm b^*$ satisfies symmetries.

\textit{Contrapositive:} Nonconvex landscape
$\impliedby$ every $\bm b^*$ breaks symmetries.
&
Theorem~\ref{thm.jensenineq}
&
\multirow{4}{3cm}{Emergence of stability out of asymmetry}
\\

\cmidrule(lr){2-3}

&
Non-Hermitian Jacobian $\impliedby$ nonconvex stability landscape.
&
Proposition~\ref{cor.sm.nonhermitian}
&
\\

\cmidrule(lr){2-3}

&
Broad class of non-Hermitian Jacobians $\implies$ nonconvex stability landscape.
&
Proposition~\ref{cor.sm.localcurvature}
&
\\

\cmidrule(lr){2-4}
&
Undirected networks $\implies$ $\Lambda_{\rm max}$ is nondifferentiable at global optimum $\bm b^*$.

\textit{Contrapositive:} Network is directed $\impliedby$ $\Lambda_{\rm max}$ is differentiable at global optimum $\bm b^*$.
&
Proposition~\ref{prop.theor_hermit_diff}
&
\multirow{6}{3cm}{Disorder-promoted stability in complex networks}
\\

\cmidrule(lr){2-3}

&
Undirected networks $\implies$ $\Lambda_{\rm max}$ is nondifferentiable at homogeneous optimum $\bm b_{\rm hom}^*$.
&
Proposition~\ref{prop.nondifferentiablehom}
&
\\

\cmidrule(lr){2-3}

&
Broad class of directed circulant networks $\implies$
$\Lambda_{\rm max}$ is differentiable at $\bm b_{\rm hom}^*$.
&
Lemma~\ref{lem.circ_diff_ass}
&
\\

\cmidrule(lr){2-3}

&
Differentiable circulant networks $\implies$ $\bm b_{\rm hom}^*$ is a saddle point.

\textit{Consequently,} there exist small perturbations $\delta\bm b$ s.t. $\Lambda_{\rm max}(\bm b_{\rm hom}^*+\delta\bm b) < \Lambda_{\rm max}(\bm b_{\rm hom}^*)$.
&
Theorem~\ref{thm.local_CSB}
&
\\

\cmidrule(lr){2-4}
&
$\bm b$ is homogeneous and $L$ is Hermitian $\implies$ modes of $J(\bm b)$ are decoupled.
&
Proposition~\ref{prop.JJ'}
&
\multirow{3}{3cm}{A mechanistic interpretation}
\\

\cmidrule(lr){2-3}

&
$\bm b$ is a global optimum and $L$ is Hermitian $\implies$ modes of $J(\bm b)$ are coupled.

\textit{Moreover,} for heterogeneous $\bm{b}$, modes of $J(\bm{b})$ are coupled.
&
Theorem~\ref{thm.mode_mixing}
&
\\


\midrule[2pt]

\multirow{5}{*}{\vlabel{network}{heterogeneity}}
&
Convex stability landscape $\implies$ globally optimal network $A^*$ satisfies the symmetries.
&
Theorem~\ref{thm.sm.optadj}
&
\multirow{3}{3cm}{Nodal versus network heterogeneity}
\\

\cmidrule(lr){2-3}

&
The constrained stability optimization problem is nonconvex.

\emph{Consequently,} optimal networks $A^*$ may break the system symmetries.
&
Corollary~\ref{cor.sm.optadj_constrain}
&
\\

\cmidrule(lr){2-4}

&
Network interactions are mutualistic $\implies$ there exist small network perturbations $\delta A$ such that $\Lambda_{\rm max}(\delta A) < \Lambda_{\rm max}(0)$.
&
Theorem~\ref{thm.ecologicalnet}
&
{Stable biodiversity in disord.~ecological nets}
\\

\bottomrule[2pt]
\end{tabularx}
\end{table*}

\newpage
\noindent
\textbf{Notation.} 
Throughout the main and supplementary text, we adopt the following notation. Column vectors are represented by bold lower-case letters (e.g., $\bm x, \bm y, \bm z$), matrices are represented by capital letters (e.g., $X, Y, Z$), and sets are represented by calligraphic capital letters (e.g., $\mathcal X, \mathcal Y, \mathcal Z$). The $(i,j)$th entry of matrix $A$ is denoted by $A_{ij}$ or $(A)_{ij}$, and $x_i$ is used to denote the $i$th element of a vector $\bm x$. 
The cardinality (number of elements) of a set $\mathcal V$ is denoted by $\abs{\mathcal V}$.
The identity matrix of order $n$ is denoted by $I_n$ and $0_{n}$ is an $n\times n$ null matrix (the subscript is often omitted when it is self-evident from the context). 
The $n$-dimensional vectors of ones and zeros are denoted by $\bm 1_n$ and $\bm 0_n$, respectively. 
The diagonal matrix formed by the entries $a_1,\ldots,a_n$ is denoted by $\operatorname{diag}(a_1,\ldots,a_n)$.
The transpose of a matrix $A$ and of a vector $\bm v$ are denoted respectively by $A^\transp$ and $\bm v^\transp$, while their conjugate transpose are denoted by $A^{\dagger}$ and $\bm v^{\dagger}$. The complex conjugate of a scalar $z$ is denoted by $z^{\mathrm{c}}$. The square root of a complex number always refers to its principal value. The imaginary unit is denoted by $\mathrm{i}$.
A realization of a random variable drawn from a Gaussian distribution with mean $\mu$ and standard deviation $\sigma$ is denoted by $x\sim N(\mu,\sigma^2)$, and a realization of a random variable drawn from a uniform distribution in the interval $[a,b]$ is $x\sim\mathcal U[a,b]$.

\bigskip
\section{Arnold tongues of coupled heterogeneous oscillators}
\label{sec.sm.arnold}

We examine the stability of the dynamical systems used to generate the Arnold tongues in Fig.~\ref{fig.arnold}.

\medskip\noindent
\textbf{Kuramoto oscillators.} Consider a pair of classical Kuramoto oscillators,
\begin{equation}
    \dot{\phi}_i = \omega_i + K\sin(\phi_j-\phi_i), \quad\quad \text{for} \,\,\, i,j=1,2,
\label{eq.sm.kuramotopair}
\end{equation}

\noindent
where $\phi_i\in\mathbb S^1$ is the phase of oscillator $i$, $\omega_i$ is the oscillator's natural frequency, and $K$ is the coupling strength. The dynamics of the phase difference $\delta = \phi_1-\phi_2$ are given by $\dot\delta = \Delta\omega - 2K\sin(\delta)$, where $\Delta\omega = \omega_1-\omega_2$ measures the parameter mismatch. Linearization of this system around the equilibrium $\delta^{\mathrm{eq}} = \sin^{-1}(\Delta\omega/2K)$ allows us to express Eq.~\eqref{eq.sm.kuramotopair} as
\begin{equation*}
    \dot{\delta} = -2K\cos(\delta^{\mathrm{eq}}) \delta + \text{higher-order terms}.
\end{equation*}
It thus follows that the synchronized state $\delta^{\mathrm{eq}}$ is stable if and only if $|\Delta\omega/2K|< 1$, as illustrated in the Arnold tongue in Fig.~\ref{fig.arnold}A.

Notably, this result can be generalized to arbitrary networks of $N$ coupled Kuramoto oscillators:
\begin{equation}
    \dot{\phi}_i = \omega_i + K\sum_{j=1}^N A_{ij}\sin(\phi_j-\phi_i), \quad\quad \text{for}\,\,\, i = 1,\ldots,N,
\label{eq.sm.kuramotonet}
\end{equation}
where $A \in \R^{N \times N}$ is the adjacency matrix of the underlying graph $\mathcal{G} = (\mathcal{V}, \mathcal{E})$, with a set of nodes $\mathcal{V} = \{1,\ldots,N\}$ and a set of edges $\mathcal{E} = \{(i,j): A_{ij} \neq 0\}$.
As shown in Ref.~\cite{Dorfler2013}, the system is phase synchronized (i.e., $|\phi_i - \phi_j|<\gamma\leq \frac{\pi}{2}$ for all pairs and time $t\geq T$) if 
\begin{equation*}
    \norm{KL^+\bm\omega}_{\mathcal E,\infty}\leq\sin(\gamma),
\end{equation*}
where $L^+$ is the pseudoinverse of the corresponding Laplacian matrix $L$, $\bm\omega = (\omega_1, \ldots, \omega_N)$ is the vector of natural frequencies, and $\|\bm x\|_{\mathcal E,\infty} = \max_{(i,j)\in\mathcal E} |x_i - x_j|$ measures the maximum dissimilarity between pairs of oscillators connected by an edge. Clearly, when all oscillators are sufficiently similar, $\norm{KL^+\bm\omega}_{\mathcal E,\infty}\approx 0$ and the synchronized state is stable.

\medskip\noindent
\textbf{First-order leaky Kuramoto oscillators.}
Consider the ``leaky'' version of the classical Kuramoto model,
\begin{equation}
\dot{\phi}_{i} + b_{i}\phi_{i} = \omega + K\sum_{j=1}^{N} A_{ij} \sin(\phi_{j}-\phi_{i}), 
\label{eq.sm.leakykuramoto}
\end{equation}
where $b_i$ is a leak rate (or damping coefficient). In the co-rotating frame, the Jacobian matrix is given by $J = -(B+KL)$ at equilibrium $\phi_i^{\mathrm{eq}} = 0$, $\forall i$, where $B = \operatorname{diag}(b_1,\ldots,b_N)$. When the adjacency matrix has only nonnegative entries and zero diagonal entries, $J$ is an M-matrix. It is thus evident that homogeneously increasing $b_i$ is sufficient to minimize  $\Lambda_{\mathrm{max}}$ globally (as shown in Fig.~\ref{fig.arnold}B), regardless of the network structure and the number of oscillators $N$.

\medskip\noindent
\textbf{Second-order Kuramoto oscillators.} Consider the second-order Kuramoto model
\begin{equation}
\ddot{\phi}_{i} + b_{i}\dot{\phi}_{i} = \omega_{i} + K\sum_{j=1}^{N} A_{ij} \sin(\phi_{j}-\phi_{i}),
\label{eq.sm.2ndkuramoto}
\end{equation}

\noindent
where $b_i$ is the damping coefficient. For nonidentical natural frequencies $\omega_i$, the Jacobian matrix evaluated at an equilibrium point $\bm\phi^{\mathrm{eq}} = (\phi_1^{\mathrm{eq}},\ldots,\phi_N^{\mathrm{eq}})$ is generally given by 
\begin{equation}
J=\begin{bmatrix}\,\,\,\,\,\,\,\,0_N&\,\,\,\,\,\,\,\,I_N\\-\Tilde{L}&-B\end{bmatrix}, \,\,\,\,\,\,\, \text{where} \,\,\, \Tilde{L}_{ij} = -KA_{ij}\cos(\phi_{i}^{\mathrm{eq}}-\phi_{j}^{\mathrm{eq}})+\delta_{ij}K\sum_{k=1}^{N}A_{ik}\cos(\phi_{i}^{\mathrm{eq}}-\phi_{k}^{\mathrm{eq}}).
\label{eq.sm.jacobian_second_kur}
\end{equation}

\noindent
In the case of a pair of coupled oscillators with reciprocal interactions (i.e., $N=2$ and $A_{12} = A_{21} = 1$), the eigenvalues of $J$ are determined by the following characteristic polynomial:
\begin{equation*}
\begin{aligned}
0&=\det(J-\lambda I_N), \\
&=\det(\lambda^2 I_N+\lambda B+ \Tilde L), \\
&=\det\begin{bmatrix}\lambda^2+\lambda b_{1}+l&-l\\-l&\lambda^2+\lambda b_{2}+l\end{bmatrix}, \\
\end{aligned}
\end{equation*}

\noindent
where $l = K\cos(\phi_1^{\mathrm{eq}} - \phi_2^{\mathrm{eq}})$, and we used the fact that the identity matrix commutes with the Laplacian matrix $L$. Expanding this determinant yields the quartic equation
\begin{equation}
0=\lambda\left(\lambda^{3}+2\bar{b}\lambda^2 +\lambda\left(\frac{4\bar{b}^2-\Delta b^2}{4}+2l\right)+2l\bar{b} \right),
\label{eq:2nd_ord_Kur}
\end{equation}

\noindent
where we define $\bar{b} = \frac{1}{2}(b_1+b_2)$ and $\Delta b = b_1 - b_2$. Thus, the largest Lyapunov exponent, as shown in Fig.~\ref{fig.arnold}C, is determined by a cubic polynomial. When the oscillators are uncoupled (i.e., $l=0$), we have an additional null eigenvalue $\lambda_2 = 0$, while the others are given by $\lambda_{3,4}=-\bar{b}\pm\frac{\Delta b}{2}$.
For small values of $l$, the eigenvalues $\{\lambda_2,\lambda_3,\lambda_4\}$ are slowly perturbed. If the system is overdamped (large $\bar b$),  increasing $\Delta b$ will force $\lambda_2$ and $\lambda_3$ to approach each other, yielding a smaller $\Lambda_{\mathrm{max}}$. On the other hand, for large values of $l$, the system is underdamped (small $\bar b$). Due to the linear term in Eq.~\eqref{eq:2nd_ord_Kur}, increasing $\Delta b$ is similar to decreasing $\bar{b}$, thus increasing $\Lambda_{\mathrm{max}}$.

\medskip\noindent
\textbf{Phase-amplitude oscillators.} A phase-amplitude oscillator model was introduced in Ref.~\cite{nishikawa2016symmetric} to show the existence of classes of oscillator networks in which synchronization can happen only in the presence of oscillator heterogeneity, even when the oscillators are identically coupled to the network. Here, we consider a modified version of this model that allows us to show the same phenomenon using only two coupled oscillators. For an arbitrary adjacency matrix $A\in\R^{N\times N}$, this oscillator network model is described as
\begin{equation}
\begin{aligned}
\dot{\phi}_{i} &=r_{i}-1 - (cb_i+\alpha)\phi_i + K\sum_{j=1}^N A_{ij}r_{i}\sin(\phi_{i}-\phi_{j}), \\
\dot{r}_{i} &= (b_{i}+\alpha)r_{i}(1-r_{i})-\epsilon\sin(\phi_{i}),
\end{aligned}
\label{eq.sm.phaseamplitude}
\end{equation}

\noindent
where $\phi_i\in\mathbb S^1$ and $r_i\in\R$ are the phase and amplitude of oscillators $i$, $K$ is the coupling strength, $b_i$ is the damping coefficient, and $\{\epsilon,\alpha\}$ are some positive constants. This system has a single equilibrium point $(\phi_{i}^{\mathrm{eq}},r_{i}^{\mathrm{eq}})=(0,1)$, $\forall i$, and the associated Jacobian is
\begin{equation}
J=\begin{bmatrix}
KL-cB-\alpha I_{N} & I_N \\-\epsilon I_N & -B-\alpha I_{N}
\end{bmatrix}.
\label{eq.sm.jacobianphaseamplitude}
\end{equation}

For $N=2$ oscillators and $A_{12}=A_{21} = 1$, we can show that a heterogeneous choice of damping coefficients $b_1\neq b_2$ is required to stabilize the synchronous (symmetric) state $(\phi_i^{\mathrm{eq}},r_i^{\mathrm{eq}})$ when $K>0$. In this case, the eigenvalues of the Jacobian matrix \eqref{eq.sm.jacobianphaseamplitude} are determined by the characteristic polynomial for 
\begin{equation*}
\begin{aligned}
%
0&=\det(\lambda^2 I_N+\lambda\big((1+c)B-KL\big)+\epsilon I_N-KLB+cB^2), \\
&=\det\begin{bmatrix}\lambda^2+\lambda\big((1+c)b_{1}-K\big)+\epsilon-Kb_{1}+cb_{1}^2&\lambda K+Kb_{2}\\\lambda K+Kb_{1}&\lambda^2+\lambda\big((1+c)b_{2}-K\big)+\epsilon-Kb_{2}+cb_{2}^2\end{bmatrix}, \\
\end{aligned}
\end{equation*}

\noindent
where we applied the transformation $\lambda\to\lambda+\alpha$ for simplicity. This determinant leads to the following quartic equation:
\begin{equation*}
\begin{aligned}
0 = & \,\, \lambda^4+2\lambda^3\big((1+c)\bar{b}-K\big)+\lambda^2\Big(2\epsilon-(4+2c)K\bar{b}+4c\bar{b}^2+\frac{(1+c^2)}{4}(4\bar{b}^2-\Delta b^2)\Big)
\\&+\lambda\Big(2\epsilon\bar{b}(1+c)-2K\epsilon-(2+4c)K\bar{b}^2+\frac{\bar{b}c(1+c)}{2}(4\bar{b}^2-\Delta b^2)\Big)
\\&+\epsilon^2-2\epsilon K\bar{b}+\frac{\epsilon c}{2}(4\bar{b}^2+\Delta b^2)-\frac{cK\bar{b}}{2}(4\bar{b}^2-\Delta b^2)+\frac{c^2}{16}(16\bar{b}^4-8\bar{b}^2\Delta b^2+\Delta b^4),
\end{aligned}
\end{equation*}

\noindent
where we once again define $\bar{b} = \frac{1}{2}(b_1+b_2)$ and $\Delta b = b_1 - b_2$. Note that the Arnold tongue in Fig.~\ref{fig.arnold}D has an unusual ``snake tongue'' shape precisely because the eigenvalue equation is described by a fourth-order polynomial whose quadratic, linear, and constant terms simultaneously depend on $\Delta b$, leading to a more complex stability region for large $K$. In the homogeneous case $\Delta b=0$, we can write
\begin{equation*}
\begin{aligned}
0&=\Big(\lambda^2+\lambda\big((1+c)\bar{b}-2K\big)+\epsilon+c\bar{b}^2-2\bar{b}K\Big)\Big(\lambda^2+\lambda(1+c)\bar{b} +\epsilon+c\bar{b}^2\Big),
\end{aligned}
\end{equation*}

\noindent
and hence the four eigenvalues of the system are
\begin{equation*}
\begin{aligned}
\lambda_{1,\pm}+\alpha=K-\frac{(1+c)\bar{b}}{2}\pm\sqrt{\frac{(1-c)^2\bar{b}^2}{4}+(1-c)K\bar{b}-K^2-\epsilon},\,\,\,\,\,\,\,\,\,\,\,\,\,\lambda_{2,\pm}+\alpha=-\frac{(1+c)\bar{b}}{2}\pm \sqrt{\frac{(1-c)^2\bar{b}^2}{4}-\epsilon}.
\end{aligned}
\end{equation*}

\noindent
Given that $K>0$,  $\mathrm{Re}(\lambda_{2,\pm})<\mathrm{Re}(\lambda_{1,\pm})$. Since $\mathrm{Re}(\lambda_{1,-})\leq\mathrm{Re}(\lambda_{1,+})$, it follows that  $\Lambda_{\rm max} = \mathrm{Re} (\lambda_{1,+})$. Thus, for $K=0$,

\begin{equation}
\lambda_{1,+}=-\alpha-\frac{(1+c)\bar{b}}{2}+\sqrt{\frac{(1-c)^2\bar{b}^2}{4}-\epsilon}.
\end{equation}

\noindent
Therefore, if

\begin{equation}
\sqrt{\frac{(1-c)^2\bar{b}^2}{4}-\epsilon}>\frac{(1+c)\bar{b}}{2}+\alpha,
\end{equation}

\noindent
then the system will be unstable for any $K>0$, as increasing $K$ will induce more instability. Our numerical results in Fig.~\ref{fig.arnold}D confirm this effect, which further shows that it is still possible to find a stable equilibrium in these parameter regions when $\Delta b\neq0$.

\medskip\noindent
\textbf{Parameter choices.}
Fig.~\ref{fig.arnold} shows the Arnold tongues for the following parameters: $\bar b = 0.6$ (leaky Kuramoto),  $\bar b = 0.75$ and $\omega_i = \omega$, $\forall i$ (second-order Kuramoto), and $(\bar b, \epsilon, \alpha, c) = (1.5,0.2,0.7,0.3)$ (phase-amplitude). The ranges of coupling strengths are normalized by $0.1$ in panels b and c and by $0.3$ in panel d.

\bigskip
\section{Relationship between system symmetries and stability}
\label{sec.sm.jensen}

\medskip
\subsection{Nodal heterogeneity and stability}
\label{sec.sm.jensen.nodal}

In this section, we derive a necessary condition under which optimal state stability occurs only at asymmetric parameter configurations. This general result connects the convexity of an objective function to the permutation symmetries of its parameters and is applicable to a broad class of dynamical systems.
We demonstrate the implications of this condition for optimizing nodal parameters in first-order Kuramoto oscillators, second-order Kuramoto oscillators, and phase-amplitude oscillators. 

We first establish the following general result:

\begin{lemma}
    Consider a continuous objective function $C(\bm y)$, where $\bm y\in\R^m$ is the vector of decision variables.
    {Let $\mathcal P$ be a finite group of matrices representing symmetries of $C$ such that $C(P\bm y)=C(\bm y)$ for every $\bm y$ and $P\bm y$ in the domain and every $P\in\mathcal P$.}
    If the objective function $C(\bm y)$ is convex, then at least one optimal solution $\bm y^*$ satisfies the symmetry relation $\bm y^* = P\bm y^*$ {for every $P\in\mathcal P$.}
\label{lem.sm.generaljensen}
\end{lemma}

\begin{proof}
{Let $\bm y^*$ be a global minimizer of $C$, and denote $C^* = C(\bm y^*)$. Since the objective function is invariant under all matrices $P\in\mathcal P$, we have that $C(P\bm y^*) = C(\bm y^*) = C^*$ for every $P\in\mathcal P$. Now, define the group-averaged point
\begin{equation}
    \bar{\bm y} = \frac{1}{\abs{\mathcal{P}}}\sum_{P\in\mathcal P} P\bm y^*,
\end{equation}
where $\abs{\cdot}$ is the number of elements in the set.
The point $\bar{\bm y}$ is a convex combination of points in the domain. Therefore, using Jensen’s inequality, we have that
\begin{equation}
    C(\bar{\bm y}) \leq \frac{1}{\abs{\mathcal{P}}}  \sum_{P\in\mathcal P} C(P\bm y^*) = \frac{1}{\abs{\mathcal{P}}} \sum_{P\in\mathcal P} C(\bm y^*) = C^*.
\label{eq.sm.inequalityderivation}
\end{equation}
Since $C^*$ is the global minimum value, we must also have that $C(\bar{\bm y})\geq C^*$. Therefore, $C(\bar{\bm y}) = C^*$, and $\bar{\bm y}$ is a global minimizer.
}

{
It remains to show that $\bar{\bm y}$ satisfies all symmetries. For any $Q\in\mathcal P$,
\begin{equation}
    Q\bar{\bm y} =  \frac{1}{\abs{\mathcal{P}}}\sum_{P\in\mathcal P} QP\bm y^*.
\end{equation}
Because $\mathcal P$ is a symmetry group, left multiplication by $Q$ only reorders elements of $\mathcal P$. Therefore,
\begin{equation}
    Q\bar{\bm y} = \frac{1}{\abs{\mathcal{P}}}\sum_{P\in\mathcal P} P\bm y^* = \bar{\bm y}.
\end{equation}
Thus, $\bar{\bm y}$ is a global minimizer that satisfies the symmetry relation $\bar{\bm y}=Q\bar{\bm y}$ for every $Q\in\mathcal P$.  
}
\end{proof}

\begin{remark}
Lemma~\ref{lem.sm.generaljensen} follows from known results on the existence of symmetric optima for certain classes of convex functions \cite{keilson1967global,waterhouse1983symmetric,boyd2004convex}.
If the objective function is strictly convex, then the inequality \eqref{eq.sm.inequalityderivation} becomes strict. Thus, in this case, only one global optimum exists, which necessarily satisfies the symmetry relation.
The same conclusion holds if the objective function is analytic. To see this, assume, for contradiction, that the set of global optima cannot be described as a finite set of points. Given that the function is convex, the set of global optima is also convex, so all derivatives within it must vanish. Therefore, the function is constant throughout its domain, which is a contradiction.
\end{remark}

We now directly apply Lemma~\ref{lem.sm.generaljensen} to formalize and prove the following result (which is presented in the main text):
\begin{theor} \label{thm.jensenineq}
Consider the Jacobian matrix $J(\textbf{b};A)$ of system \eqref{eq.generalnetwork} evaluated at the equilibrium point $\bm x^{\mathrm{eq}}$, where $J$ is an affine function of $\bm b$. Let $\bm b = (b_1,\ldots,b_N)$ be the damping parameters sought to be optimized, {and let $\mathcal P_A$ be a group of permutation matrices representing symmetries of the network (i.e., $A=P^{-1}AP$ for every
$P\in\mathcal P_A$)}. If the optimization problem \eqref{eq.optimization} is convex, then at least one globally optimal solution $\bm b^*$ must satisfy the symmetry relation $\bm b^* = P\bm b^*$.
\end{theor}

\begin{proof}
    {Let the objective function $C$ be the largest Lyapunov exponent $\Lambda_{\mathrm{max}}(J)$ associated with the Jacobian matrix $J(\bm b;A)$ with $\bm y\leftarrow\bm b$. That is, $C(\bm b) = \Lambda_{\rm max}(J(\bm b;A))$. To apply Lemma~\ref{lem.sm.generaljensen}, it remains to show that $C$ is invariant under every $P\in\mathcal P_A$. Since $P$ is a node permutation satisfying $A=P^{-1}AP$, relabeling the nodal parameters as $\bm b\mapsto P\bm b$ corresponds to a similarity transformation of the Jacobian:
    \begin{equation}
        J(P\bm b;A) = (I_q\otimes P )^{-1} J(\bm b;A) (I_q\otimes P).
    \end{equation}
    Therefore, $J(P\bm b;A)$ and $J(\bm b;A)$ have the same spectrum, and hence $C(P\bm b)=C(\bm b)$.
    The result then follows from Lemma~\ref{lem.sm.generaljensen}.
    }
\end{proof}

\begin{remark}
    Note that, in the second-order system \eqref{eq.secondordersys}, $J$ is always an affine function of $\bm b$. 
\end{remark}

In the following examples, we demonstrate the applicability of Theorem~\ref{thm.jensenineq} to the oscillator models considered in Fig.~\ref{fig.arnold}.

\begin{example}[First-order Kuramoto oscillator]
\label{examp.sm.1stkuramoto}

Consider the classical Kuramoto model \eqref{eq.sm.kuramotonet}, and let $A$ be a Hermitian matrix. We first investigate the influence of asymmetries in the natural frequency vector $\bm \omega$ on the stability of the synchronous state $\bm \phi^{\text{eq}} = (\phi_1^{\text{eq}},\ldots,\phi_N^{\text{eq}})$ such that $\dot\phi_i^{\rm eq} = 0$, $\forall i$. In the co-rotating frame, we have that $\omega_i = -\sum_j A_{ij} \sin(\phi_i^{\rm eq} - \phi_j^{\rm eq})$ and hence the Jacobian can be expressed as a time-invariant matrix
\begin{equation}
\begin{aligned}
\label{eq.sym_1_kur}
J_{ik}(\bm\phi^{\rm eq}) = \sum_{j=1}^N A_{ij} \cos(\phi_i^{\rm eq} - \phi_j^{\rm eq}) (\delta_{ik} - \delta_{kj}) = \sum_{j=1}^N A_{ij} \cos(\phi_i^{\rm eq} - \phi_j^{\rm eq}) - A_{ik} \cos(\phi_i^{\rm eq} - \phi_k^{\rm eq}) = J_{ki}(\bm\phi^{\rm eq}),
\end{aligned}
\end{equation}
where $\delta_{ij}$ denotes the Kronecker delta.
Thus, the Jacobian $J$ is an affine function of $\bm\omega$ and is Hermitian. By Theorem~\ref{thm.jensenineq}, we have that the optimization problem $\min_{\bm \omega} \Lambda_{\mathrm{max}}\big(J(\bm\omega)\big)$ is convex. Consequently, any asymmetric parameter vector $\bm \omega$ will not improve the stability of the synchronous state, since a symmetric configuration $\omega_1=\ldots=\omega_N$ is already optimal (as shown in Fig.~\ref{fig.arnold}A).

Now, we consider the leaky Kuramoto model \eqref{eq.sm.leakykuramoto}, where we assume uniform natural frequencies $\omega_i = \ldots = \omega_N$ but possibly asymmetric damping coefficients $b_i\neq b_j$. The (time-independent) Jacobian matrix is given by $J = - (B+L)$. Here, $J$ is an affine function of $\bm b$ and is once again always Hermitian under the assumption that $A$ is Hermitian. We thus have that $\min_{\bm b} \Lambda_{\mathrm{max}} \big(J(\bm b)\big)$ is convex and that at least one global optimum $\bm b^*$ must be symmetric (as shown in Fig.~\ref{fig.arnold}B). 

It follows that, for both models \eqref{eq.sm.kuramotonet} and \eqref{eq.sm.leakykuramoto}, asymmetric parameter configurations $\bm \omega$ or $\bm b$ may improve synchronization stability relative to the optimal symmetric configuration only if the underlying graph is directed (i.e., when $A \neq A^\dagger$).  
\QEDA
\end{example}

\begin{example}[Second-order Kuramoto oscillator]
Consider the second-order Kuramoto model \eqref{eq.sm.2ndkuramoto}. Independent of the adjacency matrix $A\neq 0$ structure, the Jacobian matrix \eqref{eq.sm.jacobian_second_kur} is always non-Hermitian. Therefore, we can show that $\min_{\bm b} \Lambda_{\mathrm{max}} \big(J(\bm b)\big)$ is nonconvex and, as a result, asymmetric damping configurations can potentially improve the system stability relative to symmetric configurations (as shown in Fig.~\ref{fig.arnold}C). Note that, for uniform frequencies $\omega_i$, we have $\phi^{\mathrm{eq}}_1 = \ldots = \phi^{\mathrm{eq}}_N$, and hence $J$ reduces to Eq.~\eqref{eq.jacobian}.
\QEDA
\end{example}

\begin{example}[Phase-amplitude oscillator]
\label{examp.sm.phaseamplitude}
For the phase-amplitude oscillator model \eqref{eq.sm.phaseamplitude}, the Jacobian matrix \eqref{eq.sm.jacobianphaseamplitude} is always non-Hermitian since the anti-diagonal matrix blocks are never equal for $\epsilon\neq 0$. Thus, we can show that $\min_{\bm{b}} \Lambda_{\mathrm{max}}\big(J(\bm b)\big)$ is nonconvex and, as a result, allowing parameter asymmetry can be beneficial even when $A$ is Hermitian (as shown in Fig.~\ref{fig.arnold}D).
\QEDA
\end{example}

{
The following example shows that Theorem~\ref{thm.jensenineq} does not necessarily apply to systems with more general coupling functions, such as higher-order interactions.
}

\begin{example}[{Higher-order Kuramoto model}]
\label{examp.sm.higherorder}
{
Higher-order interactions provide a natural extension of the class of systems considered in Eq.~\eqref{eq.generalnetwork} by allowing each (diffusive) coupling function $\bm g$ to depend on groups of three or more nodes~\cite{battiston2021physics,bick2023higher}. In synchronization problems, these interactions qualitatively change the structure of the linearized dynamics \cite{gallo2022synchronization}, influencing phenomena such as explosive synchronization \cite{kuehn2021universal} and trade-offs between the depth and size of basins of attraction \cite{zhang2024deeper,wang2026frequency}. A simple example is the Kuramoto model with both pairwise and three-body interactions:
\begin{equation}\label{eq.higherorderkuramoto}
\dot{\phi}_{i}=\omega_{i} + \sum_{j=1}^N A_{ij}^{(2)}\sin(\phi_{j}-\phi_{i}) + \sum_{j,k=1}^N A_{ijk}^{(3)}\sin(\phi_{j}+\phi_{k}-2\phi_{i}),
\end{equation}
where $A^{(2)}\in\R^{N\times N}$ and $A^{(3)}\in\R^{N\times N\times N}$. Linearizing around the phase-locked solution $\bm\phi^{\rm eq} = (\phi_1^{\rm eq},\ldots,\phi^{\rm eq}_N)^\transp$ yields the Jacobian matrix
\begin{equation} \label{eq.jacobianhigherorder}
J_{ij}=
\begin{cases}
    A_{ij}^{(2)} \cos(\phi_{j}^{*}-\phi_{i}^{*}) + \sum_{k=1}^{N} A_{ijk}^{(3)} \cos(\phi_{j}^{*}+\phi_{k}^{*}-2\phi_{i}^*),  \quad & i \neq j, \\
    -\sum_{j'\neq i}^{N}J_{ij'}, \quad & i = j.
\end{cases}
\end{equation}
This expression shows that higher-order interactions can generate non-Hermitian Jacobians even when the nodal dynamics are 1D and the underlying network is undirected. This is the case in Eq.~\eqref{eq.jacobianhigherorder} because the three-body contribution to $J_{ij}$ is generally not invariant under the permutation of nodes $i$ and $j$. Moreover, in this case, Eq.~\eqref{eq.optimization} is a nonconvex function of $\bm\omega$, which suggests that heterogeneity among the natural frequencies $\omega_i$ could in principle increase the stability of $\bm\phi^{\rm eq}$ even though the nodal dynamics are 1D. However, Theorem~\ref{thm.jensenineq} does not apply directly to this model, because the Jacobian depends on $\bm\omega$ only indirectly through the phase-locked solution $\bm\phi^{\rm eq}$, and therefore $J$ is not an affine function of $\bm\omega$. Prior work has shown that frequency heterogeneity in the higher-order Kuramoto model~\eqref{eq.higherorderkuramoto} can enlarge the basins of attraction of synchronous states, but at the cost of reduced linear stability~\cite{wang2026frequency}.
}
\QEDA
\end{example}

\subsection{{Nonconvexity, non-Hermiticity, and nonnormality}}
\label{sec.sm.nonconvexity}

In agreement with the analyses in Examples~\ref{examp.sm.1stkuramoto}--\ref{examp.sm.phaseamplitude}, the following propositions establish the relationship between (non-)Hermiticity of $J$ and the (non)convexity of the optimization problem \eqref{eq.optimization}. In particular, Proposition~\ref{cor.sm.nonhermitian} identifies non-Hermiticity as a necessary condition for nonconvexity, whereas Proposition~\ref{cor.sm.localcurvature} provides a sufficient condition for nonconvexity that holds generically within the set of all adjacency matrices $A$.

\begin{prop}
\label{cor.sm.nonhermitian}
    Let matrix $J(\bm b)$ depend affinely on the parameter vector $\bm b$. If the optimization problem $\min_{\bm b} \Lambda_{\mathrm{max}}\big(J(\bm b)\big)$ is nonconvex, then $J$ is non-Hermitian for some $\bm b$ in the feasible region $\R^N_{\geq 0}$.
\end{prop}

\begin{proof}
    We begin by proving the contrapositive of the first claim. Assume that $J$ is a Hermitian matrix for every $\bm b$. In this case, all eigenvalues of $J$ are real. Consequently, the largest Lyapunov exponent coincides with the largest eigenvalue (i.e., $\Lambda_{\mathrm{max}} = \max_i\lambda_i$). Following Ref.~\cite{overton1988minimizing}, when $J$ depends affinely on the parameters $\bm b$, the optimization problem is convex.
\end{proof}

\begin{remark}
    {While we focus on the case in which the parameters $b_i$ are nonnegative, our results extend naturally to cases in which the feasible region is the entire space $\R^N$ or a different convex set $\mathcal B\subseteq\R^N$, as in the case when $\bm b$ represents quantities other than those considered here.}
\end{remark}

\begin{prop}
\label{cor.sm.localcurvature}
{
    Consider the Jacobian matrix $J(\textbf{b};A)$ of system \eqref{eq.generalnetwork} evaluated at the equilibrium point $\bm x^{\mathrm{eq}}$.
    Suppose that, at some point $\bm b_0$, there exists an adjacency matrix $A_0$ for which $J(\bm b_0;A_0)$ has only simple eigenvalues and distinct real parts for every pair of distinct nonconjugate eigenvalues. For almost all choices of $A\in\R^{N\times N}$, the optimization problem $\min_{\bm b} \Lambda_{\mathrm{max}}\big(J(\bm b;A)\big)$ is nonconvex if the leading eigenvalue branch has negative curvature (i.e., there exist a feasible point $\bm{b}_{0}$ and direction $\delta\bm b_0$ such that $\frac{d^2}{d\epsilon^2} \operatorname{Re}\lambda_1(\bm b_0+\epsilon\delta\bm b_0;A)\big|_{\epsilon=0} <0$).
    }
\end{prop}

\begin{proof}
{
    Let the nonidentically null eigenvalues $\{\lambda_i : i\notin\mathcal Z\}$ of the Jacobian matrix $J$ be ordered according to their real parts, $\operatorname{Re}(\lambda_1) \geq \operatorname{Re}(\lambda_2) \geq \ldots$, so that $\Lambda_{\rm max} = \operatorname{Re}(\lambda_1)$. To determine the nonconvexity of Eq.~\eqref{eq.optimization}, three cases must be considered: 
    \begin{enumerate}
        \item[i)] the leading eigenvalue $\lambda_1$ is simple and strictly dominant up to a complex conjugate, so that $\operatorname{Re}(\lambda_1) > \operatorname{Re}(\lambda_j)$ for all $j\geq 2$ except when $\lambda_2=\lambda_1^{\rm c}$;
        \item[ii)] the leading eigenvalue has an algebraic multiplicity greater than 1; and
        \item[iii)] the real parts of two or more distinct nonconjugate eigenvalues locally coincide at some point $\bm b_0$, such that $\Lambda_{\rm max}(\bm b_0) = \operatorname{Re}(\lambda_1(\bm b_0)) = \operatorname{Re}(\lambda_j(\bm b_0))$ for some $j$ for which $\lambda_j\neq\lambda_1^{\rm c}$.
    \end{enumerate}
    }

    {
    For a fixed $\bm{b}_{0}$, we now show that the latter two cases occur only for an exceptional set of matrices $A$ of Lebesgue measure zero. From Eq.~\eqref{eq.generalnetwork}, it follows that  $J$ is an analytic function of $A$. Thus, the coefficients of the characteristic polynomial $\bm p(\lambda; A) = \operatorname{det}(J - \lambda I)$ are analytic functions of $A_{ij}$, $\forall i,j$, and so are the coefficients of its discriminant $\Delta(A)$. To eliminate case ii), recall that a matrix possesses repeated eigenvalues if and only if $\Delta(A)=0$. By assumption, there exists an adjacency matrix $A_0$ for which $J(\bm b_0;A_0)$ has only simple eigenvalues. Therefore, $\Delta(A_0)\neq 0$ and hence $\Delta$ is not identically zero. The zero set of a nonzero real analytic function has Lebesgue measure zero \cite{mityagin2015zero}. Thus, for almost every $A$, all eigenvalues of $J$ are simple. 
    To eliminate case iii), suppose two distinct nonconjugate eigenvalues satisfy $\operatorname{gap}(A) = 0$, where $\operatorname{gap}(A) = \operatorname{Re}(\lambda_1(A)) - \operatorname{Re}(\lambda_2(A))$. By assumption, there exists an adjacency matrix $A_0$ for which $\operatorname{gap}(A_0) \neq 0$, and hence the set $\{A\in\R^{N\times N} : \operatorname{gap}(A) = 0\}$ has Lebesgue measure zero.
    }

    {
    Thus, for almost all choices of $A$, $\Lambda_{\rm max}$ is determined by a simple, strictly dominant eigenvalue (up to a complex conjugate). Therefore, if there exists a point with negative curvature, the objective function is nonconvex.
    }
\end{proof}

\begin{remark}
{
Proposition~\ref{cor.sm.localcurvature} provides a direct test for the nonconvexity of the optimization problem~\eqref{eq.optimization}. For example, consider the 1D homogeneous parameterization $\bm b=b_{\hom}\bm 1_N$, with $b_{\hom}\in(0,\infty)$. If the curve $\Lambda_{\rm max}(b_{\hom}\bm 1_N)$ has negative curvature at any point where the leading eigenvalue is simple and strictly dominant, then the optimization problem is nonconvex (see Fig.~\ref{fig.verifynonconvexity} for an example). In fact, the conditions in Proposition \ref{cor.sm.localcurvature} hold for all 2D systems in Table~\ref{tab.systems}, and hence the corresponding optimization problems are nonconvex in all such cases.
}
\end{remark}

\begin{SCfigure}[1.4][t]
    \centering
    \includegraphics[width=0.3\linewidth]{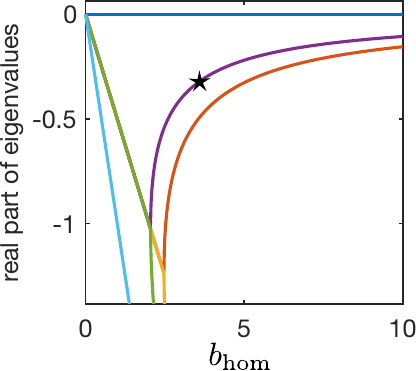}
    \caption{{\textbf{A simple test to verify the nonconvexity of the stability optimization problem.} The plot shows the real parts of all eigenvalues (color coded) of the Jacobian matrix \eqref{eq.jacobian} as functions of the homogeneous parameter $\bm b_{\rm hom} = b_{\rm hom}\bm 1_N$, for $N=3$ nodes. The star marks a point at which $\Lambda_{\rm max}(\bm b)$ is determined by a simple, strictly dominant, and nonindentically null eigenvalue branch with negative curvature. Proposition~\ref{cor.sm.localcurvature} therefore implies that the optimization problem \eqref{eq.optimization} is nonconvex. Consistent with the generic nature of this result, the calculation is performed for a randomly weighted all-to-all adjacency matrix.}}
    \label{fig.verifynonconvexity}
\end{SCfigure}

{
Proposition~\ref{cor.sm.nonhermitian} can be directly extended by replacing \textit{Hermiticity} with the weaker condition of \textit{normality}. We note that if $J$ is normal for all optimization variables in the set, then $\Lambda_{\mathrm{max}}$ can be expressed as \cite{trefethen2020spectra}
\begin{equation} \label{eq.lambdasympart}
\Lambda_{\rm max}(J) = \max_{i}\Re(\lambda_{J,i}) 
                      = \max_{i} \lambda_{J_{\rm s},i},
\end{equation}
where $\lambda_{J,i}$ is an eigenvalue of $J$ and $\lambda_{J_{\rm s},i}$ is an eigenvalue of $J_{\rm s} = (J+J^\dagger)/2$. Thus, for the case in which $J$ is an affine function of $\bm b$, the optimization problem \eqref{eq.optimization} is convex if $J$ is normal for all $\bm b$ (given that $\Lambda_{\rm max}$ reduces to the largest eigenvalue of the {Hermitian} matrix $J_{\rm s}$). Using the contrapositive of this assertion, we arrive at the conclusion that nonconvexity of optimization problem~\eqref{eq.optimization} implies nonnormality of $J$ for some $\bm b$ (which is analogous to Proposition \ref{cor.sm.nonhermitian}). Moreover, for a broad class of Jacobian matrices, the next proposition shows that nonnormality and non-Hermiticity are generically equivalent within the set of all adjacency matrices $A$.
}

\begin{prop}
\label{prop.nonnormal_herm}
{
Consider the Jacobian matrix $J(\textbf{b};A)$ of system \eqref{eq.generalnetwork} evaluated at the equilibrium point $\bm x^{\mathrm{eq}}$, where $J$ is an affine function of $\bm b$ and $\pdv{\bm g(\bm x_i,\bm x_j)}{\bm x_j}\big|_{\bm x = \bm x^{\rm eq}} \neq 0$ for some $j\neq i$. Suppose there exists an adjacency matrix $A_0$ such that $J(\bm b';A_0)$ is non-Hermitian for some $\bm b'$.
Then, for almost all choices of $A\in\R^{N\times N}$, $J(\bm b';A)$ is non-Hermitian and nonnormal.
}
\end{prop}

\begin{proof}
{By assumption, there exists a matrix $A_0$ for which the function $J(\bm{b}',A_0) - J^\dagger(\bm{b}',A_0)\neq 0$, which implies that this function is not identically zero. Therefore, the set $\{A \in\R^{N\times N}: J(\bm{b}',A)-J^\dagger(\bm{b}',A) = 0\}$ has Lebesgue measure zero, which proves that $J(\bm{b}';A)$ is generically non-Hermitian.}

Since $J(\bm{b}; A)$ is an affine function of $\bm{b}$, we can decompose the Jacobian as $J(\bm{b}; A)=J_{0}(A)+\sum_{\alpha=1}^{N}b_{\alpha}J_{\alpha}$. Therefore,
\begin{equation}
\label{eq.nonnormal_A_gen}
J(\bm{b};A)J(\bm{b};A)^{\dagger}-J(\bm{b};A)^{\dagger}J(\bm{b};A)=[J_{0}(A),J_{0}^{\dagger}(A)]+\sum_{\alpha=1}^{N}b_{\alpha}([J_{0}(A),J_{\alpha}^{\dagger}]+[J_{\alpha},J_{0}^{\dagger}(A)])+\sum_{\alpha,\beta=1}^{N}b_{\alpha}b_{\beta} [J_{\alpha},J_{\beta}^{\dagger}],
\end{equation}
where $[X,Y]$ is the commutator between $X$ and $Y$. Thus, $J(\bm{b}';A)$ is nonnormal for some $A$ if and only if
\begin{equation}
\label{eq.nonnormal_A_gen2}
[J_0(A),J_0^\dagger(A)]+\sum_{\alpha=1}^{N}b_{\alpha}'([J_{0}(A),J_{\alpha}^{\dagger}]+[J_{\alpha},J_{0}^{\dagger}(A)])\neq -\sum_{\alpha,\beta=1}^{N}b_{\alpha}'b_{\beta}' [J_{\alpha},J_{\beta}^{\dagger}].
\end{equation}
By assumption, we can choose $i\neq j$ such that $G_{ij} = \pdv{\bm g(\bm x_i,\bm x_j)}{\bm x_j}\big|_{\bm x = \bm x^{\rm eq}} \neq 0$. Since $J_0(A)$ is affine in $A$, the coefficient of the monomial $A_{ij}^2$ in $[J_0(A),J_0^\dagger(A)]$ has its $(j,j)$th-block given by $-G_{ij}^\dagger G_{ij}\neq0$.
Therefore, the left-hand side of Eq.~\eqref{eq.nonnormal_A_gen2} is nonconstant on $A$, and its right-hand side is constant on $A$. The commutator in Eq.~\eqref{eq.nonnormal_A_gen} is thus not identically zero, implying that $\{A \in\R^{N\times N}: [J(\bm{b}',A),J^\dagger(\bm{b}',A)]= 0\}$ has Lebesgue measure zero. This proves that $J(\bm b';A)$ is nonnormal for almost all $A\in\R^{N\times N}$.
\end{proof}

{The next corollary shows that a stronger claim can be made for a class of systems that includes all second-order models in Table~\ref{tab.systems}, namely that the Jacobian matrix is nonnormal for \textit{almost all} $\bm b$.}

\begin{corol}
\label{corol.Jac_so_normal}
{Consider the following class of Jacobian matrices
\begin{equation} \label{eq.sm.classofj}
J(\bm{b}; A)=
\begin{bmatrix}0&I_{N}\\-M_{1}-\epsilon B&-\gamma M_{2}-B\end{bmatrix},
\end{equation}
where $M_{1}$ and $M_{2}$ are real matrices with exclusively nonnegative eigenvalues and $\epsilon,\gamma\geq 0$. Then, 
$J(\bm{b}; A)$ is nonnormal for almost all choices of $\bm{b}\in\R^N_{\geq 0}$.
}
\end{corol}

\begin{proof}
{
It follows that
\begin{equation*}
\begin{aligned}
J^{\dagger}J&=\begin{bmatrix}0&-M_{1}^{\dagger}-\epsilon B\\I_{N}&-\gamma M_{2}^{\dagger}-B\end{bmatrix}\begin{bmatrix}0&I_{N}\\-M_{1}-\epsilon B&-\gamma M_{2}-B\end{bmatrix}=\begin{bmatrix}(M_{1}^{\dagger}+\epsilon B)(M_{1}+\epsilon B)&(M_{1}^{\dagger}+\epsilon B)(\gamma M_{2}+B)\\(\gamma M_{2}^{\dagger}+B)(M_{1}+\epsilon B)&(\gamma M_{2}^{\dagger}+B)(\gamma M_{2}+B)\end{bmatrix}, \,\,\, \text{and}\\
JJ^{\dagger}&=\begin{bmatrix}0&I_{N}\\-M_{1}-\epsilon B&-\gamma M_{2}-B\end{bmatrix}\begin{bmatrix}0&-M_{1}^{\dagger}-\epsilon B\\I_{N}&-\gamma M_{2}^{\dagger}-B\end{bmatrix}=\begin{bmatrix}I_{N}&-\gamma M_{2}^{\dagger}-B\\-\gamma M_{2}-B&(M_{1}+\epsilon B)(M_{1}^{\dagger}+\epsilon B)+(\gamma M_{2}+B)(\gamma M_{2}^{\dagger}+B)\end{bmatrix}.
\end{aligned}
\end{equation*}
The equality $J^{\dagger}J = JJ^{\dagger}$ requires the off-diagonal block matrices to be identical, which generically does not hold.
}
\end{proof}

\begin{remark}
\label{rem.nonnorm_1}
{
 An analogous corollary can be established for the Jacobians of the 2D models in Table~\ref{tab.systems} (FitzHugh-Nagumo and phase-amplitude oscillator models).
}
\end{remark}

\begin{corol}
\label{corol.Jac_so_normal_!D}
{The Jacobian matrix $J(\bm{b}; A)$ given by Eq.~\eqref{eq.jac1storder} is nonnormal for some choice of $\bm{b}\in\R^N_{\geq 0}$ if $A$ is non-Hermitian.}
\end{corol}

\begin{proof}
{We can write $[J,J^{\dagger}] = [B,(-L+L^{\dagger})]+LL^{\dagger}-L^{\dagger}L$.
If $A\neq A^\dagger$, the commutator $[B,(-L+L^{\dagger})]$ cannot be zero for every $\bm b\in\R^N_{\geq 0}$. Therefore, there exists some $\bm{b}$ for which $J(\bm{b};A)$ is nonnormal.}
\end{proof}

\begin{SCfigure}[0.5][t]
    \centering
    \includegraphics[width=0.69\linewidth]{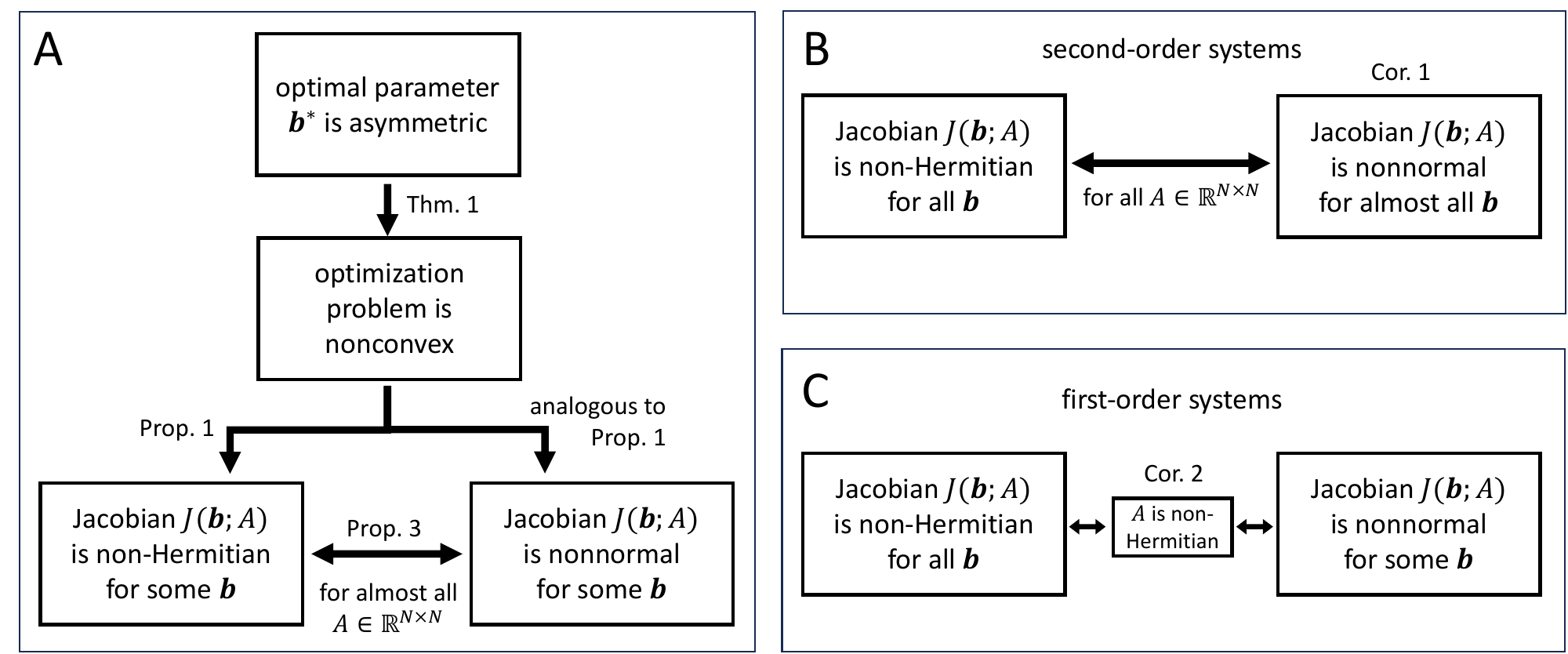}
    \caption{{\textbf{Relationship between nonconvexity, non-Hermiticity, and nonnormality.} The diagram summarizes the logical implications among these properties for (\textbf{A}) the broad class of Jacobian matrices arising from the network model in Eq.~\eqref{eq.generalnetwork}, (\textbf{B}) networks with second-order nodal dynamics described by the Jacobian matrices in Eqs.~\eqref{eq.jacobian} and \eqref{eq.sm.classofj}, and (\textbf{C}) networks with first-order nodal dynamics described by the Jacobian matrix in Eq.~\eqref{eq.jac1storder}.}}
    \label{fig.sm.diagram}
\end{SCfigure}

{
These theoretical results are summarized in Fig.~\ref{fig.sm.diagram}. Because non-Hermiticity and nonnormality are essentially equivalent in the setting considered here, we frame the discussion throughout the paper in terms of the (non-)Hermiticity of the Jacobian for consistency. This property can be inferred directly from the structure of the governing equations, whereas establishing nonnormality generally requires an explicit calculation. Moreover, Hermiticity provides a simple sufficient condition for convexity (contrapositive of Proposition~\ref{cor.sm.nonhermitian}).
}

\begin{remark} \label{rem.pseudospectra}
{
Throughout the text, we have focused our stability analysis on the largest Lyapunov exponent (i.e., the \textit{spectral abscissa}) of the Jacobian matrix. This quantity captures the asymptotic decay rate of infinitesimal perturbations. However, when the Jacobian is nonnormal, the asymptotic decay rate does not generally bound for the short-time response of the system, since perturbations may undergo substantial transient amplification before eventually decaying \cite{nishikawa2006synchronization,asllani2018structure,baggio2020efficient, duan2022network}. Lemma~\ref{lem.sm.generaljensen} can be applied to any objective function that is invariant under the relevant symmetries, including but not limited to $\Lambda_{\rm max}$. Another useful objective function is the \textit{pseudospectral abscissa} \cite{trefethen2020spectra},
\begin{equation}
\alpha_{\epsilon}(J)=\sup_{z\in\sigma_{\epsilon}(J)}\Re(z),
\end{equation}
where $\sigma(J)$ denotes the matrix spectrum and
$\sigma_{\epsilon}(J)=\{z\in\sigma(J+E) : \|E\|_2<\epsilon\}$ is the $\epsilon$-pseudospectrum, measuring how much the spectrum of $J$ can vary under perturbations of size $\epsilon>0$. The pseudospectral abscissa is thus relevant for assessing robustness to model uncertainties \cite{burke2003optimization}, and it also relates to transient amplification. For example, it provides an upper bound for the state deviation $\delta \bm x(t)$ from equilibrium, given by $\sup_{t\geq 0}\|\exp(Jt)\|_{2}\geq {\alpha_{\epsilon}(J)}/{\epsilon}$ \cite{trefethen2020spectra}. In the limit $\epsilon\rightarrow 0$, we have that $\alpha_{\epsilon}(J)-\epsilon$ approaches the spectral abscissa $\Lambda_{\rm max}(J)$. In contrast, for $\epsilon\rightarrow\infty$, the quantity $\alpha_{\epsilon}(J)-\epsilon$ approaches the \textit{numerical abscissa} $\omega(J)=\max_{i}(\lambda_{J_{\rm s},i})$, which determines the instantaneous growth rate of the state at time $t=0$.
For any finite value of $\epsilon$, the statements in Theorem~\ref{thm.jensenineq} and Proposition~\ref{cor.sm.nonhermitian} also hold for the optimization problem
\begin{equation}
\min_{\bm b} \,\,\, \alpha_{\epsilon}\big(J(\bm b;A)\big|_{\bm x^{\rm eq}}\big).
\end{equation}
In other words, the pseudospectral abscissa is a nonconvex function only when $J$ is non-Hermitian, in which case symmetry-broken parameter configurations may improve stability.
The limit case $\epsilon\rightarrow\infty$ is the sole exception. Following Eq.~\eqref{eq.lambdasympart}, since $\omega(J)$ is the largest eigenvalue of the Hermitian matrix $J_{\rm s}$, minimizing $\omega\big(J(\bm b)\big)$ is a convex optimization problem. Thus, symmetry-broken parameter configurations may enhance the (pseudo)spectral abscissa, but cannot enhance the numerical abscissa.
}
\end{remark}

\subsection{Network heterogeneity and stability}
\label{sec.sm.jensen.network}

Theorem~\ref{thm.jensenineq} provides a necessary condition under which the stability of a network system can benefit from \textit{nodal heterogeneities}, represented by an asymmetric parameter vector $\bm b$. We now build on this result to investigate when the stability of a system can benefit from \textit{network heterogeneities}, represented by network structures that break the symmetries of the nodal parameters.

\begin{theor}
\label{thm.sm.optadj}
     Consider the Jacobian matrix $J(\textbf{b};A)$ of system \eqref{eq.generalnetwork} evaluated at the equilibrium point $\bm x^{\mathrm{eq}}$. Let $A_{ij}$, for $i,j=1,\ldots,N$, be the edge weights sought to be optimized, {and let $\mathcal P_b$ be a group of permutation matrices representing symmetries of the nodal parameters (i.e., $\bm b=P\bm b$ for every
$P\in\mathcal P_b$)}. If the optimization
     \begin{equation}
         \min_{A} \Lambda_{\mathrm{max}}\big(J(\bm b;A)\big)
    \label{eq.sm.unconstrainednetopt}
     \end{equation}
     is convex, 
     then at least one globally optimal solution $A^*$ must satisfy the symmetry relation $A^* = P^\transp A^* P$ {for every $P\in\mathcal P_b$.}
\end{theor}

\begin{proof}
    The proof is analogous to Theorem~\ref{thm.jensenineq}.
\end{proof}

\begin{corol}
\label{cor.sm.optadj_constrain}
     The constrained optimization problem 
     \begin{equation}
    \begin{aligned}
        \min_{A} \,\,\,& \Lambda_{\mathrm{max}} \big(J(\bm b;A)\big), \\
        {\rm s.t.} \,\,\,& \sum_{i,j=1}^{N}|A_{ij}|\leq A_{\rm{max}}, \,\, \norm{A}_0\leq S_{\rm max},
    \end{aligned}
    \label{eq.constrainednetopt}
    \end{equation}
    is nonconvex for $1\leq S_{\rm max}<N^2-N$.
\end{corol}

\begin{proof}
        The constraint $\norm{A}_0\leq S_{\rm max}$ defines a nonconvex feasible set, and thus it follows that the problem~\eqref{eq.constrainednetopt} is nonconvex.
\end{proof}

\begin{remark}
    Corollary \ref{cor.sm.optadj_constrain} also applies to budget constraints on the node in-degrees, given by $|d_{i}| \leq d_{{\rm max}, i}$ for each node $i$.
\end{remark}

The following examples demonstrate the applicability of Theorem~\ref{thm.sm.optadj} and Corollary~\ref{cor.sm.optadj_constrain} to networks with 1D and 2D nodal dynamics, corresponding to the results presented in the diagram in Figs.~\ref{fig.symmetries}e–h.

\begin{example}[Network systems with 1D nodal dynamics]
\label{examp.sm.optadj}
 Consider the network optimization problem \eqref{eq.sm.unconstrainednetopt}, where $J=-(L+B)$ and $B$ is a diagonal matrix with fixed entries $B_{ii} = b_{i}$.  Without loss of generality, we can write  $b_{i}>b_{i+1}\,\,\,\forall i<N$. The unconstrained problem \eqref{eq.sm.unconstrainednetopt} admits trivial optimal solutions in which the entries of $A$ are arbitrarily large, such as all-to-all networks with unbounded weights.
 We therefore focus instead on the (physically meaningful) constrained problem \eqref{eq.constrainednetopt}, with $S_{\rm max}<N^2-N$. In this setting, we can choose $A_{\text{max}}=1+\sum_{i=2}^{N}(b_{1}-b_{i})$ without loss of generality. To ensure that all optimal solutions $A^*$ correspond to weakly connected graphs, we always exclude the eigenvalue of $J$ with the largest real part when computing $\Lambda_{\rm max}$ (while ensuring that this excluded eigenvalue is equal to $-\max_i b_i$). To understand why this exclusion is necessary, consider the homogeneous case (i.e., $B=bI_N$). In this case, the Jacobian always has an eigenvalue equal to $-b$ (corresponding to the null Laplacian eigenvalue), independently of the choice of $A$. Therefore, the optimization problem \eqref{eq.constrainednetopt} should seek to minimize the eigenvalue with the \textit{second} largest real part since the leading eigenvalue is invariant. For a heterogeneous parameter $\bm b$, we can write $J=-(L+B)=-(L'+B')-(l_{\rm max}+b_{\rm max})I_N$, where $b_{\rm max} = \max_i b_i$ and $l_{\rm max} = \max_{i,j} L_{ij}$. Since $-(L'+B')$ is a primitive matrix, its largest eigenvalue is $l_{\rm max}$\cite{berman1994nonnegative} and, hence, the largest eigenvalue of $J$ is lower bounded by $-\max_i b_i$. Thus, in this general case, it is also more relevant to minimize the eigenvalue with the second largest real part.

The optimal solutions of this constrained problem are weighted adjacency matrices of the \textit{master-slave} type, that is, directed tree networks with a single master node and a hierarchical organization among the remaining nodes. Up to node permutation, master-slave networks can be represented as
\begin{equation}
 A_{ij}=\delta_{j,1}\left(1-\delta_{i,1}\right)\left(b_{1}-b_{i}+\frac{1}{N-1}\right).
 \label{eq.sol_A_network}
 \end{equation}
We can prove this result by considering an arbitrary Laplacian $L^{(1)}$ of a weakly connected graph with algebraic connectivity $\alpha_{\rm ac}$. It is always possible to construct a master-slave Laplacian $L^{(2)}$ with the same algebraic connectivity by assigning all its nonzero values to be equal to $\alpha_{\rm ac}$. Since $L^{(2)}$ contains only a single nonzero entry per row, it follows that $\sum_{i,j}A_{ij}^{(2)}<\sum_{i,j}A_{ij}^{(1)}$. Therefore, the rescaled Laplacian $\frac{L^{(2)}\sum_{i,j}A_{ij}^{(1)}}{\sum_{i,j}A_{ij}^{(2)}}$ satisfies the same budget constraint as $L^{(1)}$ while yielding a strictly larger algebraic connectivity. This result implies that master-slave architectures maximize stability under the imposed constraint.

For the special case where the parameter vector is fully symmetric (i.e., $B=bI_{N}$), a possible optimal adjacency matrix and its corresponding Jacobian matrix are given by
\begin{equation}
\label{opt_adj_1}
 A_{ij}^*=\frac{\delta_{j,1}(1-\delta_{i,1})}{N-1}
 \,\,\, \text{and} \,\,\,
 J_{ij}(\bm b;A^*)=-\frac{\delta_{j,1}(1-\delta_{i,1})}{N-1}+\delta_{i,j}\left(b+\frac{1-\delta_{i,1}}{N-1}\right).
 \end{equation}
This solution corresponds to a directed star network in which all edges have equal weight. However, this configuration is not unique: among the class of master-slave networks, the network structure illustrated in Fig.~\ref{fig.symmetries}E is equally optimal. We note that this solution is also optimal for homogeneous $b$ in the 2D case with the Jacobian \eqref{eq.jacobian}, since, in this case, the network modes decouple \cite{nishikawa2006maximum}.

In the case where we have two symmetry clusters (i.e., $B$ has two distinct parameter values such that $b_{i}=b_{1}$ if $i\leq N_1$ and $b_{i}=b_{2}$ if $N_1<i\leq N$), the optimal adjacency matrix and its corresponding Jacobian matrix are
 \begin{equation}
 A_{ij}^{*}=\begin{cases}\frac{\delta_{j,1}(1-\delta_{i,1})}{N-1}, \,\,\,  &\text{ if } i\leq N_1,\\
 \delta_{j,1}\left(b_{1}-b_{2}+\frac{1}{N-1}\right), \,\,\, &\text{ otherwise, }\end{cases}
 \,\,\, \text{and} \,\,\, 
 J_{ij}(\bm{b};A^{*})=\begin{cases}-\frac{\delta_{j,1}\left(1-\delta_{i,1}\right)}{N-1}+\delta_{i,j}\left(1-\delta_{i,1}\right)\left(b_{1}+\frac{1}{N-1}\right), \,\,\, &\text{ if } i\leq N_1,\\
 -\delta_{j,1}\left(b_{1}-b_{2}+\frac{1}{N-1}\right)+\delta_{i,j}\left(b_{1}+\frac{1}{N-1}\right), \,\,\, &\text{ otherwise. }\end{cases}
\end{equation}
This solution corresponds to a weighted directed star network in which edges pointing to nodes with parameter value $b_1$ have different weights from edges pointing to nodes with parameter value $b_2$. 
This structure generalizes naturally to multiple symmetry clusters, as illustrated in Fig.~\ref{fig.symmetries}G for a system with four symmetry clusters.
\QEDA
\end{example}

\begin{example}[Network systems with 2D nodal dynamics]
\label{examp.sm.optadj2ndorder}

Consider the Jacobian matrix \eqref{eq.jacobian} associated with a network of $N=5$ nodes and a partially symmetric parameter vector $\bm b = [3, 2, 2, 1, 1]$.  We study the constrained optimization problem \eqref{eq.constrainednetopt}, with $A_{\text{max}}=1+\sum_{i=2}^{N}(b_{1}-b_{i})$ and $S_{\rm max}=\frac{N^2-N}{2}$. 
In this case, the directed tree structures defined by Eq.~\eqref{eq.sol_A_network} are no longer optimal, yielding a Lyapunov exponent $\Lambda_{\rm max} = -0.5$.
To identify the optimal solution, we employ the sequential least squares programming (SLSQP) algorithm using 1{,}000 independently sampled random initial conditions. The resulting optimal adjacency matrix corresponds instead to the following directed network:
\begin{equation}
A^* = \begin{bmatrix}
0& 0.039& 0.042& 0& 2.09\\
       0& 0& 0& 0& 1.781\\
       1.52& 0& 0& 0& 0\\
       0& 0& 0& 0& 0.754\\
       0.342& 0.096& 0& 0.335& 0
\end{bmatrix},
\end{equation}
which exhibits a substantially less constrained topology than the directed trees obtained in the 1D case.
\QEDA
\end{example}

For LV networks, we now show that under certain conditions, any random sufficiently small perturbation to the interaction weights enhances the system stability.

\begin{theor}
\label{thm.ecologicalnet}
Consider the LV model \eqref{eq.lvmodel}, and assume that the growth rate vector $\bm b$ is nonuniform (i.e., $\bm b \neq b\bm 1_{N}$ for any scalar $b$). Let the adjacency matrix be given by $A = - I_N + \frac{\sigma}{N} \delta A$, where $\sigma\in\R$ and $\delta A\in\R^{N\times N}$, and let the Jacobian matrix at the equilibrium point $\bm x^{\mathrm{eq}} = -A^{-1}\bm b$ be denoted by $J(\sigma;A) = XA$. 
Suppose the matrix perturbation $\delta A$ satisfies the following conditions:
\begin{enumerate}
    \item $(\delta A)_{ii} = 0$ for all $i$ (no self-interactions),
    \item $(\delta A)_{ij} \geq 0$ for all $i\neq j$ (nonnegative interactions),
    \item $(\delta A)_{i_{0}j}>0$ for $i_{0} = \argmin_i b_i$ and some $j\neq i_{0}$ (the species with the lowest growth rate receives nonzero input).
\end{enumerate}

\noindent
Then, there exists a constant $\sigma_{\mathrm{max}}$ such that, for all $0<\sigma<\sigma_{\mathrm{max}}$, the largest Lyapunov exponent of the Jacobian matrix satisfies
\begin{equation}
    \Lambda_{\mathrm{max}}\big(J(\sigma;A)\big) < \Lambda_{\mathrm{max}}\big(J(0;A)\big).
\end{equation}
\end{theor}

\begin{proof}
For $\sigma=0$, $J = -X$ and $\Lambda_{\mathrm{max}}(J;A) = -\min_i x_i^{\mathrm{eq}}$. We consider the effect of small perturbations $\delta A$ using a first-order expansion
\begin{equation}
A^{-1} = \left(I_N-\frac{\sigma}{N} \delta A \right)^{-1} = \sum_{k=0}^{\infty}\frac{\sigma^k}{N^k} (\delta A)^{k} =  I_N+\frac{\sigma}{N} \delta A + \mathcal O\left(\frac{\sigma^2}{N^2}(\delta A)^2\right).
\label{eq.sm.firstorderlvadj}
\end{equation}

\noindent
Thus, the equilibrium point is given by $\bm x^{\mathrm{eq}} = - A^{-1} \bm b \approx (I_N+\frac{\sigma}{N} \delta A) \bm b$ and, hence, the Jacobian matrix is approximated by
\begin{equation}
    J = XA \approx \text{diag} \Big( \big(I_N+\frac{\sigma}{N} \delta A \big) \bm b\Big)\Big(-I_N+\frac{\sigma}{N}\delta A \Big) = B + \frac{\sigma}{N} \left(B\delta A - \text{diag}(\delta A\bm b) \right).
\label{eq.sm.lvjacobianderivation}
\end{equation}

\noindent
Let us define the matrix $\Psi = B\delta A - \text{diag}(\delta A\bm b)$, and note that
\begin{equation*}
    \Psi_{ij} = 
    \begin{cases}
        -\sum_{j=1}^{N}(\delta A)_{ij}b_{j} \,\,\, & \text{if} \,\, i = j, \\
        \,\,\,\,\,\,\,\,\,\,\,\,\,\,\,(\delta A)_{ij}b_{i} \,\,\, & \text{if} \,\, i\neq j.
    \end{cases}
\end{equation*}

\noindent
We can now analyze the Gershgorin discs associated with the approximated Jacobian matrix \eqref{eq.sm.lvjacobianderivation} to assess the stability properties of $\bm x^{\mathrm{eq}}$. Each Gershgorin disc $D_i(J_{ii},R_i)$ has the following center and radius:
\begin{equation*}
    \begin{aligned}
        J_{ii} &\approx - b_i - \frac{\sigma}{N}\sum_j (\delta A)_{ij} b_j, \\
        R_{i} &= \sum_{j\neq i=1}^{N} |J_{ij}| \approx \frac{\sigma b_{i}}{N} \sum_{j\neq i=1}^{N} (\delta A)_{ij}.
    \end{aligned}
\end{equation*}

\noindent
Thus, the rightmost point of disc $D_i$ lies at $-b_i + \frac{\sigma}{N} \sum_{j \neq i=1}^{N} (\delta A)_{ij} (b_i - b_j)$. For $i_{0}=\argmin_i b_i$, the difference $(b_{i_{0}} - b_j)$ is nonpositive and, as a result, the inequality $\sum_{j\neq i_{0}=1}^{N} (\delta A)_{ij}(b_{i_{0}} - b_j) < 0$ holds given that $(\delta A)_{i_{0} j} > 0$ by assumption. Moreover, $\sum_{j\neq i=1}^{N} |(\delta A)_{ij}(b_{i} - b_j)|$ is maximized for $i=i_{0}$. It follows that, under the first-order approximation in Eq.~\eqref{eq.sm.firstorderlvadj} (which is valid for sufficiently small $\sigma<\sigma_{\mathrm{max}}$), the disc $D_{i_{0}}$ is shifted leftwards more than any other disc $D_{i\neq i_{0}}$ is shifted rightwards. Therefore, $\Lambda_{\mathrm{max}}(\sigma)<\Lambda_{\mathrm{max}}(0)$ for all $\sigma\in (0,\sigma_{\mathrm{max}})$.
\end{proof}

\begin{remark}
    Theorem~\ref{thm.ecologicalnet} establishes that introducing positive edge weights $(\delta A)_{ij}$ into an initially uncoupled network (i.e., when $A = - I_N$) will always improve the stability of the overall ecological system. This analytical result is consistent with the numerical findings presented in Fig.~\ref{fig.ecologicalnet}, which show that, in mutualistic networks (characterized by strictly nonnegative weights $A_{ij}\geq 0$, $\forall i\neq j$), small structural perturbations lead to increased stability. In the simulations, each perturbation $(\delta A)_{ij}$ was modeled as a strictly positive random variable drawn from a half-normal distribution, i.e., $(\delta A)_{ij} \sim |\mathcal N(0,\sigma^2)|$. 
    Thus, the standard deviation $\sigma$ directly controls the strength of interspecies interactions and corresponds to the perturbation magnitude analyzed in Theorem~\ref{thm.ecologicalnet}.
\end{remark}

\bigskip
\section{Conditions for disorder-promoted stability}
\label{sec.sm.dpsconditions}

In this section, we analytically investigate how parameter asymmetry affects the stability of systems that exhibit a graph automorphism with full symmetry and whose Jacobian matrix is described by Eq.~\eqref{eq.jacobian}. To this end, we first characterize the differentiability of the stability function $\Lambda_{\mathrm{max}}(\bm b)$ at a global optimum $\bm b^*$ (Sec.~\ref{sec.sm.differentiability}). Building on these results, we identify a class of network structures for which (random) parameter perturbations have a high probability of enhancing the system stability (Sec.~\ref{sec.sm.circulant}).
These results are based on the following assumptions:

\begin{assump}
    The Laplacian matrix $L$ is diagonalizable, only has eigenvalues with nonnegative real parts, and size $N\geq 3$.
\end{assump}

\begin{assump}
    The damping coefficients are nonnegative (i.e., $b_i\geq 0$, $\forall i$).
\end{assump}

\noindent
The following notation is also adopted in this section. Unless noted otherwise, we sort the eigenvalues of $L$ by the ascending order of their real part: $\mathrm{Re}(\lambda_{L,1}) \leq \ldots \leq \mathrm{Re}(\lambda_{L,N})$. Due to the quadratic nature of the considered eigenvalue problems, the eigenvalues of $J$ have conjugate branches, respectively denoted $\{\lambda_{i,+}\}$ and $\{\lambda_{i,-}\}$. The formulas for the derivatives of eigenvalues and eigenvectors of the quadratic eigenvalue problem are taken from Ref.~\cite{adhikari2001eigenderivative}. 

\smallskip
\subsection{Conditions for differentiability of the stability landscape at optimal parameter configurations}
\label{sec.sm.differentiability}

Proposition~\ref{prop.theor_hermit_diff} establishes a sufficient condition under which the gradient function of the objective function $\Lambda_{\mathrm{max}}(\bm b)$ is \textit{never} zero, implying that a global optimum $\bm b^*$ must lie at a nondifferentiable point in the stability landscape. Proposition~\ref{prop.nondifferentiablehom}, in turn, focuses on the differentiability conditions at optimal \textit{homogeneous} parameters.

\begin{prop}
\label{prop.theor_hermit_diff}
{Consider the Jacobian matrix \eqref{eq.jacobian}.}
If $L$ is a Hermitian matrix, then $\nabla_{\bm b}\Lambda_{\mathrm{max}}\big(J(\bm b;A)\big) \neq 0$ for all $\bm b\in\R^N_{\geq 0}$ such that $\Lambda_{\mathrm{max}}(\bm b)$ is differentiable and $\bm b$ is not a saddle point.
\end{prop}

\begin{proof}
Let $\lambda_{\mathrm{max}}$ be the eigenvalue of the Jacobian matrix \eqref{eq.jacobian} with the largest real part. 
Because $J$ is associated with a quadratic eigenvalue problem, the eigenvalue $\lambda_{\mathrm{max}}$ satisfies the following characteristic equation
\begin{equation}
 \Big(\lambda_{\mathrm{max}}^2(\bm b) I_N+\lambda_{\mathrm{max}}(\bm b)B+L\Big) \bm u_{\mathrm{max}}(\bm b)=0, 
 \label{eq:quad_eig_prob}
\end{equation}
where $\bm u_{\mathrm{max}}$ is the right eigenvector associated with $\lambda_{\mathrm{max}}$, and the dependence on the damping vector $\bm b$ is explicitly indicated. 
%
If $\bm b$ is chosen such that the gradient is well-defined, then both $\lambda_{\mathrm{max}}$ and $\bm u_{\mathrm{max}}$ are analytic functions in a neighborhood of this point. 
Taking the gradient of Eq.~\eqref{eq:quad_eig_prob} with respect to $\bm b$ and setting it equal to zero yields
\begin{align}
\Big(\lambda^2_{\mathrm{max}}I_N+\lambda_{\mathrm{max}} B+L\Big)\frac{\partial\bm u_{\mathrm{max}}}{\partial b_{i}}
+\frac{\partial\lambda_{\mathrm{max}}}{\partial b_{i}}\Big(2\lambda_{\mathrm{max}} I_N+B\Big)\bm u_{\mathrm{max}}+\lambda_{\mathrm{max}}R_{i} \bm u_{\mathrm{max}}=0_N,  
\label{eq:grad_b}
\end{align}
where $(R_{i})_{jk} =\delta_{ij}\delta_{jk}$ is the matrix with a single nonzero entry at the $i$th entry of the diagonal.

Now, let $\bm v_{\mathrm{max}}$ be the left eigenvector associated with $\lambda_{\mathrm{max}}$, which satisfies $\bm v_{\mathrm{max}}^\dagger(\bm b) \Big(\lambda_{{\mathrm{max}}}^2(\bm b)I_N +\lambda_{\mathrm{max}}(\bm b)B+L\Big)=0$.
For quadratic eigenvalue problems \cite{adhikari2001eigenderivative}, the left and right eigenvectors can be normalized as follows:
\begin{equation}
\bm v_{\mathrm{max}}^\dagger \Big(2\lambda_{\mathrm{max}} I_N+B\Big)\bm u_{\mathrm{max}}=1. \label{eq:normalization}
\end{equation}
Left-multiplying Eq.~\eqref{eq:grad_b} by $\bm v^\dagger_{\mathrm{max}}$ and applying the normalization~\eqref{eq:normalization}, we obtain 
\begin{equation}
\frac{\partial\lambda_{\mathrm{max}}}{\partial b_{i}}=-\lambda_{\mathrm{max}}v^{\mathrm{c}}_{\mathrm{max},i}u_{\mathrm{max},i}, 
\label{eq:grad_b_res}
\end{equation}
where $u_{{\mathrm{max}},i}$ indicates the $i$th element of $\bm u_{{\mathrm{max}}}$.
Therefore, the vanishing gradient condition $\frac{\partial\lambda_{\mathrm{max}}}{\partial b_{i}}=0$ is only satisfied when $v_{{\mathrm{max}},i}^{\mathrm{c}}\,u_{{\mathrm{max}},i}=0$, $\forall i$.
However, this contradicts the normalization condition~\eqref{eq:normalization}, which ensures that there exists at least one index $i$ such that $v^{\mathrm{c}}_{{\mathrm{max}},i}u_{{\mathrm{max}},i}\neq 0$. Thus, the gradient of $\lambda_{{\mathrm{max}}}$ is never zero, that is, $\nabla_{\bm b} \lambda_{\mathrm{max}}(\bm b) \neq 0$ at any differentiable point $\bm b$.

It remains to consider whether the gradient of $\lambda_{\mathrm{max}}$ could be nonzero yet purely imaginary, thereby allowing $\nabla \Lambda_{\mathrm{max}} \coloneq \mathrm{Re} (\nabla \lambda_{\mathrm{max}})$ to be zero. Since $L$ is Hermitian, the left and right eigenvectors of Eq.~\eqref{eq:quad_eig_prob} are proportional, i.e., $\bm v_{j}\propto \bm u_{j}^c$, $\forall j$ \cite{tisseur2001quadratic}. Eq.~\eqref{eq:grad_b_res} thus reduces to $-\lambda_{\mathrm{max}} u_{{\mathrm{max}},i}^2$, which is real (complex) when $\lambda_{\mathrm{max}}$ is real (complex). We now consider the two cases separately:
\begin{itemize}
    \item If $\lambda_{\mathrm{max}}$ is real, then $u_{\mathrm{max},i}^2$ is real and nonzero for some $i$, which implies that the derivative $\partial\lambda_{\mathrm{max}}/\partial b_{i}$ is nonzero.

    \item Otherwise, we can write
    \begin{equation}
        \lambda_{\mathrm{max}}=\frac{-\bm u_{\mathrm{max}}^\dagger B \bm u_{\mathrm{max}}+\sqrt{(\bm u_{\mathrm{max}}^\dagger B\bm u_{\mathrm{max}})^2-4\bm u_{\mathrm{max}}^\dagger L \bm u_{\mathrm{max}}}}{2\bm u_{\mathrm{max}}^\dagger \bm u_{\mathrm{max}}}.
        \label{quadr_eig}
    \end{equation}
    It follows that $\frac{\bm u_{\mathrm{max}}^\dagger B\bm u_{\mathrm{max}}}{\bm u_{\mathrm{max}}^\dagger \bm u_{\mathrm{max}}}>0$ since $B$ is a diagonal and positive definite matrix by assumption. Additionally, given that $L$ is a semi-positive definite matrix, we have $\big(\bm u_{\mathrm{max}}^\dagger B\bm u_{\mathrm{max}}\big)^2-4 \bm u_{\mathrm{max}}^\dagger L\bm u_{\mathrm{max}}<0$. Thus, the term inside the radical remains negative for a small perturbation of $\bm b$ in a direction that reduces the term outside the square root. It is possible to prove that such a descending direction exists, which implies that $\lambda_{\mathrm{max}}$ decreases and, therefore, this point is not a local minima. 
\end{itemize}

\noindent
We thus conclude that local minima of the largest Lyapunov exponent can only occur at nondifferentiable points.
\end{proof}

\begin{remark}
The proof of Proposition~\ref{prop.theor_hermit_diff} also shows that, at differentiable points, the real part of the gradient of any eigenvalue of $J$\textemdash not just the one related to the largest Lyapunov exponent\textemdash is nonzero.
\label{rem.proposition.nonzeroeig}
\end{remark}

\begin{prop}
{Consider the Jacobian matrix \eqref{eq.jacobian}.}
If all the eigenvalues of the Laplacian matrix $L$ are real, then $\Lambda_{\mathrm{max}}\big(J(\bm b;A)\big)$ is not differentiable at the homogeneous optimal point $\bm b^{*}_{\mathrm{hom}} = b^{*}_{\mathrm{hom}} \bm 1_{N}$.
\label{prop.nondifferentiablehom}
\end{prop}

\begin{proof}
For a homogeneous damping vector $\bm b = b \bm 1_N$, it follows from Eq.~\eqref{eq:quad_eig_prob} that the eigenvalue spectrum of $J$ is given by
\begin{equation}
    \lambda_{j,\pm}^2+\lambda_{j,\pm} b+\lambda_{L,j}=0,
\label{eq.sm.homogeneouseigproblem}
\end{equation}
\noindent
where $\lambda_{L,j}$ is the $j$th eigenvalue of $L$, and thus each eigenvalue $\lambda_{j,\pm}$ of $J$ depends exclusively on a single eigenvalue $\lambda_{L,j}$ of $L$. 
Note that every identically null eigenvalue of $J$ is associated with a zero eigenvalue of $L$ (which is excluded by definition in the computation of the largest Lyapunov exponent $\Lambda_{\mathrm{max}}$), but the converse is not necessarily true.

Now, consider the constrained optimization problem
\begin{equation}
\begin{aligned}
        \min_{\bm b} \,\,\,& \Lambda_{\mathrm{max}} \big(J(\bm b;A)\big), \\
        \text{s.t.} \,\,\,& \bm b = b\bm 1_N,
\end{aligned}
\label{eq.constrainedoptimization}
\end{equation}
whose optimum is denoted as the homogeneous optimal point $\bm b^*_{\mathrm{hom}} = b^*_{\mathrm{hom}}\bm 1_N$. From Eq.~\eqref{eq.sm.homogeneouseigproblem}, this optimization is equivalent to
\begin{equation*}
\min_{b\geq0} \,\,\, \max_{j\notin\mathcal Z} \,\,\, \frac{1}{2} \left(-b+\mathrm{Re}\left(\sqrt{b^2-4\lambda_{L,j}}\right)\right),
\end{equation*}
where the set of identically null eigenvalues of the Jacobian matrix is denoted by $\mathcal Z = \{j   : \, \lambda_{j,\pm}(\bm b)=0,\,\,\,\forall \bm b\in \R^{N}_{\geq 0}\}$. Applying the transformation $b = 2\hat{b}$, we have the following optimization problem:
\begin{equation}
\min_{\hat{b}\geq0} \,\,\, \max_{j\notin\mathcal Z} \,\,\, -\hat{b}+\mathrm{Re}\left(\sqrt{\hat{b}^2-\lambda_{L,j}}\right).
\label{eq.sm.equivalentconstrainedopt}
\end{equation}

\noindent
By assumption, we have that $\lambda_{L,j}\in\R$ and, hence, it follows that
\begin{equation*}
\mathrm{Re}\left(\sqrt{\hat{b}^2-\lambda_{L,j}}\right) = 
\begin{cases}
    \sqrt{\hat{b}^2-\lambda_{L,j}}, \quad & \text{if} \,\, \hat{b}^2>\lambda_{j}, \\
    0, \quad & \text{otherwise}. 
\end{cases}
\end{equation*}
\noindent
Consequently, $\mathrm{Re}\left(\sqrt{\hat{b}^2-\lambda_{L,j}}\right)\leq \mathrm{Re}\left(\sqrt{\hat{b}^2-\lambda_{L,k}}\right)$ if $\lambda_{L,j}\geq\lambda_{L,k}$. The optimization problem \eqref{eq.constrainedoptimization} thus reduces to
\begin{equation}
\label{eq.sm.equivproblem}
\min_{\hat{b}} \,\,\, -\hat{b}+\mathrm{Re}\left(\sqrt{\hat{b}^2-\lambda_{L,1}}\right).
\end{equation}
\noindent
The solution is $\hat{b}^*=\sqrt{\lambda_{L,1}}$, which corresponds to a nondifferentiable point of the objective function in Eq.~\eqref{eq.sm.equivproblem}.
\end{proof}

\begin{remark}
These results show that, when $L$ is Hermitian, the optimal damping configuration lies at a nondifferentiable point of the stability landscape, posing a challenge for standard analysis and optimization tools due to the problem's nonsmoothness. This limitation motivates the study of non-Hermitian structures, such as non-Hermitian circulant matrices, for which the stability function may be differentiable, thereby facilitating analytical tractability and the use of gradient-based optimization approaches.
\end{remark}

\smallskip
\subsection{Conditions for disorder-promoted stability in circulant networks}
\label{sec.sm.circulant}

We now focus on the important class of network structures described by \textit{circulant matrices}, for which disorder-promoted stability can be established analytically. Within this class, we identify conditions under which the optimal damping parameter $\bm b^*$ always lies at a heterogeneous point. 
We first prove this property for networks of second-order dynamical systems characterized by the Jacobian matrix \eqref{eq.jacobian} (Lemma~\ref{lem.circ_diff_ass} and Theorem~\ref{thm.local_CSB}), and then extend these results to more general second-order systems (Corollary~\ref{cor.dpsgeneraljac}). Before we proceed, we define the class of circulant matrices.

\begin{defin}
A circulant matrix $C\in\C^{N\times N}$ has the form
\begin{equation}
    C = \begin{bmatrix}
        c_1 & c_{N} & \ldots & c_2 \\
        c_2 & c_1 & \ldots & c_3 \\
        \vdots & \vdots & \ddots & \vdots \\
        c_{N-1} & c_{N-2} & \ddots & c_{N} \\
        c_{N} & c_{N-1} & \ldots & c_1
    \end{bmatrix}.
\end{equation}

\noindent
Thus, a circulant matrix can be fully generated by one vector $\bm c = (c_1,\ldots,c_N)$, which appears as the first column of $C$, whereas the remaining columns of $C$ are each cyclic permutations of the vector $\bm c$. Additionally, if $c_{i}=c_{(N+2-i) \, {\mathrm{mod}} \, N}$, for $2\leq i\leq N$, then the circulant matrix is Hermitian. 
\label{def.sm.circulant}
\end{defin}

Fig.~\ref{fig.circulant}A illustrates circulant network structures and the theoretical results that apply to each subclass. For Hermitian circulant matrices, Proposition~\ref{prop.nondifferentiablehom} has established that the homogeneous optimal point always lies at a nondifferentiable point, making it analytically challenging to determine whether disorder can be beneficial for stability. In what follows, Lemma~\ref{lem.circ_diff_ass} shows that, for a subclass of non-Hermitian circulant matrices (represented by directed first-neighbor ring networks), the homogeneous optimal point is always differentiable. Consequently, it follows from Theorem~\ref{thm.local_CSB} that disorder-promoted stability is \textit{always} guaranteed in this subclass of circulant networks. For broader classes of circulant matrices, the largest Lyapunov exponent may or may not be differentiable; however, whenever differentiability holds,  Theorem~\ref{thm.local_CSB} ensures that disorder promotes stability.

\begin{figure}
    \centering
	\includegraphics[width=\textwidth]{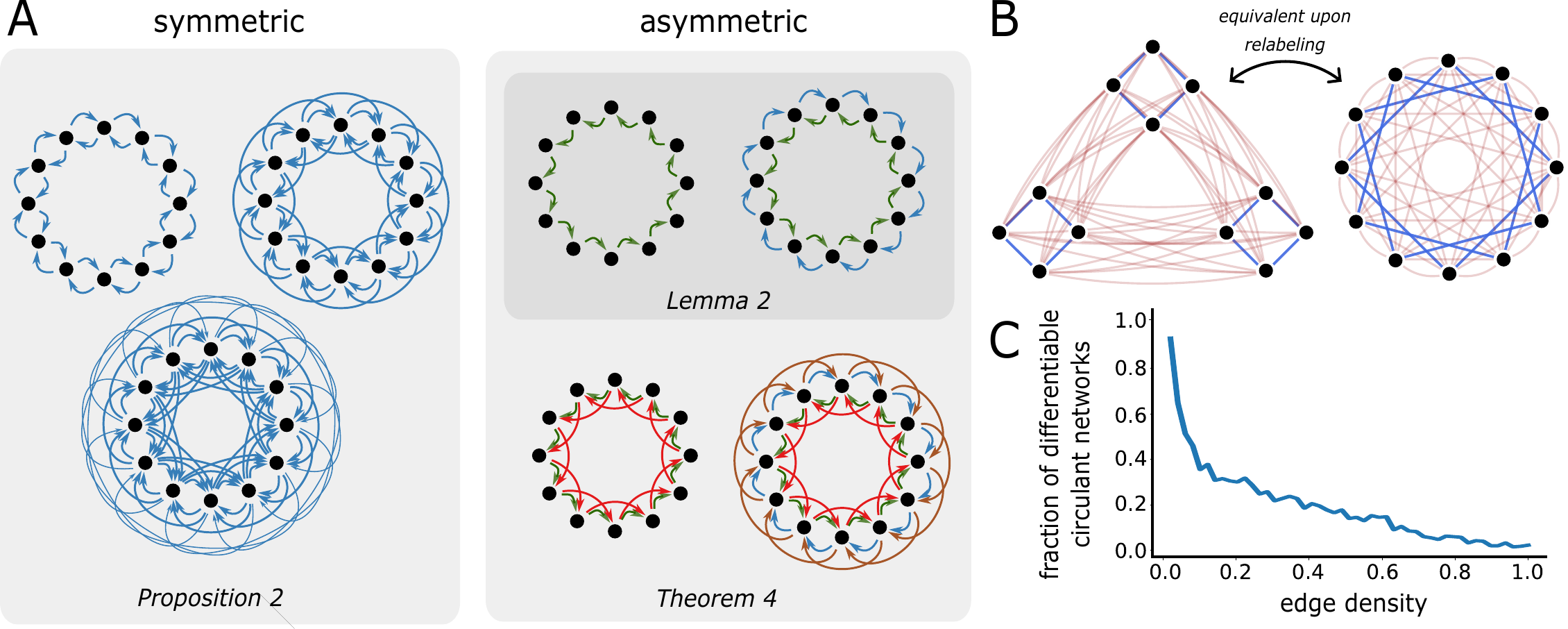}
    \caption{\textbf{Classes of circulant networks with distinct edge densities and symmetries.}  
    (\textbf{A}) Circulant networks comprise a broad class of networks with an automorphism with full symmetry, including both (directed and undirected) ring and all-to-all networks. The diagram illustrates the subclasses of circulant networks for which our theoretical results apply, with edge colors indicating subsets of edges sharing the same weights. 
    (\textbf{B}) Symmetrically coupled clusters of the same size, as often considered in the study of chimera states \cite{abrams2004chimera,zhang2020critical} and cluster synchronization \cite{montbrio2004synchronization}, are also a circulant network. This example shows that a ring of connected clusters (with distinct intra- and inter-cluster coupling strengths) can be identified as a circulant network upon relabeling of nodes.
    (\textbf{C}) Fraction of directed circulant networks with differentiable stability landscapes at the homogeneous optimum, shown as a function of the proportion of nonzero entries in the generating vector $\bm c$. 
    Note that this proportion directly determines the edge density of a circulant network. The results are averaged over $100$ independent network realizations of size $N=100$.
    }
\label{fig.circulant}
\end{figure}

\begin{lemma}
Consider the Jacobian matrix~\eqref{eq.jacobian}. Let $L$ be a non-Hermitian circulant matrix generated by the vector 
\begin{equation*}
\bm c = [1,-\delta,0,...,0,-1+\delta],\,\,\,\,\,\,0\leq \delta\leq 1,\,\,\,\,\,\,\delta\neq \frac{1}{2}.
\label{gen_vector_diff}
\end{equation*} 
\noindent
Then, $\Lambda_{\mathrm{max}}(\bm b)$ is differentiable at $\bm b_{\mathrm{hom}}^{*} = b_{\mathrm{hom}}^{*} \bm 1_N$.
\label{lem.circ_diff_ass}
\end{lemma}

\begin{proof}
Since the damping vector is homogeneous (i.e., $\bm b = b\bm 1_N$), it follows from Eq.~\eqref{eq.sm.homogeneouseigproblem} that each eigenvalue of $J$ depends on a single, unique eigenvalue of $L$. From Ref. \cite{gray2006toeplitz}, the eigenvalues of $L$ are
\begin{equation*}
\begin{aligned}
\lambda_{L,j} &= 1-\delta\omega^{-j}+(\delta-1)\omega^{-(N-1)j} \\
&=1-\omega^{-(N-1)j}-\delta(\omega^{-j}-\omega^{-(N-1)j})
\\
&=2\sin^2\left(\frac{\pi j}{N}\right)+\mathrm{i}(2\delta-1)\sin\left(\frac{2\pi j}{N}\right),
\end{aligned}
\end{equation*}
\noindent
for all $0\leq j\leq N-1$, where $\omega=\exp\big(\frac{2\pi \mathrm{i}}{N}\big)$. Here, we denote $\lambda_{0}=\lambda_{L,0}=0$ as the only identically null eigenvalue. Due to the conjugate symmetry $\lambda_{L,j}=\lambda_{L,N-j}^{\mathrm{c}}$ for circulant matrices, and the fact that $\mathrm{Re}(\sqrt{z})=\mathrm{Re}(\sqrt{z^{\mathrm{c}}})$, we only have to consider up to $N/2$ eigenvalues if $N$ is even, and up to $(N-1)/2$ eigenvalues if $N$ is odd. We introduce the shorthand notation
\begin{equation*}
\Upsilon(z)\coloneq 2\sin^2\left(\frac{z}{2}\right)+\mathrm{i}(2\delta-1)\sin(z),\,\,\,\,\,\,\,\,\,\,\,\,\,\, \varrho(z)\coloneq \mathrm{Re}\big(\Upsilon(z)\big),\,\,\,\,\,\,\,\,\,\, \iota(z)\coloneq \Im\big(\Upsilon(z)\big),\,\,\,\,\,\,\,\,\,\,\,z\coloneq \frac{2\pi j}{N}, \,\,\,\,\,\,0< z\leq \pi.
\end{equation*}


Recall that the optimization problems \eqref{eq.constrainedoptimization} and \eqref{eq.sm.equivalentconstrainedopt} are equivalent for homogeneous vectors $\bm b = b\bm 1_N$. Thus, each Lyapunov exponent is given by
\begin{equation}
\Lambda_{+}(\hat{b},j) = \left(-\hat{b}+\mathrm{Re}\left(\sqrt{\hat{b}^2-\lambda_{L,j}})\right)\right),
\label{eq.lyapexp+}
\end{equation}
\noindent
for $0\leq j\leq\frac{N}{2}$, where we exclude the negative branch since it is not related to the largest Lyapunov exponent (which is our quantity of interest).
Let us define the continuous function $h(\hat{b},z)\coloneq-\hat{b}+\mathrm{Re}\left(\sqrt{\hat{b}^2-\Upsilon(z)}\right)$, which is a generalization of Eq.~\eqref{eq.lyapexp+}.
Now, we will prove that there exist $\hat b_2\geq\hat b_1$ such that: 1) the inequality $h(\hat{b},z)\leq 0$ (resp. $h(\hat{b},z)\geq 0$) holds for all $z$ and $\hat{b}\geq\hat{b}_{2}$ (resp. $\hat{b}\leq\hat{b}_{2}$); and 2) the inequality $\frac{\partial h(\hat{b},z)}{\partial z}< 0$ holds for all $z$ and $\hat{b}\geq \hat{b}_{1}$. These properties imply $\Lambda_{+}(\hat{b},j)> \Lambda_{+}(\hat{b},j+1)$, $\forall j$, in the region of $\hat{b}$ where $\Lambda_{+}(\hat{b},j)\leq 0$.

To prove property 1, we calculate the interval where $h(\hat{b},z)\leq 0$:
\begin{equation*}
\begin{aligned}
    0 &\geq -\hat{b}+\mathrm{Re}\Big(\sqrt{\hat{b}^2-\Upsilon}\Big), \\
    \hat{b} &\geq \frac{1}{\sqrt{2}}\sqrt{\sqrt{(\hat{b}^2-\varrho)^2+\iota^2}+\hat{b}^2-\varrho}, \\
    \hat{b}^2+\varrho &\geq \sqrt{(\hat{b}^2-\varrho)^2+\iota^2}, \\
    (\hat{b}^2+\varrho)^2-(\hat{b}^2-\varrho)^2 &\geq \iota^2, \\
    \hat{b}^2 &\geq \frac{\iota^2}{4\varrho}.
\end{aligned}
\end{equation*}
\noindent
Therefore, $\hat{b}_2 = \frac{\iota}{2\sqrt{\varrho}}$.

To prove property 2, we now calculate the sign of the partial derivative of $h(\hat{b},z)$ with respect to $z$:
\begin{equation*}
\frac{\partial}{\partial z} h(\hat{b},z) = \frac{1}{2\sqrt{2}}\frac{\frac{\frac{\partial \iota}{\partial z}\iota-\frac{\partial \varrho}{\partial z}(\hat{b}^2-\varrho)}{\sqrt{(\hat{b}^2-\varrho)^2+\iota^2}}-\frac{\partial \varrho}{\partial z}}{\sqrt{\sqrt{(\hat{b}^2-\varrho)^2+\iota^2}+\hat{b}^2-\varrho}}\stackrel{S}{=}\frac{\frac{\partial \iota}{\partial z}\iota-\frac{\partial \varrho}{\partial z}(\hat{b}^2-\varrho)}{\sqrt{(\hat{b}^2-\varrho)^2+\iota^2}}-\frac{\partial \varrho}{\partial z}\stackrel{S}{=}\frac{\partial \iota}{\partial z}\iota-\frac{\partial \varrho}{\partial z}\Big((\hat{b}^2-\varrho)+\sqrt{(\hat{b}^2-\varrho)^2+\iota^2}\,\Big),
\end{equation*}
\noindent
where $\stackrel{S}{=}$ denotes equality in sign. Given that $\frac{\partial \varrho}{\partial z}=\sin(z)\geq 0$ for all $z$, if there exists some $\hat{b}_1$ for which $\pdv{h}{z}< 0$ for all $z$, then it follows that $\pdv{h}{z}$ is also negative for all $\hat{b}>\hat{b}_{1}$ and for all $z$. At $\hat{b}=\hat{b}_2$, we obtain the following relation for all $z$:
\begin{equation*}
\begin{aligned}
    \pdv{}{z} h(\hat{b},z)_{\hat{b}=\hat{b}_{2}} &\stackrel{S}{=} \frac{\frac{\partial \iota}{\partial z}\iota-\frac{\partial \varrho}{\partial z}(\frac{\iota^2}{4\varrho}-\varrho)}{\sqrt{(\frac{\iota^2}{4\varrho}-\varrho)^2+\iota^2}}-\frac{\partial \varrho}{\partial z} \\
    &= \frac{4\frac{\partial \iota}{\partial z}\iota\varrho-\frac{\partial \varrho}{\partial z}(\iota^2-4\varrho^2)}{\sqrt{(\iota^2-4\varrho^2)^2+16\varrho^2\iota^2}}-\frac{\partial \varrho}{\partial z} \\
    &\stackrel{S}{=}4\frac{\partial \iota}{\partial z}\iota\varrho-2\frac{\partial \varrho}{\partial z}\iota^2 \\
    &=2(2\delta-1)^2\sin(z)\big(4\sin^2\left(\frac{z}{2}\right)\cos(z)-\sin^2(z)\big) \\
    &\stackrel{S}{=}4\sin^2\left(\frac{z}{2}\right)\cos(z)-\sin^2(z) \\
    &=-4\sin^4\left(\frac{z}{2}\right)< 0.
\end{aligned}
\end{equation*}
\noindent
Thus, we have $\hat{b}_2\geq \hat{b}_1$. As a result, for all $\hat{b}\geq\hat{b}_2$, it follows that $\Lambda_{+}(\hat{b},j)$ decreases as $j$ increases. As a result, in this region, we can express the Lyapunov exponents of the positive branch of the Jacobian matrix as follows:
\begin{equation}
\Lambda_{+}(\hat{b},j) = -\hat{b}+\mathrm{Re}\Big(\sqrt{\hat{b}^2-\lambda_{L,j}}\,\Big),
\label{eq.sm.eigJtoeigL}
\end{equation}
for $1\leq j\leq \frac{N}{2}$. Since the Laplacian eigenvalues $\{\lambda_{L,j}\}$ have their real parts sorted in ascending order, the Lyapunov exponents $\{\Lambda_{+}(x,j)\}$ are sorted in descending order when indexed by $j$.

From Eq.~\eqref{eq.sm.eigJtoeigL}, it thus follows that, in a neighborhood around the homogeneous optimal parameter vector $\bm b^{*}_{\mathrm{hom}} = b^{*}_{\mathrm{hom}} \bm 1_N$ (which necessarily satisfies $\Lambda_{\mathrm{max}}(\hat{b})<0$), the largest Lyapunov exponent can be expressed as 
\begin{equation}
\Lambda_{\mathrm{max}}(\hat{b})=-\hat{b}+\mathrm{Re}\Big(\sqrt{\hat{b}^2-\lambda_{L,1}}\Big)=-\hat{b}+\frac{1}{\sqrt{2}}\sqrt{\sqrt{(\hat{b}^2-\varrho)^2+\iota^2}+\hat{b}^2-\varrho},
\label{eq.sm.largestlyapexpcirculant} 
\end{equation}
\noindent
where $\varrho= 2\sin^2\big(\frac{\pi}{N}\big)$ and $\iota=(2\delta-1)\sin\big(\frac{2\pi}{N}\big)$.
Since $\iota\neq 0$ for $\delta\neq \frac{1}{2}$, the term inside the square root in Eq.~\eqref{eq.sm.largestlyapexpcirculant} is never negative, and thus $\Lambda_{\rm max}(\hat b)$ is differentiable in this neighborhood.
\end{proof}

Fig.~\ref{fig.circulant}C shows the fraction of directed circulation networks that exhibit differentiable homogeneous optima, plotted as a function of the edge density. Notably, sparser circulant networks are more likely to be differentiable than denser ones. This is in agreement with Lemma~\ref{lem.circ_diff_ass}, which states that every directed (first-neighbor) ring network is differentiable at the homogeneous optima. To compute the results in Fig.~\ref{fig.circulant}C, the circulant networks were generated by setting $fN$ values of the generating vector $\bm c$ to be nonzero (with random weight $c_{i} \sim \mathcal U[0,1]$), where $f\in[0,1]$ controls the edge density. We then assessed the differentiability of the networks by computing the maximum value of the numerical second derivative; if this value exceeds $2.5\times 10^5$, a network is classified as nondifferentiable. The threshold value was determined by analyzing the subclass of circulant matrices considered in Lemma~\ref{lem.circ_diff_ass}, which are guaranteed to be differentiable. We selected a value that correctly labeled more than $90\%$ of these networks as differentiable under the given numerical precision. Since this threshold is conservative, Fig.~\ref{fig.circulant}C provides a lower bound on the fraction of differentiable networks.

Having characterized the differentiability properties of systems with circulant networks in the vicinity of the homogeneous optimal point, we are now prepared to prove the following theorem on the existence of disorder-promoted stability:

\begin{theor}
Consider the optimization problem \eqref{eq.optimization} for the Jacobian matrix~\eqref{eq.jacobian}, and let $L$ be a non-Hermitian circulant matrix whose real parts of all eigenvalues are nonnegative. 
If $\Lambda_{\mathrm{max}}(\bm b)$ is differentiable at the homogeneous optimal point $\bm b_{\mathrm{hom}}^* =  b^{*}_{\mathrm{hom}}\bm 1_N$, then $\bm b_{\rm hom}^{*}$ is not a local optimum.
\label{thm.local_CSB}
\end{theor}

\begin{proof}
We recall that the global optimum of the constrained optimization problem~\eqref{eq.constrainedoptimization} is denoted as the homogeneous optimal point $\bm b_{\mathrm{hom}}^* = b^*_{\mathrm{hom}}\bm 1_N$ and define the index $j\coloneq\operatorname{arg \, max}_{k\neq \mathcal Z} \mathrm{Re}\big(\lambda_k(b^{*}_{\mathrm{hom}} \bm 1_N)\big)$ that relates a specific eigenvalue of $J$ to the largest Lyapunov exponent under the homogeneous optimal damping. At $\bm b = \bm b_{\mathrm{hom}}^*$, the Hessian matrix of the Lyapunov exponent is given by
\begin{equation*}
H_{\alpha \beta}=\frac{\partial^2\Lambda_{\mathrm{max}}}{\partial b_{\alpha}\partial b_{\beta}}=\frac{\partial^2\mathrm{Re}\left(\lambda_{j}\right)}{\partial b_{\alpha}\partial b_{\beta}}=\mathrm{Re}\left(\frac{\partial^2\lambda_{j}}{\partial b_{\alpha}\partial b_{\beta}}\right),
\end{equation*}
\noindent
where $1\leq \alpha,\beta\leq N$. Given the system's symmetries, $H$ is a circulant and Hermitian matrix. If at least one of the eigenvalues of $H$ is negative, then $\bm b^*_{\mathrm{hom}}$ cannot be a local optimum of the unconstrained problem \eqref{eq.optimization}. To prove this, it suffices to show that 
\begin{equation}
\exists (\alpha,\beta) \text{ s.t. }H_{\alpha\beta}>\max(H_{\alpha\alpha},H_{\beta\beta}),
\label{eq.sm.eigHcondition}
\end{equation}
\noindent
since we then have $\bm v^{\dagger}H\bm v<0$ for $v_{l}\coloneq\delta_{l,\alpha}-\delta_{l,\beta}$,, implying the existence of a descent direction. Once again, due to the symmetries of the system,  the diagonal elements of the matrix are all equal to each other, and hence condition \eqref{eq.sm.eigHcondition} simplifies to
\begin{equation}
\exists (\alpha,\beta) \text{ s.t. }H_{\alpha\beta}>H_{\beta\beta}.
\label{eq.sm.existencecondition}
\end{equation}
\noindent

For simplicity, we will consider the case where we have an identically null eigenvalue related to the eigenvector $\bm 1_N$, but the following derivations are still valid when $\mathcal Z=\emptyset$.
We adopt the following shorthand notation:
\begin{equation*}
C_{,\alpha}\coloneq\frac{\partial C}{\partial b_{\alpha}},\,\,\,\,\,\,\,C_{,\alpha\beta}\coloneq\frac{\partial C}{\partial b_{\alpha}\partial b_{\beta}},
\end{equation*}
\noindent
where $C$ can be a matrix, a vector, or a scalar, and the derivatives are only taken with respect to $b_{\alpha}$ and $b_{\beta}$. Using Eq.~\eqref{eq:quad_eig_prob} for $\bm b = b\bm 1_N$, we can express the eigenvalues of $J$ in the following order:
\begin{equation}
\lambda_{k}=
\begin{cases}
\frac{1}{2}\big(-b+\sqrt{b^2-4\lambda_{L,k}}\,\big) &\text{   for   } 1\leq k\leq N,
\\
\frac{1}{2}\big(-b-\sqrt{b^2-4\lambda_{L,k-N}}\,\big)&\text{  for  } N+1\leq k\leq 2N.\nonumber
\end{cases}
\end{equation}
\noindent
For each pair of eigenvalues, the respective right eigenvectors $\bm u$ of the quadratic problem \eqref{eq:quad_eig_prob} are the eigenvectors of the circulant matrix $L$. For all circulant matrices of size $N$, the eigenvectors can always be expressed as orthogonal $N$ discrete Fourier modes \cite{gray2006toeplitz}:
\begin{equation*}
(\bm u_{\alpha})_{\beta}\propto\exp(\frac{2\pi \mathrm{i}(\alpha-1)(\beta-1)}{N}).
\end{equation*}
\noindent
Due to their orthogonality, there is only one other eigenvector $\bm u_{j'}$ such that $\bm u_{j}^{\dagger} \bm u_{j'}\neq 0$, and its index is $j'=j+N$. From now on, similarly to Proposition~\ref{prop.theor_hermit_diff}, we apply the normalization
\begin{equation*}
\bm v_{j}^{\dagger}\bm u_{j}=\frac{1}{2\lambda_{j}+b},\,\,\,\forall j,
\end{equation*}
\noindent
where $\bm u_j^\dagger\bm u_j=1$ and $\bm v_j$ denotes the left eigenvectors.
As a result, we have $\bm v_{j}^{\dagger}=\frac{\bm u_{j}^{\dagger}}{2\lambda_{j}+b},\,\,\,\forall j$.

We can then write the first derivatives of the eigenvector $\bm u_{j}$ as \cite{adhikari2001eigenderivative}
\begin{equation*}
\bm u_{j,\alpha}=\sum_{k=1}^{2N}a_{jk}^{(\alpha)}\bm u_{k},
\end{equation*}
\noindent
where the coefficients $a_{jk}^{(\alpha)}$ are given by
\begin{equation*}
\begin{aligned}
&a_{jk}^{(\alpha)}=-\frac{\bm v_{k}^{\dagger}\big(\lambda_{j}^{2}(I_{N})_{,\alpha}+\lambda_{j}B_{,\alpha}+L_{,\alpha}\big)\bm u_{j}}{\lambda_{j}-\lambda_{k}}=\frac{\lambda_{j}\bm v_{k}^{\dagger}R_{\alpha}\bm u_{j}}{-\lambda_{j}+\lambda_{k}}=\frac{\lambda_{j}(v_{k}^{\mathrm{c}})_{\alpha}(u_{j})_{\alpha}}{\lambda_{k}-\lambda_{j}}=\frac{\lambda_{j}\exp\big(\frac{2\pi \mathrm{i}(\alpha-1)(j-k)}{N}\big)}{N(2\lambda_{k}+b)(\lambda_{k}-\lambda_{j})},\,\,\,\forall j\neq k,
\\&a_{jj}^{(\alpha)}=-\bm u_{j}^{\dagger}\sum_{\substack{k=1 \\ k\neq j}}^{2N}a_{jk}^{(\alpha)}\bm u_{k}=-a_{jj'}=\frac{\lambda_{j}}{N(2\lambda_{j'}+b)(\lambda_{j}-\lambda_{j'})}.\nonumber
\end{aligned}
\end{equation*}
\noindent
Here, the diagonal elements of $a^{(\alpha)}$ are determined by our normalization choice, and we used the fact that 
\begin{equation*}
\lambda_{j,\alpha}=\lambda_{j,\beta}=-\lambda_{j} \bm v_{j}^{\dagger}R_{\alpha}\bm u_{j}=-\lambda_{j} \bm v_{j}^{\dagger}R_{\beta}\bm u_{j}=-\lambda_{j}(v_{j}^{\mathrm{c}})_{\alpha}(u_{j})_{\alpha}=-\frac{\lambda_{j}}{N(2\lambda_{j}+b)}.
\end{equation*}

The second derivatives of the eigenvalues can be calculated with the general formula \cite{adhikari2001eigenderivative}
\begin{equation}
\begin{aligned}
\lambda_{j,\alpha\beta} =& -\frac{1}{\bm v_{j}^{\dagger}G_{j}\bm u_{j}}\Big(\underbrace{\bm v_{j}^{\dagger}\big(F_{j,\alpha\beta}+\lambda_{j,\alpha}\Tilde{G}_{j,\beta}+\lambda_{j,\beta}\tilde{G}_{j,\alpha}\big)\bm u_{j}+2\lambda_{j,\alpha}\lambda_{j,\beta}\bm v_{j}^{\dagger}\bm u_{j}+\lambda_{j,\alpha}\bm v_{j}^{\dagger}G_{j}\bm u_{j,\beta}+\bm v_{j}^{\dagger}\lambda_{j,\beta}G_{j}\bm u_{j,\alpha}}_{\text{independent of the choice of } (\alpha,\beta)}\\ 
&+\underbrace{\bm v_{j}^{\dagger}\tilde{F}_{j,\alpha}\bm u_{j,\beta}+\bm v_{j}^{\dagger}\tilde{F}_{j,\beta}\bm u_{j,\alpha}}_{\text{dependent on the choice of } (\alpha,\beta)}\Big),
\end{aligned}
\label{expr_sec_derivative}
\end{equation}
\noindent
where here we obtain $F_{j,\alpha\beta}\coloneq\lambda_{j}B_{,\alpha\beta}+L_{,\alpha\beta}=0$, $\Tilde{G}_{j,\alpha}\coloneq B_{,\alpha}=R_{\alpha}$, $\Tilde{F}_{j,\alpha}\coloneq\lambda_{j}^2(I_{N})_{,\alpha}+\lambda_{j}B_{,\alpha}+L_{,\alpha}=\lambda_{j}R_{\alpha}$, $G_{j,\beta}\coloneq\Tilde{G}_{j,\beta}+2\lambda_{j,\beta}I_{N}=R_{\beta}+2\lambda_{j,\beta}I_N$, $G_{j}\coloneq 2\lambda_{j}I_{N}+B=(2\lambda_{j}+b)I_N$, and $\bm v_{j}^{\dagger}G_{j}\bm u_{j}=I_N$.

In what follows, we consider a single pair $(\alpha,\beta)$ and focus only on the terms that are dependent on a choice of $(\alpha,\beta)$ in Eq.~\eqref{expr_sec_derivative}. 
For the diagonal elements of the Hessian matrix, we have that these terms are given by 
\begin{equation}
\begin{aligned}
(v_{j}^{\mathrm{c}})_{\alpha}(u_{j,\alpha})_{\alpha} &= \sum_{k=1}^{2N}a_{jk}^{(\alpha)}(v_{j}^{\mathrm{c}})_{\alpha}(u_{k})_{\alpha} \\
&=a_{jj}^{(\alpha)}(v_{j}^{\mathrm{c}})_{\alpha}(u_{j})_{\alpha}+a_{jj'}^{(\alpha)}(v_{j}^{\mathrm{c}})_{\alpha}(u_{j'})_{\alpha}+\sum_{k=1,k\neq j,k\neq j'}^{2N}\frac{\lambda_{j}\exp(\frac{2\pi \mathrm{i}(\alpha-1)(j-k)}{N})}{N(2\lambda_{k}+b)(\lambda_{k}-\lambda_{j})}(v_{j}^{\mathrm{c}})_{\alpha}(u_{k})_{\alpha} 
\\
&=\sum_{\substack{k=1 \\ k\neq j,j'}}^{2N}\frac{\lambda_{j}}{N^2(\lambda_{k}-\lambda_{j})(2\lambda_{j}+b)(2\lambda_{k}+b)}\\
&\eqqcolon\sum_{\substack{k=1 \\ k\neq j,j'}}^{2N}\frac{A_{k}}{N^2}.
\end{aligned}
\label{eq.sm.diag_term}
\end{equation}
\noindent
For the nondiagonal elements, they are given by
\begin{equation}
\begin{aligned}
(v_{j}^{\mathrm{c}})_{\beta}(u_{j,\alpha})_{\beta}+(v_{j}^{\mathrm{c}})_{\alpha}(u_{j,\beta})_{\alpha} &= \sum_{k=1}^{2N}\big(a_{jk}^{(\alpha)}(v_{j}^{\mathrm{c}})_{\beta}(u_{k})_{\beta}+a_{jk}^{(\beta)}(v_{j}^{\mathrm{c}})_{\alpha}(u_{k})_{\alpha}\big) \\
&=\sum_{\substack{k=1 \\ k\neq j,j'}}^{2N}\big(a_{jk}^{(\alpha)}(v_{j}^{\mathrm{c}})_{\beta}( u_{k})_{\beta}+a_{jk}^{(\beta)}(v_{j}^{\mathrm{c}})_{\alpha}(u_{k})_{\alpha}\big)  \\
&=\sum_{\substack{k=1 \\ k\neq j,j'}}^{2N}\frac{\lambda_{j}\Big(\exp\big(\frac{2\pi\mathrm{i}}{N}(j-k)(\alpha-\beta)\big)+\exp\big(\frac{2\pi\mathrm{i}}{N}(-j+k)(\alpha-\beta)\big)\Big)}{N^2(\lambda_{k}-\lambda_{j})(2\lambda_{j}+b)(2\lambda_{k}+b)} 
\\
&=\sum_{\substack{k=1 \\ k\neq j,j'}}^{2N}\frac{2A_{k}\cos\left(\frac{2\pi}{N}(j-k)(\alpha-\beta)\right)}{N^2}.
\end{aligned}
\label{eq.sm.non_diag_term}
\end{equation}
\noindent
Because the $(\alpha,\beta)$-independent terms in Eq.~\eqref{expr_sec_derivative} cancel, the difference between the nondiagonal and diagonal elements of $H$ is proportional to the difference between Eqs.~\eqref{eq.sm.non_diag_term} and \eqref{eq.sm.diag_term}:
\begin{equation}
\begin{aligned}
\lambda_{j,\alpha\beta}-\lambda_{j,\alpha\alpha} &= \lambda_{j}\big(2(v_{j}^{\mathrm{c}})_{\alpha}(u_{j,\alpha})_{\alpha}-(v_{j}^{\mathrm{c}})_{\beta}(u_{j,\alpha})_{\beta}-(v_{j}^{\mathrm{c}})_{\alpha}(u_{j,\beta})_{\alpha}\big)
\\
&=\frac{2}{N^2}\sum_{\substack{k=1 \\ k\neq j,j'}}^{2N}\lambda_{j}A_{k}\Big(1-\cos\Big(\frac{2\pi}{N}(j-k)(\alpha-\beta)\Big)\Big) \\
&=\frac{4}{N^2}\sum_{\substack{k=1 \\ k\neq j}}^{N}\lambda_{j}(A_{k}+A_{k+N})\sin^2\Big(\frac{\pi}{N}(j-k)(\alpha-\beta)\Big). 
\end{aligned}
\label{inter_sum}
\end{equation}

Proving that the real part of Eq.~\eqref{inter_sum} is positive will thus imply that the Hessian matrix satisfies $\bm v^\dagger H \bm v < 0$ for some $\bm v$.
To determine the sign of Eq.~\eqref{inter_sum}, it suffices to analyze the sign of $\lambda_{j}(A_{k}+A_{k+N})$, given by
\begin{equation}
\lambda_{j}(A_{k}+A_{k+N})=\frac{\lambda_{j}^2}{(-\lambda_{L,k}+\lambda_{L,j})\sqrt{b^2-4\lambda_{L,j}}}.\nonumber
\end{equation}
\noindent
At the homogeneous optimal point, it follows from Proposition~\ref{prop.theor_hermit_diff} that $\frac{\lambda_{j}}{\sqrt{b^2-4\lambda_{L,j}}}=\mathrm{i}\gamma$, where $\gamma$ is real and nonzero. Therefore,
\begin{equation}
\lambda_{j}(A_{k}+A_{k+N})=\gamma^2\frac{\sqrt{b^2-4\lambda_{L,j}}}{\lambda_{L,k}-\lambda_{L,j}}.
\label{eq.sm.AkAkN}
\end{equation}
\noindent
Substituting Eq.~\eqref{eq.sm.AkAkN} into Eq.~\eqref{inter_sum} yields
\begin{equation}
\label{final_sum}
\lambda_{j,\alpha\beta} -\lambda_{j,\alpha\alpha} = \frac{4\gamma^2}{N^2}\sum_{\substack{k=1 \\ k\neq j}}^{N}\underbrace{\frac{\varrho_{j}+\mathrm{i}\iota_{j}}{\lambda_{L,k}-\lambda_{L,j}}}_{\Xi_k}\sin^2\left(\frac{\pi}{N}(k-j)(\alpha-\beta)\right),
\end{equation}
\noindent
where $\varrho_{j}\coloneq\mathrm{Re}\big(\sqrt{b^2-4\lambda_{L,j}}\,\big)$ and $\iota_{j}\coloneq\mathrm{Im}\big(\sqrt{b^2-4\lambda_{L,j}}\,\big)$. Based on the ordering of eigenvalues, it follows that $\varrho_{j}\geq \varrho_{k}$, $\forall k>1$. To analyze the sign of the real part of Eq.~\eqref{final_sum}, we consider the following coefficient:
\begin{equation}
4\Xi_k^{-1} = \frac{4\lambda_{L,k}-4\lambda_{L,j}}{\varrho_{j}+\mathrm{i}\iota_{j}}=\frac{-b^2+4\lambda_{L,k}+b^2-4\lambda_{L,j}}{\varrho_{j}+\mathrm{i}\iota_{j}}=\varrho_{j}+\mathrm{i}\iota_{j}-\frac{b^2-4\lambda_{L,k}}{\varrho_{j}+\mathrm{i}\iota_{j}}=\varrho_{j}+\mathrm{i}\iota_{j}-\frac{(\varrho_{k}+\mathrm{i}\iota_{k})^2}{\varrho_{j}+\mathrm{i}\iota_{j}},
\label{sign_analysis}
\end{equation}
\noindent
and we note that
\begin{equation}
\begin{aligned}
\mathrm{Re}\Big (\varrho_{j}+\mathrm{i}\iota_{j}-\frac{(\varrho_{k}+\mathrm{i}\iota_{k})^2}{\varrho_{j}+\mathrm{i}\iota_{j}}\Big) &= \varrho_{j}-\frac{\varrho_{j}(\varrho_{k}^2-\iota_{k}^2)+2\varrho_{k}\iota_{k}\iota_{j}}{\varrho_{j}^2+\iota_{j}^2} \\
&\stackrel{S}{=}\varrho_{j}^3+\varrho_{j}\iota_{j}^2-\varrho_{k}^2\varrho_{j}+\iota_{k}^2\varrho_{j}-2\varrho_{k}\iota_{k}\iota_{j}\\
&\geq \varrho_{j}\iota_{j}^2+\iota_{k}^2\varrho_{j}-2\varrho_{k}\iota_{k}\iota_{j} \\
&=\varrho_{j}(\iota_{k}-\iota_{j})^2-2\iota_{k}\iota_{j}(\varrho_{k}-\varrho_{j}) \\
&\geq 0.
\end{aligned}
\label{eq.sm.realpartofexp}
\end{equation}
\noindent
The last inequality in Eq.~\eqref{eq.sm.realpartofexp} is immediate when $\text{sign}(\iota_{k})=\text{sign}(\iota_{j})$, and it follows from the following when $\text{sign}(\iota_{k})\neq\text{sign}(\iota_{j})$:
\begin{equation*}
\varrho_{j}(\iota_{k}-\iota_{j})^2-2\iota_{k}\iota_{j}(\varrho_{k}-\varrho_{j}) \geq -2\varrho_{j}\iota_{k}\iota_{j}-2\iota_{k}\iota_{j}(\varrho_{k}-\varrho_{j})\geq 0.
\end{equation*}
\noindent
Therefore, we proved that each coefficient $\Xi_k$ in Eq.~\eqref{final_sum} has a positive real part, except for $k=1$. To conclude the proof, we now verify that $\mathrm{Re}(\Xi_1+\Xi_{k})>0$ for at least one $k\neq 1$. Since one can always choose a pair $(\alpha,\beta)$ such that $\sin^2\big(\frac{\pi x}{N}(\alpha-\beta)\big)\leq \sin^2\big(\frac{\pi y}{N}(\alpha-\beta)\big)$ for any pair $1\leq x,y\leq N-1$, we can thus choose a pair $(\alpha,\beta)$ such that $\mathrm{Re}(\lambda_{j,\alpha\beta} - \lambda_{j,\alpha\alpha})>0$. Choosing $k$ such that $\lambda_{L,k}^{\mathrm{c}}=\lambda_{L,j}$ yields 
\begin{equation}
\begin{aligned}
\Xi_1+\Xi_{k}=\frac{\varrho_{j}+\mathrm{i}\iota_{j}}{\lambda_{L,j}^{\mathrm{c}}-\lambda_{L,j}}-\frac{\varrho_{j}+\mathrm{i}\iota_{j}}{\lambda_{L,j}}&=\frac{\varrho_{j}+\mathrm{i}\iota_{j}}{\lambda_{L,j}^{\mathrm{c}}-\lambda_{L,j}}\left(2-\frac{\lambda_{L,j}^{\mathrm{c}}}{\lambda_{L,j}}\right)
\\&=\underbrace{\frac{\varrho_{j}+\mathrm{i}\iota_{j}}{\lambda_{L,j}^{\mathrm{c}}-\lambda_{L,j}}}_{\chi_1}\underbrace{\left(\frac{\mathrm{Re}(\lambda_{L,j})^2+3\mathrm{Im}(\lambda_{L,j})^2}{\mathrm{Re}(\lambda_{L,j})^2+\mathrm{Im}(\lambda_{L,j})^2}+\frac{2\mathrm{i}\mathrm{Re}(\lambda_{L,j})\mathrm{Im}(\lambda_{L,j})}{\mathrm{Re}(\lambda_{L,j})^2+\mathrm{Im}(\lambda_{L,j})^2}\right)}_{\chi_2}.
\end{aligned}
\label{final_sum_circ}
\end{equation} 
\noindent
Then, $\mathrm{Re}(\chi_1)\mathrm{Re}(\chi_2)-\Im(\chi_1)\Im(\chi_2)>0$. This inequality implies that Eq.~\eqref{final_sum_circ} has a positive real part, which leads to Eq.~\eqref{eq.sm.existencecondition}.
\end{proof}

\begin{corol}
Consider the optimization problem \eqref{eq.optimization} for the Jacobian matrices
\begin{equation*}J_{1}=\begin{bmatrix}0&I_N\\-L&-KB\end{bmatrix}\,\,\,\,\,\,\,\text{and}\,\,\,\,\,\,\,\,\,J_{2}=\begin{bmatrix}K&I_N\\-L&-B\end{bmatrix},
\end{equation*}

\noindent
where $L$ is a non-Hermitian circulant matrix and $K$ is an all-to-all Laplacian matrix.
If $\Lambda_{\mathrm{max}}(\bm b)$ is differentiable at the homogeneous optimal point $\bm b_{\mathrm{hom}}^* =  b^{*}_{\mathrm{hom}}\bm 1_N$, then $\bm b_{\rm hom}^{*}$ is not a local optimum.
\label{cor.dpsgeneraljac}
\end{corol}

\begin{proof}
For homogeneous damping $\bm b = b\bm 1_N$, we have that the eigenvalues associated with matrices $J_1$ and $J_2$ are respectively given by
\begin{equation}
\lambda_{j,\pm}^{(1)}=\frac{-b\lambda_{K,j}\pm\sqrt{b^2\lambda_{K,j}^2-4\lambda_{L,j}}}{2}\,\,\,\,\,\text{and}\,\,\,\,\,\,\lambda_{j,\pm}^{(2)}=\frac{-b+\lambda_{K,j}\pm\sqrt{(b+\lambda_{K,j})^2-4\lambda_{L,j}}}{2},\nonumber
\end{equation}

\noindent
where we used the fact that $K$ and $L$ admit a common set of orthogonal eigenvectors since they are both circulant matrices of the same size. Note that each eigenvalue of $J_1$ and $J_2$ is a function of a single, unique eigenvalue of $L$. In what follows, we only consider $J_1$ since the proof for $J_2$ is very similar. The derivatives of the eigenvalues of $J_1$ are
\begin{equation*}
\lambda_{j,\alpha}^{(1)}=-\bm v_{j}^{\dagger}(\lambda_{j}^2I_{N}+\lambda_{j}KB_{,\alpha}+L_{,\alpha})\bm u_{j}=-\lambda_{j}\bm v_{j}^{\dagger}KR_{\alpha}\bm u_{j}=-\lambda_{j}\lambda_{K,j}(v_{j}^{\mathrm{c}})_{\alpha}(u_{j})_{\alpha}.
\end{equation*}

\noindent
Applying the normalizations $\bm v_{j}^{\dagger}=\frac{\bm u_{j}^{\dagger}}{2\lambda_{j}+\lambda_{K,j}b}$ and $\bm u_{j}^{\dagger}\bm u_{j}=1$, we can derive the first and second derivatives of the eigenvectors as in the proof of Theorem~\ref{thm.local_CSB} and show that the Hessian matrix has negative eigenvalues.
\end{proof}

\begin{remark} 
Corollary~\ref{cor.dpsgeneraljac} shows rigorously that disorder-promoted stability extends beyond the specific Jacobian structure \eqref{eq.jacobian} considered in Theorem~\ref{thm.local_CSB}, and that it can also be established for broader classes of coupled dynamical systems. In particular, the Jacobian matrix $J_1$ arises in mechanical systems with velocity-dependent damping (scaled by the spatial stiffness matrix $K$), whereas $J_2$ corresponds to certain classes of phase-amplitude oscillators like the Stuart-Landau model.
\end{remark}

\bigskip
\section{Relationship between mode mixing and disorder-promoted stability}
\label{sec.sm.modemix}

For networks of second-order dynamical systems described by the Jacobian matrix \eqref{eq.jacobian}, we can explicitly characterize how heterogeneity in the damping coefficients $b_i$ affects the eigenstructure of $J$\textemdash in particular, the relationship between its eigenvalues and eigenvectors. The analysis reveals how the largest Lyapunov exponent and, consequently, the system's stability are affected by changes in damping.

When the damping is homogeneous, Eq.~\eqref{eq.sm.homogeneouseigproblem} shows that the Jacobian eigenvalues can be grouped into $N$ pairs, each corresponding uniquely to one eigenvalue of $L$. Accordingly, the Jacobian eigenvectors can also be partitioned into $N$ corresponding pairs, and we refer to each such pair as a \textit{network mode}. When the damping becomes heterogeneous, the correspondence between the eigenvectors of the Jacobian and those of the Laplacian becomes less direct due to \textit{mode mixing}, as each Jacobian eigenvector generally becomes a linear combination of multiple Laplacian eigenvectors. To show this, note that the eigenvector pair $\bm{v}_{J,i,\pm}$ is related to an eigenvector $\bm{u}_{i}$ of the quadratic problem \eqref{eq:quad_eig_prob} according to
\begin{equation}
\bm{v}_{J,i,\pm}=\begin{bmatrix}\bm{u}_{i}\\\lambda_{i,\pm} \bm{u}_{i}\end{bmatrix}, \quad \text{for} \,\, i = 1,\ldots,N.
\label{Jac_quad_rel}
\end{equation}

\noindent
Expanding the Laplacian matrix in its eigenbasis in Eq.~\eqref{quadr_eig} yields

\begin{equation}
\lambda_{i,\pm}=\frac{-\bm u_{i}^\dagger B \bm u_{i}\pm\sqrt{(\bm u_{i}^\dagger B\bm u_{i})^2-4\sum_{j=1}^{N}\lambda_{L,j}\abs{\nu_{j}}^2}}{2\bm u_{i}^\dagger \bm u_{i}},
\end{equation}

\noindent
where $\bm{v}_{L,j}$ is an eigenvector of the Laplacian and $\bm{u}_{i}=\sum_{j=1}^{N}\nu_{j}\bm{v}_{L,j}$. Thus, any deviation from homogeneous damping (i.e., when $B\neq bI_N$) causes each $\bm u_i$ to draw components from multiple Laplacian modes, thereby inducing mode mixing. The following result shows that, for homogeneous damping, the network modes are always orthogonal.

\begin{prop}
{Consider the Jacobian matrix \eqref{eq.jacobian}.}
If the Laplacian matrix $L$ is diagonalizable, then the matrices
\begin{equation*}
J=\begin{bmatrix}
0_{N}&I_{N}\\-L&-B
\end{bmatrix}\,\,\, \text{and} \,\,\,J'=\begin{bmatrix}
0_N&I_{N}\\-D_{L}&-U^{-1}BU\end{bmatrix}
\end{equation*}
\noindent
are similar, where $U$ is the matrix that diagonalizes the Laplacian matrix $L$ and $D_{L} = U^{-1}LU$. 

Furthermore, if  $B = b I_N$ for some scalar $b$ and $L$ has $N$ distinct eigenvalues and orthogonal eigenvectors, then any pair $(i,j)$ of eigenvectors $(\bm v_{J,i},\bm v_{J,j})$ of $J$ satisfies
\begin{equation}
\bm v^{\dagger}_{J,i}\bm v_{J,j}=\delta_{ij}+c\delta_{jj'}, \,\,\,\, \text{for} \,\,\, j'\neq i \,\,\, \text{and} \,\,\, 0<c\leq 1.
\end{equation}
\label{prop.JJ'}
\end{prop}

\begin{proof}
To show the similarity between the matrices $J$ and $J'$, consider the following transformation:
\begin{equation*}
(I_{2}\otimes U)^{-1}J(I_{2}\otimes U)=\begin{bmatrix}
0&I_{N}\\-D_{L}&-U^{-1}BU
\end{bmatrix}=J',
\end{equation*}
\noindent
where $\otimes$ denotes the Kronecker product. Thus, $J$ and $J'$ share the same eigenvalues, and the eigenvectors are related by $(I_{2}\otimes U)^{-1}\bm v_{J,k}=\bm v_{J',k}$ for all $1\leq k\leq 2N$. 

Now, suppose $B = bI_N$ and let $\bm v_{J',k} \coloneq (\bm w_{J',k}, \bm u_{J',k})$ be an eigenvector of $J'$, where $\bm w_{J',k},\bm u_{J',k}\in\R^N$. The eigenvalue equation of $J'$ can be equivalently expressed as
\begin{equation}
\begin{aligned}
    \lambda_{J',k}\bm u_{J',k} &= \bm w_{J',k}, \\
    \lambda_{J',k} \bm u_{J',k} &= -D_{L}\bm w_{J',k}-b\bm u_{J',k}.
\end{aligned}
\label{eq.sm.penfinal_subeq}
\end{equation}
\noindent
The solution for Eq.~\eqref{eq.sm.penfinal_subeq} is $\bm w_{J',k}=\bm e_{j}$, $\bm u_{J',k}=\lambda_{J',j,\pm}\bm e_{j}$, and $\lambda_{J',k}=\lambda_{J',j,\pm}=\frac{1}{2}\big(-b\pm\sqrt{b^2-4\lambda_{L,j}}\,\big)$, where $\bm e_{k}$ is the $k$th vector in the canonical basis and $j=\frac{k+k\!\!\mod 2}{2}$. This solution is unique, as the eigenvalues of $L$ are all distinct.

In conclusion, $(\bm v_{J',k,s_{1}})^{\dagger}\bm v_{J',j,s_{2}}=(1-c)\delta_{jk}\delta_{s_{1}s_{2}}+c\delta_{jk}$, where $0<c\leq 1$ and $s_{1},s_{2}\in\{+,-\}$. Given that the similarity matrix $I_{2}\otimes U$ is a unitary matrix, the orthogonality relations are maintained for the eigenvectors of $J$.
\end{proof}

\begin{remark}
The structure of matrix $J'$ highlights how heterogeneity in the damping vector $\bm b$ induces mode mixing. In the homogeneous case ($B=bI_N$), Proposition~\ref{prop.JJ'} establishes that every eigenvector of $J$ is orthogonal to all but one other eigenvector, showing that the \textit{network modes} are fully decoupled. However, when $\bm b$ is heterogeneous, the submatrix $U^{-1} B U$ introduces an off-diagonal structure in $J'$, leading to coupling between these network modes (i.e., mode mixing).
\end{remark}

To understand the next theorem, it is important to note that nondifferentiable points $\Lambda_{\mathrm{max}}(\bm b,A)$ can arise when multiple eigenvalues have real part equal to the largest Lyapunov exponent. Following Refs.~\cite{lewis2003mathematics,burke2001optimal}, we refer to these eigenvalues as  \textit{active} eigenvalues, and the cardinality of this set as the \textit{activity} of the system (for a given choice of parameters).  


\begin{defin}
    The activity set of a matrix $J$ is defined as $\mathcal A\big(J(\bm b)\big) = \{ \lambda_i(\bm b) \, : \, \mathrm{Re}(\lambda_i) = \Lambda_{\mathrm{max}}\big(J(\bm b)\big) \}$, where $\lambda_i$ is an eigenvalue of $J$.
\end{defin}

\begin{theor}
Consider the Jacobian matrix \eqref{eq.jacobian} evaluated at a global optimum $\bm b^*$ of Eq.~\eqref{eq.optimization}, and suppose that $L$ is Hermitian. Then, at least two eigenvectors of the Jacobian associated with eigenvalues in $\mathcal A$ are parallel.
\label{thm.mode_mixing}
\end{theor}

\begin{proof}
It follows from Proposition~\ref{prop.theor_hermit_diff} that the largest Lyapunov exponent $\Lambda_{\mathrm{max}}$ is a nondifferentiable function at $\bm b^*$. There are two conditions in which $\Lambda_{\mathrm{max}}$ is nondifferentiable:
\begin{enumerate}
\item if at least two eigenvectors in $\mathcal A$ are parallel;

\item if a single eigenvalue has a point of nondifferentiability at $\bm b^*$.
\end{enumerate}

\noindent
To prove that condition 2 implies condition 1, we use Eq.~\eqref{quadr_eig} to express
\begin{equation}
\lambda_{i,\pm}=\frac{-\bm u_{i}^\dagger B \bm u_{i}+\sqrt{(\bm u_{i}^\dagger B\bm u_{i})^2-4\bm u_{i}^\dagger L \bm u_{i}}}{2\bm u_{i}^\dagger \bm u_{i}},
\end{equation}

\noindent
where $\bm u_{i}$ satisfies Eq.~\eqref{eq:quad_eig_prob} for $\lambda_{i,\pm}$. The eigenvectors of $J$ can be related to $\bm{u}_i$ via Eq.~\eqref{Jac_quad_rel}. Therefore, an eigenvalue $\lambda_{i,\pm}$ is nondifferentiable only when $(\bm u_{i}^\dagger B\bm u_{i})^2=4\bm u_{i}^\dagger L \bm u_{i}$. Since this condition is equivalent to $\lambda_{i,+}=\lambda_{i,-}$, it follows from Eq.~\eqref{Jac_quad_rel} that $v_{J,i,+}=v_{J,i,-}$. Therefore, at least one pair of eigenvectors is parallel.
\qedhere
\end{proof}

Theorem~\ref{thm.mode_mixing} identifies a sufficient condition for disorder-promoted stability: the emergence of mode mixing, characterized by parallelism among eigenvectors associated with active eigenvalues. Importantly, mode mixing between more than two eigenvectors can only be satisfied at heterogeneous damping configurations, since the network modes are necessarily orthogonal at a homogeneous damping configuration for diagonalizable matrices $L$ (Proposition~\ref{prop.JJ'}). Thus, provided that the eigenvalues of $L$ are distinct, mode mixing becomes increasingly more likely when the system's activity is higher, as higher activity can increase eigenvector interactions under heterogeneous damping. In the opposite extreme, where $L$ has only one distinct eigenvalue, it can be shown that the optimal parameters must be homogeneous (even when $L$ is nondiagonalizable). These insights motivate the following conjecture, which posits that heterogeneity in the damping vector improves stability only when $L$ possesses a sufficiently nondegenerate eigenvalue spectrum.

\begin{conjec}
Let $\mathcal L = \{\lambda_{L,i} \, : \, \lambda_{L,i} \neq 0\}$ be the set of nonzero eigenvalues of $L$, where $\lambda_{L,1}$ is the smallest eigenvalue in this set. It holds that
\begin{equation}
    \frac{\Lambda_{\mathrm{max}}\big(J(\bm b^*;A)\big) - \Lambda_{\mathrm{max}}\big(J(\bm b_{\mathrm{hom}}^*;A)\big)}{\abs{\Lambda_{\mathrm{max}}\big(J(\bm b_{\mathrm{hom}}^*;A)\big)}} < 0
\label{eq.sm.conjecture}
\end{equation}
only if the algebraic multiplicity of $\lambda_{L,1}$ is one. Moreover, the magnitude of the normalized difference in Eq.~\eqref{eq.sm.conjecture} increases with the relative distance between elements in the set $\mathcal L$.
\label{conj.activity}
\end{conjec}

Although we have not been able to prove this conjecture, our numerical analysis on model-generated and empirical networks has supported it thus far. To illustrate Theorem~\ref{thm.mode_mixing} and Conjecture~\ref{conj.activity}, we present the following example. 

\begin{figure}
\centering
\includegraphics[width=0.8\textwidth]{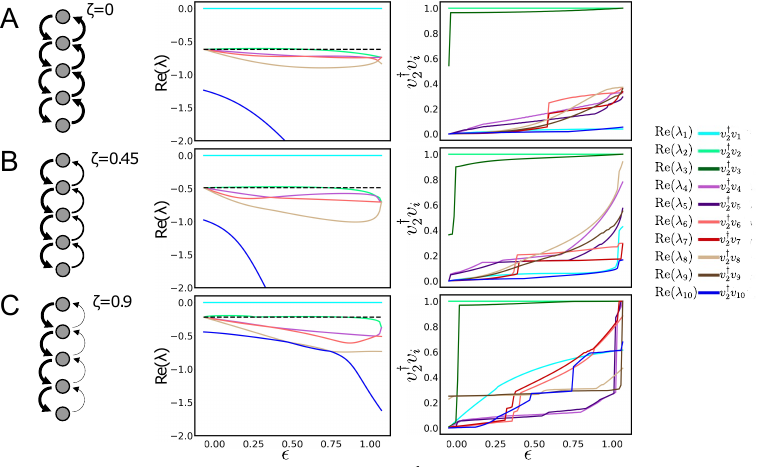}
\caption{\label{fig.dotproduct} \textbf{Mode mixing in chain networks.} (\textbf{A}--\textbf{C}) Real part of eigenvalues (middle panels) and dot product between eigenvectors (right panels) of the Jacobian matrix as a function of $\epsilon$. The control parameter in the Laplacian matrix is $\zeta=0$~(\textbf{A}), $\zeta=0.45$~(\textbf{B}), and $\zeta=0.9$~(\textbf{C}), as illustrated by the corresponding network topologies (left images). Some eigenvalues have identical real parts and are thus superimposed, and all dot products are computed with respect to the same reference eigenvector (i.e., $\bm v_2^\dagger \bm v_i$, for $i=1,\ldots,10$). For reference, the dashed lines in the middle panels show the largest Lyapunov exponent corresponding to the homogeneous optimum.}
\end{figure}

\begin{example}
\label{examp.modemixing}
Let $N = 5$ and consider the Laplacian matrix
\begin{equation*}
L=\begin{bmatrix}1-\zeta&-1+\zeta&0&0&0\\-1&2&-1&0&0\\0&-1&2&-1&0\\0&0&-1&2&-1\\0&0&0&-1+\zeta&1-\zeta\end{bmatrix},
\label{lap_example}
\end{equation*}
\noindent
which represents a chain of coupled oscillators with free boundary conditions. The parameter $\zeta$ allows us to control the relative spacing between the eigenvalues of $L$ while preserving the system symmetries. We quantify this relative spacing as $\Delta \lambda\coloneq\abs{\frac{\langle\delta\lambda\rangle}{\lambda_{L,1}}}$, where $\langle\delta\lambda\rangle\coloneq \frac{1}{|\mathcal L|}\sum_{i=1}^{|\mathcal L|}(\lambda_{L,i+1}-\lambda_{L,i})$. We consider $\zeta\in\{0,0.45,0.9\}$, which leads to the respective relative spacings  $\Delta \lambda\in\{1.89, 2.95, 14.06\}$.

The damping vector is parameterized as
\begin{equation*}
\bm b(\epsilon) \coloneq \bm b_{\text{hom}}^* + \epsilon(\bm b^* -\bm b_{\text{hom}}^*),
\end{equation*}

\noindent
where $\bm b^*$ denotes the global optimal solution of the optimization problem \eqref{eq.optimization}, $\bm b_{\text{hom}}^*$ denotes the homogeneous optimal solution of problem \eqref{eq.constrainedoptimization}, and $0\leq\epsilon\leq 1$.
Fig.~\ref{fig.dotproduct} shows that, as $\epsilon$ increases (i.e., as the system transitions from the homogeneous to the heterogeneous optimum), the dot products between eigenvectors tend to increase, indicating more alignment. At the heterogeneous optimal point ($\epsilon = 1$), a pair of eigenvectors of $J$ coincides, leading to strong mode mixing. As proposed in Conjecture~\ref{conj.activity}, this behavior occurs due to the nondegenerate eigenvalue spectrum of $L$. 

For $\zeta = 0$, the Laplacian matrix is Hermitian and, by Theorem~\ref{thm.mode_mixing}, mode mixing is guaranteed to occur at the global optimum (Fig.~\ref{fig.dotproduct}A). Our numerical experiments suggest this is also the case for non-Hermitian matrices, since they consistently show that the optimal solution lies at a nondifferentiable point. In line with Conjecture~\ref{conj.activity}, Fig.~\ref{fig.dotproduct} also shows that the magnitude of the relative improvement $\frac{\Lambda_{\mathrm{max}}(\bm b^*)-\Lambda_{\mathrm{max}}(\bm b_{\text{hom}}^*)}{\abs{\Lambda_{\mathrm{max}}(\bm b_{\text{hom}}^*)}}$ grows with $\Delta\lambda$, reaching values $\{-0.196, -0.445, -0.706\}$ for $\zeta\in\{0,0.45,0.9\}$, respectively.
\QEDA
\end{example}

\begin{figure}[H]
\centering
\includegraphics[width=0.75\textwidth]{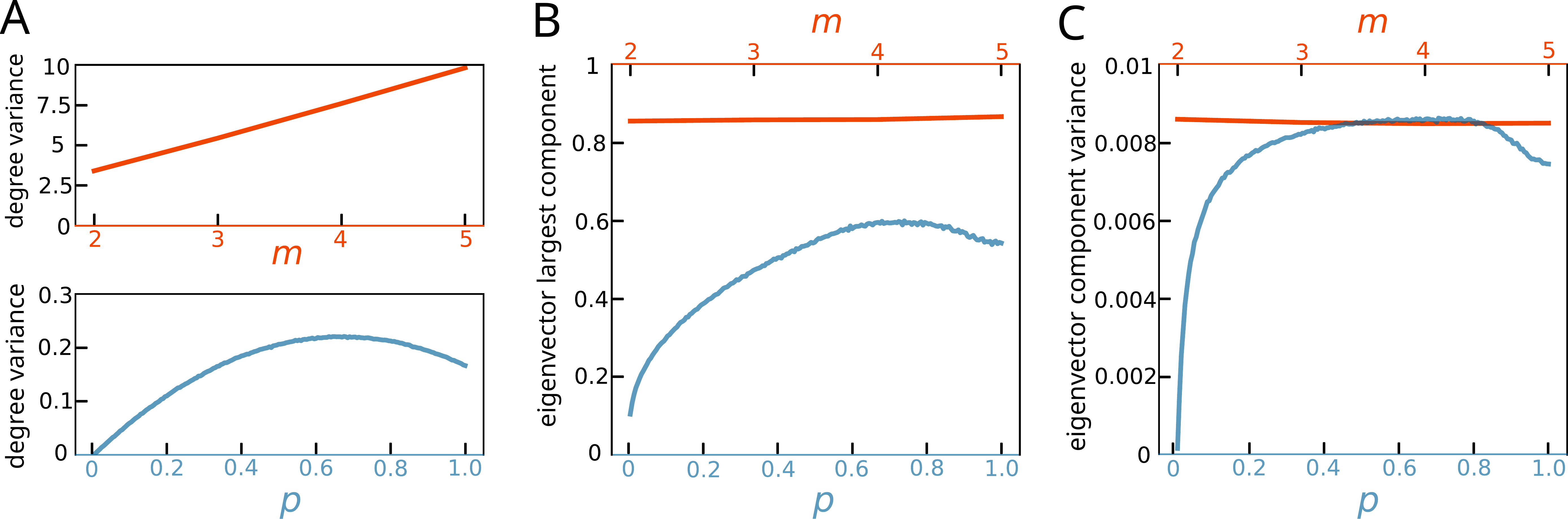}
\caption{\label{fig.sm.isotropy} 
\textbf{Eigenvector localization in degree-homogeneous and degree-heterogeneous networks.}
(\textbf{A}) Variance of node in-degrees as a function of the rewiring probability $p$ (for SW networks, in blue) and the attachment parameter $m$ (for SF networks, in red).
(\textbf{B--C}) Largest absolute value and variance of the components of the dominant Laplacian eigenvector, plotted as functions of $p$ and $m$. The dominant eigenvector $\bm v$ corresponds to the eigenvector associated with the Laplacian eigenvalue with the largest real part. All quantities are averaged over 2{,}000 independent network realizations of size $N=100$, and eigenvectors are normalized.
These results show that degree-heterogeneous networks exhibit substantially stronger eigenvector localization compared to degree-homogeneous networks.
}
\end{figure}

\section{{Disorder-promoted stability in a broader class of dynamical states}}
\label{sec.sm.spatiotemp}

{Beyond the equilibrium states analyzed in the main text, here we present numerical simulations showing that disorder-promoted stability can also be observed in more general forms of collective behavior, including periodic, chaotic, and multistable regimes. 
Example~\ref{examp.sm.lvperiodic} shows that disordered network structures can enhance stability in LV ecological systems exhibiting periodic dynamics. 
Example~\ref{examp.sm.chaossync} shows that disorder in nodal parameters can promote stable synchronization in networks of chaotic R\"ossler oscillators. 
Finally, Example~\ref{examp.sm.basin} shows that, in multistable phase-amplitude oscillator systems, improving the linear stability of the synchronous state can also enlarge its basin of attraction.}

\begin{remark}
{
    To characterize the stability of periodic and chaotic solutions in the following examples, we use Lyapunov exponents $\Lambda_i$, defined as the asymptotic exponential rate at which infinitesimal perturbations $\delta\bm x(t)$ to a trajectory $\bm x(t)$ grow or decay, as given by $\norm{\delta\bm x(t)}\sim e^{\Lambda_i t}\norm{\delta\bm x(0)}$.
    For an equilibrium point (as considered in the main text), $\Lambda_i$ reduces to the real part of an eigenvalue of the Jacobian matrix, $\Lambda_i = \Re\big(\lambda_i(J)\big)$.
    For a periodic orbit with period $T$, on the other hand, $\Lambda_i$ is  related to the Floquet multiplier $\mu_i$ by $\Lambda_i = \frac{1}{T} \log |\mu_i|$. In Figs.~\ref{fig.RE_limit_cycle} and \ref{fig.sm.rossler}, we numerically estimate the Lyapunov exponents using the standard variational-equation approach, in which infinitesimal perturbation vectors are evolved along the trajectory and periodically re-orthonormalized to extract their asymptotic exponential growth rates.
}
\end{remark}

\begin{example}[{Lotka-Volterra system with periodic dynamics}]
\label{examp.sm.lvperiodic}

Consider the LV network model presented in Eq.~\eqref{eq.lvmodel}. Under certain conditions (described in detail below), the model has a single unstable equilibrium point and a single stable limit cycle (under the constraint that they are feasible). In this regime, we can isolate how disordered parameter configurations affect the stability of periodic orbits, providing a natural extension of the equilibrium-point analysis presented in Fig.~\ref{fig.ecologicalnet}.

{
To simplify the analysis, we explore the equivalence between the LV equation with $N$ species and the \textit{replicator equation} with $N'=N+1$ strategies \cite{hofbauer1998evolutionary}, as commonly done in mathematical ecology \cite{allesina2026global}. The replicator equation is used in evolutionary game theory to describe a population choosing among $N'$ possible strategies. Let $p_i$ denote the fraction of the population choosing strategy $i$, and $M\in\R^{N'\times N'}$ be the payoff matrix. The replicator equation is given by \cite{taylor1978evolutionary}
\begin{equation}
    \dot p_i = p_i\left(\sum_{j=1}^{N'} M_{ij}p_j - \bm p^\transp M\bm p\right),\quad\text{for } i=1,\ldots,N'.
\end{equation}
Thus, strategies whose payoff exceeds the population average $\bm p^\transp M\bm p$ increase in frequency, whereas strategies with below-average payoff decrease in frequency. The mapping between the replicator equation and the LV model is obtained (up to a rescaling of time) by writing
\begin{equation}
p_{i}=\begin{cases}
\frac{1}{1+\sum_{k=1}^{N}x_{k}}, &\, \text{ if }\,i=1,\\
\frac{x_{i}}{1+\sum_{k=1}^{N}x_{k}}, &\, \text{ if }\, i=2,\ldots,N+1,
\end{cases}
\quad \text{and} \quad
M_{ij}=\begin{cases}
0, &\, \text{ if } \;i=1,\\
b_{i-1}, &\, \text{ if } \;j=1 \text{ and } \;i\neq 1,\\
A_{i-1,j-1}, &\, \text{ otherwise}.
\end{cases}
\end{equation}
For the replicator equation, it is known that if the payoff matrix has the cyclic form
\begin{equation}
M_{ij}=\begin{cases}
\alpha_0,&\quad\text{ if }\;i=j,\\
\alpha_{i},&\quad\text{ if }i=(j+1)\,\text{mod}\,N',\\
0,&\quad\text{ otherwise},
\end{cases}
\end{equation}
with $\alpha_0\leq 0$ and $\alpha_{i}\geq 0,\;\forall i\geq1$, then the system admits a stable limit cycle encircling the unstable equilibrium point $\bm p^{*}\neq 0$~\cite{hofbauer1998evolutionary} for certain values of $
\alpha_{0}$.
The corresponding growth rates and adjacency matrix in the LV model are then respectively given by $b_{i}=\alpha_{i}+\alpha_0\delta_{iN}$ and
\begin{equation}
A_{ij}=\begin{cases}
\alpha_0-\alpha_{N}\delta_{i1},\quad &\text{ if }\;i=j=1, \\
-\alpha_{N},\quad &\text{ if }\;j=i \text{ and } i>1,\\
\alpha_{i},\quad &\text{ if }\;i=(j+1)\,\text{mod}\,N \text{ and } i\neq N.\\
\end{cases}
\end{equation}
}

\begin{figure}[t]
\centering
\includegraphics[width=\textwidth]{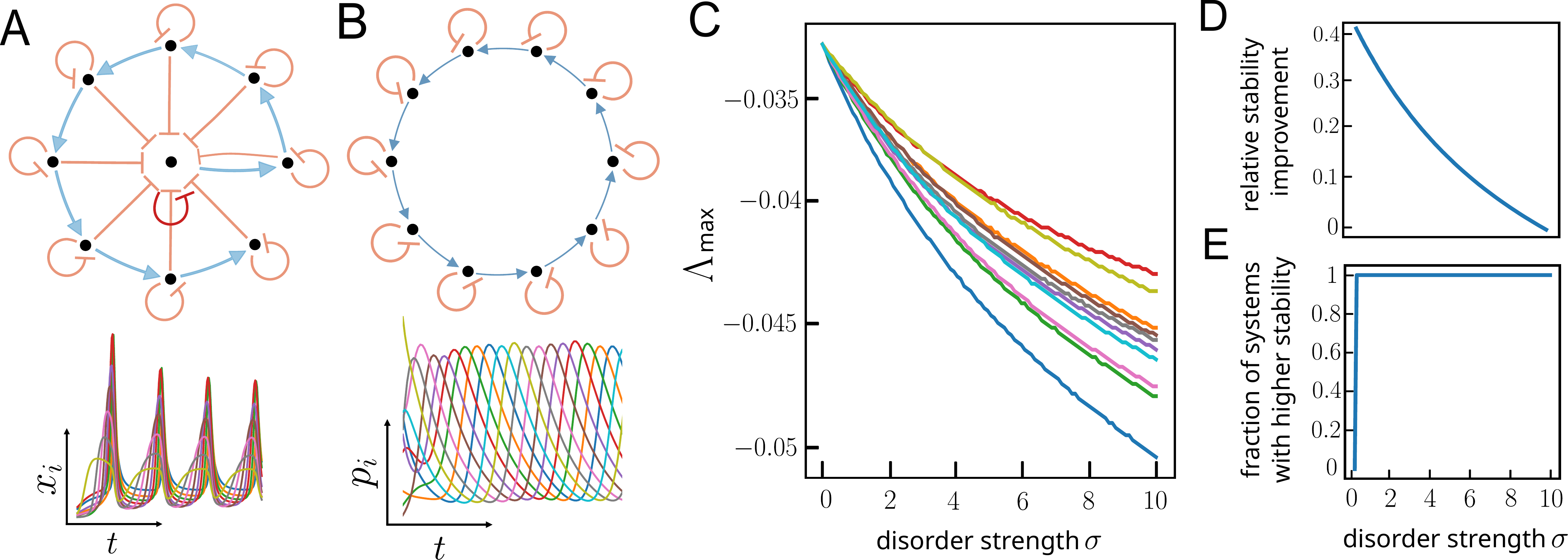}
\caption{\label{fig.RE_limit_cycle} 
{
\textbf{Disorder-promoted stability of limit cycles.}
(\textbf{A, B}) Network structures associated with the adjacency matrix $A$ of the LV model with $N=9$ species (A) and the payoff matrix $M$ of the equivalent replicator equation with $N'=10$ strategies (B). The blue and red edges represent excitatory and inhibitory interactions, respectively.
The bottom panels show the time series of the populations $x_i$ in the LV model and $p_i$ in the replicator equation.
(\textbf{C}) Largest transverse Lyapunov exponent $\Lambda_{\rm max}(\sigma)$ associated with the limit cycle for the network structure in panel B, shown as a function of the disorder strength $\sigma$. Each curve corresponds to a different realization of disorder.
(\textbf{D}) Average relative stability improvement, $(\Lambda_{\mathrm{max}}(\sigma) - \Lambda_{\mathrm{max}}(0))/|\Lambda_{\mathrm{max}}(0)|$, as a function of $\sigma$.
(\textbf{E}) Fraction of systems with disorder-promoted stability (satisfying $\Lambda_{\rm max}(\sigma)<\Lambda_{\rm (0)}$), as a function of $\sigma$.
The results in panels D and E are averaged over 100 realizations of disorder, with $\alpha_{i}\sim\bar\alpha+\mathcal{U}[0,\sigma/N]$, $\forall i$, $\bar\alpha = 1$, and $\alpha_0=-0.5$.
}
}
\end{figure}

{
Fig.~\ref{fig.RE_limit_cycle}A,B illustrates the network structures of the equivalent LV and replicator systems. In the replicator system, when all cyclic edges are identical ($\alpha_i=\bar\alpha$, $\forall i$), the underlying network $M$ is rotationally symmetric and therefore all nodes are structurally identical (even though this symmetry does not exist in the corresponding LV system). The non-Hermiticity of the payoff matrix $M$ implies that the variational dynamics along the limit cycle are governed by a non-Hermitian (time-dependent) Jacobian. Consequently, by Proposition~\ref{cor.sm.nonhermitian}, the largest transverse Lyapunov exponent $\Lambda_{\rm max}$ associated with the limit cycle can be a nonconvex function of $\bm\alpha = (\alpha_1,\ldots,\alpha_{N'})$ and, thus, by Theorem~\ref{thm.jensenineq}, the parameter configuration that maximizes stability need not coincide with the homogeneous choice $\bm\alpha=\bar\alpha\bm 1_{N'}$. Moreover, the underlying network structure $M$ is directed and circulant, indicating that this homogeneous point can be differentiable and that disorder-promoted stability is likely. To probe this effect, we introduce disorder into the edge interactions $\alpha_{i}\sim\bar\alpha+\mathcal{U}[0,\sigma/N]$, thus making the network structure asymmetric and disordered. Fig.~\ref{fig.RE_limit_cycle}C--E shows that increasing the disorder strength indeed decreases $\Lambda_{\rm max}$, enhancing the stability of the limit cycle. These results demonstrate that our framework extends beyond equilibrium points and can be used to analyze the relationship between non-Hermiticity, disorder, and the stability of periodic orbits.}
\QEDA
\end{example}

\begin{example}[{Chaos synchronization}]
\label{examp.sm.chaossync}
{
We study the problem of chaos synchronization, that is, we analyze the parameters for which all oscillators follow the same chaotic trajectory. To this end, consider an undirected ring network of $N=10$ coupled R\"ossler oscillators~\cite{rossler1976equation}:
\begin{equation} \label{eq.rossler}
\begin{aligned}
\dot{x}_{i}&=-y_{i} - z_{i},\\
\dot{y}_{i}&=x_{i} + a_{i}y_{i} + K\sum_{j=1}^N A_{ij} (y_j-y_i),\\
\dot{z}_{i}&=b + z_{i}(x_{i} - c),
\end{aligned}
\end{equation}
for $i=1,\ldots,N$, where $\bm x_i=(x_i,y_i,z_i)^\transp$ is the state of oscillator $i$, $A_{ij}$ is the adjacency matrix, $K$ is the coupling strength, $(b,c)$ are fixed parameters, and $\bm{a}=(a_1,\ldots,a_N)$ is the set of nodal parameters in which we introduce disorder. The Jacobian of the system, evaluated at the phase-synchronized state, is generally non-Hermitian. Therefore, by Proposition~\ref{cor.sm.nonhermitian}, $\Lambda_{\rm max}$ can be a nonconvex function of $\bm a$. Following our framework in Fig.~\ref{fig.modemixing}A, this implies that the stability of the synchronous state can potentially be enhanced at parameter configurations $\bm a$ that do not respect the global symmetries of the ring network $A$ and are therefore heterogeneous.
}

\begin{figure}[t]
\centering
\includegraphics[width=0.95\textwidth]{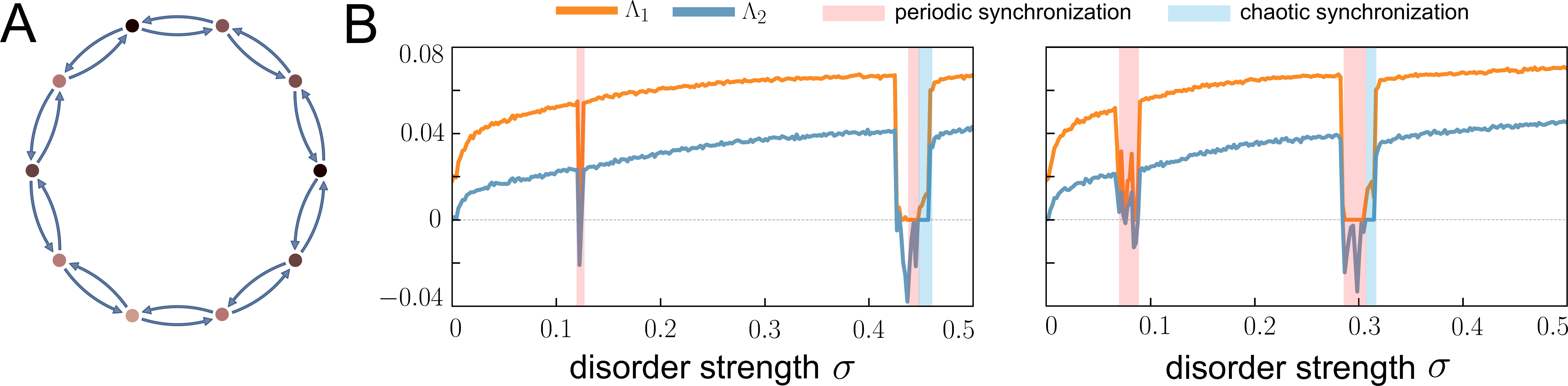}
\caption{\label{fig.sm.rossler}
{
\textbf{Disorder-promoted synchronization among chaotic oscillators.}
(\textbf{A}) Ring network of $N=10$ R\"ossler oscillators, with node colors representing the heterogeneity in parameter $a_{i}$. 
(\textbf{B}) Largest Lyapunov exponents $\Lambda_{1}$ (orange) and second largest Lyapunov exponent $\Lambda_{2}$ (blue), shown as a function of the disorder strength $\sigma$. Each panel shows a different realization of the disordered parameter, randomly drawn as $a_i\sim \bar a+\mathcal{U}[0,\sigma/N]$, $\forall i$. 
The green (purple) shade denotes intervals of $\sigma$ in which the oscillators have stable chaotic (periodic) synchronization.
The simulations are shown for $(\bar a,b,c,K)=(0.1,0.1,14,0.225)$.
}
}
\end{figure}

{
Fig.~\ref{fig.sm.rossler} shows how the two largest Lyapunov exponents, $\Lambda_1$ and $\Lambda_2$, vary as the disorder in $\bm a$ is increased, with parameters sampled as $a_i\sim \bar{a}+\mathcal{U}[0,\sigma/N]$. Synchronization among chaotic oscillators is often empirically diagnosed based on the signs of these exponents \cite{Rosenblum1997}:
\begin{enumerate}
    \item If $\Lambda_1,\Lambda_2>0$, the oscillators are desynchronized; 
    \item If $\Lambda_1>0$ and $\Lambda_2 = 0$, the oscillators are synchronized on a chaotic orbit (i.e., chaotic synchronization);
    \item If $\Lambda_1=0$ and $\Lambda_2 < 0$, the oscillators are synchronized on a periodic orbit (i.e., periodic synchronization).
\end{enumerate}
For homogeneous parameters ($\sigma=0$) and sufficiently small coupling $K$, both exponents are positive, indicating that there is no synchronization. In contrast, when disorder is introduced into $\bm a$, both periodic and chaotic synchronization emerge over finite ranges of $\sigma$. This demonstrates that parameter heterogeneity can promote synchronization in regimes where the corresponding homogeneous system remains desynchronized.
}
\QEDA
\end{example}

\begin{example}[{Multistable phase-amplitude oscillators}]
\label{examp.sm.basin}
{
Consider the following network of coupled phase-amplitude oscillators (which is a simplified version of Stuart-Landau oscillators):
\begin{equation}
    \begin{aligned}
        \dot r_i &= b_i r_i (1-r_i) + \varepsilon r_i \sum_{j=1}^N A_{ij} \cos(\phi_j-\phi_i), \\
        \dot\phi_i &= \omega + r_i - 1 + r_i \sum_{j=1}^N A_{ij} \sin(\phi_j-\phi_i),
    \end{aligned}
\end{equation}
for $i=1,\ldots,N$, where $(r_i,\phi_i)$ denote the amplitude and phase of oscillator $i$, $\varepsilon$ is the coupling strength, $\omega$ is the natural frequency, and $\bm b = (b_1,\ldots,b_N)$ is the vector of dissipation parameters into which we introduce disorder. We take the adjacency matrix $A$ to be a directed circulant network with distinct clockwise and counterclockwise edge weights, as illustrated in Fig.~\ref{fig.multistable}A.
This system is multistable because different initial conditions can converge to distinct attractors, including frequency-synchronized states and incoherent states. Here, we focus on the stability of a frequency-synchronized solution of the form $r_j(t)=r_j^{\rm eq}$ and $\phi_j(t)=\Omega t+\phi_j^{\rm eq}$, where all oscillators rotate with the same frequency $\Omega$. The equilibrium amplitudes $r_j^{\rm eq}$ and phase offsets $\phi_j^{\rm eq}$ are generally nonidentical for different oscillators when $b_i\neq b_j$. The Jacobian matrix evaluated around this frequency-synchronized state is non-Hermitian, and hence, by Proposition~\ref{cor.sm.nonhermitian},  $\Lambda_{\rm max}$ can be a nonconvex function of $\bm b$. 
}

{
Following our framework in Fig.~\ref{fig.modemixing}A, the parameter configuration that maximizes stability need not preserve the symmetries of the underlying network, and may therefore be heterogeneous. Moreover, because the network structure is directed and circulant, it is likely that stability is promoted for disordered parameter configurations of the form $b_i\sim\mathcal N(b_{\rm hom},\sigma^2)$. Fig.~\ref{fig.multistable}B shows that the largest Lyapunov exponent $\Lambda_{\rm max}$ indeed decreases over finite ranges of the disorder strength $\sigma$, indicating enhanced linear stability of the frequency-synchronized state.
Importantly, this stabilization is accompanied by an enlargement of the basin of attraction over comparable ranges of disorder, and both effects occur across a large fraction of disorder realizations (Fig.~\ref{fig.multistable}C). Thus, in this example, disorder not only enhances local linear stability but also increases the system's resilience to finite perturbations.}
\QEDA
\end{example}

\begin{figure}[H]
    \centering
    \includegraphics[width=0.85\linewidth]{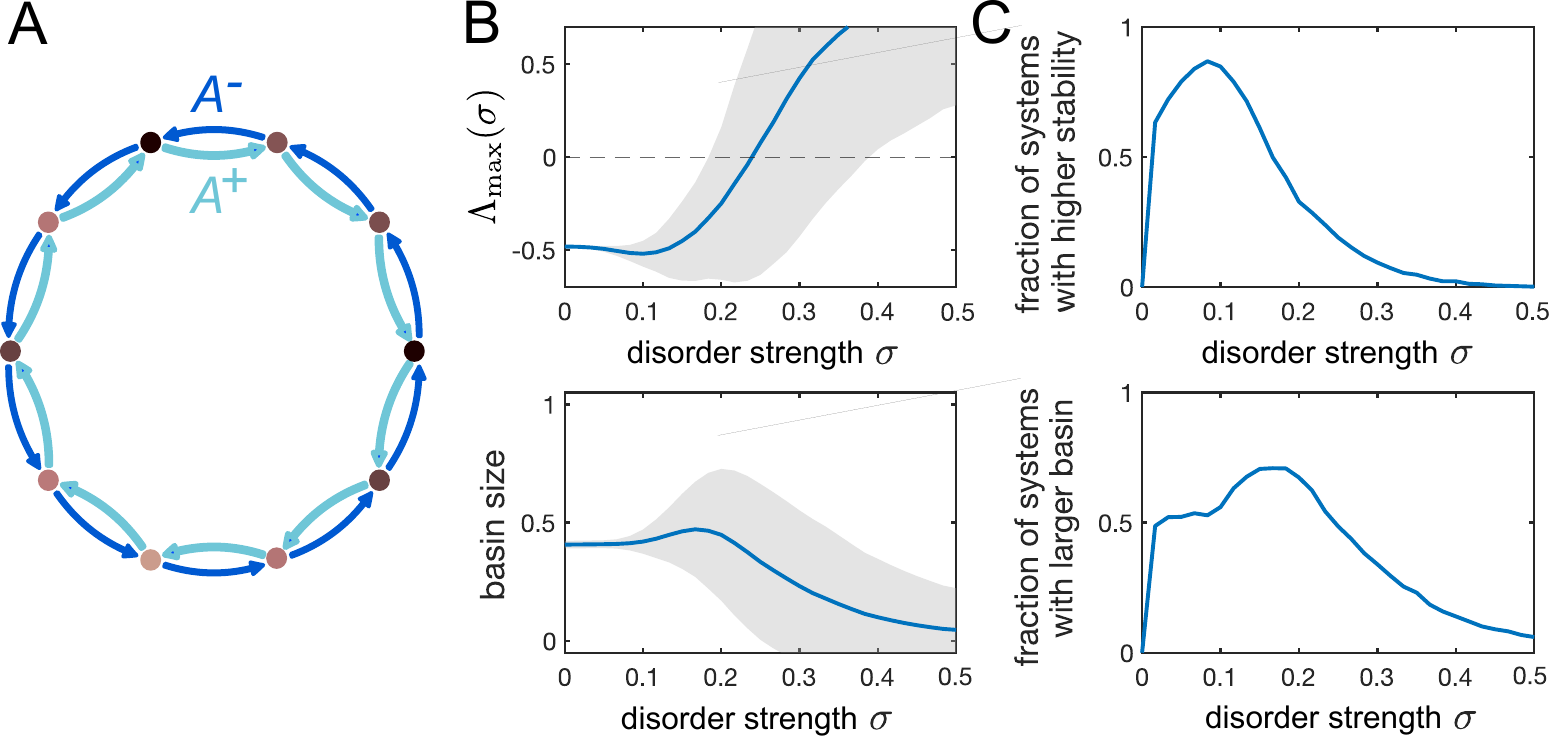}
    \caption{\label{fig.multistable}
    {
    \textbf{Disorder-promoted stability in multistable oscillator networks.}
    (\textbf{A}) Directed circulant network of $N=10$ coupled phase-amplitude oscillators, with clockwise and counterclockwise edge weights $A^+$ and $A^-$, respectively.
    (\textbf{B}) Largest Lyapunov exponent (top) and size of the basins of attraction (bottom), shown as a function of the disorder strength $\sigma$. The curves and shaded areas indicate the average and standard deviation, respectively.
    (\textbf{C}) Fraction of systems with larger stability (top) and basin sizes (bottom), shown as a function of $\sigma$.
    The results in panels B and C are computed over 1{,}000 independent disorder realizations, with $b_i\sim\mathcal N(\bar b,\sigma^2)$, $\forall i$, and $(\bar b,\omega,\varepsilon,A^-,A^+)=(0.5,0,0.3,1,2)$. The basin sizes are estimated as the fraction of 1{,}000 randomly sampled initial conditions ($r_i\sim\mathcal U[0,3]$ and $\phi_i\sim\mathcal U[-\pi,\pi]$) for which the system converges to the frequency-synchronized state.
    }
    }
\end{figure}







\end{document}